\documentclass[epsf]{iopart}
\usepackage{iopams,amsfonts,amsthm,rotating,graphicx,fancybox,fancyhdr,bm,eucal,dsfont,appendix}
\usepackage{datetime,esint,color}
\usepackage{accents}
\usepackage[normalem]{ulem}

\newcommand{\dbtilde}[1]{\accentset{\approx}{#1}}
\newcommand*\rfrac[2]{{}^{#1}\!/_{#2}}

\eqnobysec
\settimeformat{hhmmsstime}
\ddmmyyyydate

\theoremstyle{plain}
\newtheorem{theorem}{Theorem}[section]
\newtheorem{lemma}[theorem]{Lemma}
\newtheorem{proposition}[theorem]{Proposition}
\newtheorem{corollary}[theorem]{Corollary}
\theoremstyle{definition}
\newtheorem{definition}[theorem]{Definition}
\newtheorem{remark}[theorem]{Remark}
\newtheorem{conjecture}[theorem]{Conjecture}

\begin{document}

\title{Nonperturbative theory of power spectrum in complex systems}

\author{\hspace{-2cm}Roman Riser$^{1,2}$, Vladimir Al. Osipov$^{1}$, and Eugene Kanzieper$^1$\footnote{Corresponding author.}}

\address{\hspace{-2cm}$^1$~Department of Mathematics, Holon Institute of Technology, Holon 5810201, Israel\\
         \hspace{-2cm}$^2$~Department of Mathematics and Research Center for Theoretical Physics\\
         \hspace{-1.8cm} and Astrophysics, University of Haifa, Haifa 3498838, Israel
        }
\noindent\newline\newline
\begin{abstract}
The power spectrum analysis of spectral fluctuations in complex wave and quantum systems has emerged as a useful tool for studying their internal dynamics. In this paper, we formulate a nonperturbative theory of the power spectrum for complex systems whose eigenspectra -- not necessarily of the random-matrix-theory (RMT) type -- possess stationary level spacings. Motivated by potential applications in quantum chaology, we apply our formalism to calculate the power spectrum in a tuned circular ensemble of random $N \times N$ unitary matrices. In the limit of infinite-dimensional matrices, the exact solution produces a universal, parameter-free formula for the power spectrum, expressed in terms of a fifth Painlev\'e transcendent. The prediction is expected to hold universally, at not too low frequencies, for a variety of quantum systems with completely chaotic classical dynamics and broken time-reversal symmetry. On the mathematical side, our study brings forward a conjecture for a double integral identity involving a fifth Painlev\'e transcendent.

\noindent\newline\newline
Published in:~Annals of Physics {\bf 413}, 168065 (2020).
\end{abstract}
\newpage
\tableofcontents

\newpage
\section{Introduction} \label{intro}

The power spectrum analysis of stochastic spectra \cite{GKKMRR-2011} had recently emerged as a powerful tool for studying both system-specific and universal properties of complex wave and quantum systems. In the context of quantum systems, it reveals whether the corresponding classical dynamics is regular or chaotic, or a mixture of both, and encodes a `degree of chaoticity'. In combination with other long- and short-range spectral fluctuation measures, it provides an effective way to identify system symmetries, determine a degree of incompleteness of experimentally measured spectra, and get the clues about systems’ internal dynamics. Yet, the theoretical foundations of the power spectrum analysis of stochastic spectra have not been settled. In this paper, a nonperturbative theory of the power spectrum will be presented.

To set the stage, we review traditional spectral fluctuation measures (Section~\ref{backgro}), define the power spectrum (Definition~\ref{def-01}) and briefly discuss its early theoretical and numerical studies as well as the recently reported experimental results (Section~\ref{powerspectrum}). We then argue (Section~\ref{uncor}), that a form-factor approximation routinely used for the power spectrum analysis in quantum chaotic systems is not flawless and needs to be revisited.

\subsection{Short- and long-range measures of spectral fluctuations} \label{backgro}

Spectral fluctuations of quantum systems reflect the nature -- regular or chaotic -- of their underlying classical dynamics \cite{B-1987,BGS-1984,R-2000}. In case of fully chaotic classical dynamics, {\it hyperbolicity} (exponential sensitivity to initial conditions) and {\it ergodicity} (typical classical trajectories fill out available phase space uniformly) make quantum properties of chaotic systems universal \cite{BGS-1984}. At sufficiently long times $t > T_*$, the single particle dynamics is governed by global symmetries of the system and is accurately described by the random matrix theory (RMT) \cite{M-2004,PF-book}. The emergence of universal statistical laws, anticipated by Bohigas, Giannoni and Schmit \cite{BGS-1984}, has been advocated within a field-theoretic \cite{AASA-1996, AAS-1997} and a semiclassical approach \cite{RS-2002} which links correlations in quantum spectra to correlations between periodic orbits in the associated classical geodesics. The time scale $T_*$ of compromised spectral universality is set by the period $T_1$ of the shortest closed orbit and the Heisenberg time $T_{\rm H}$, such that $T_1 \ll T_* \ll T_{\rm H}$.

Several statistical measures of level fluctuations have been devised in quantum chaology. {\it Long-range} correlations of eigenlevels on the unfolded energy scale \cite{M-2004} can be measured by the variance $\Sigma^2(L)={\rm var}[{\mathcal N}(L)]$ of number of levels ${\mathcal N}(L)$ in the interval of length $L$. The $\Sigma^2(L)$ statistics probes the two-level correlations only and exhibit \cite{B-1985} a universal RMT behavior provided the interval $L$ is not too long, $1 \ll L \ll T_{\rm H}/T_1$. The logarithmic behavior of the number variance,
\begin{eqnarray} \label{nv}
    \Sigma^2_{\rm chaos}(L) = \frac{2}{\pi^2\beta} \ln L + {\mathcal O}(1),
\end{eqnarray}
indicates presence of the long-range repulsion between eigenlevels. Here, $\beta=1,2$ and $4$ denote the Dyson symmetry index \cite{M-2004,PF-book}. For more distant levels, $L \gg T_{\rm H}/T_1$, system-specific features show up in $\Sigma^2_{\rm chaos}(L)$ in the form of quasi-random oscillations with wavelengths being inversely proportional to periods of short closed orbits.

Individual features of quantum chaotic systems become less pronounced in spectral measures that probe the {\it short-range} fluctuations as these are largely determined by the long periodic orbits \cite{RS-2002}. The distribution of level spacing between (unfolded) consecutive eigenlevels, $P(s) = \langle \delta(s - E_j + E_{j+1})\rangle$, is the most commonly used short-range statistics. Here, the angular brackets denote averaging over the position $j$ of the reference eigenlevel or, more generally, averaging over such a narrow energy window that keeps the classical dynamics essentially intact. At small spacings, $s \ll 1$, the distribution of level spacings is mostly contributed by the {\it two-point} correlations, showing the phenomenon of symmetry-driven level repulsion, $P(s) \propto s^\beta$. (In a simple-minded fashion, this result can be read out from the Wigner surmise \cite{M-2004}). As $s$ grows, the spacing distribution becomes increasingly influenced by spectral correlation functions of {\it all} orders. In the universal regime ($s \lesssim T_{\rm H}/T_*$), these are best accounted for by the RMT machinery which produces parameter-free (but $\beta$-dependent) representations of level spacing distributions in terms of Fredholm determinants/Pfaffians and Painlev\'e transcendents. For quantum chaotic systems with broken time-reversal symmetry $(\beta=2)$ -- that will be the focus of our study -- the level spacing distribution is given by the famous Gaudin-Mehta formula, which when written in terms of Painlev\'e transcendents reads \cite{PF-book,JMMS-1980}
\begin{eqnarray}\label{LSD-PV}
    P_{\rm chaos}(s) = \frac{d^2}{ds^2} \exp \left(
            \int_{0}^{2\pi s} \frac{\sigma_0(t)}{t} dt
        \right).
\end{eqnarray}
Here, $\sigma_0(t)$ is the fifth Painlev\'e transcendent defined as the solution to the nonlinear equation $(\nu=0)$
\begin{equation} \label{PV-family}
    (t \sigma_\nu^{\prime\prime})^2 + (t\sigma_\nu^\prime -\sigma_\nu) \left(
            t\sigma_\nu^\prime -\sigma_\nu + 4 (\sigma_\nu^\prime)^2
        \right) - 4 \nu^2 (\sigma_\nu^\prime)^2 = 0
\end{equation}
subject to the boundary condition $\sigma_0(t) = -t/2\pi -(t/2\pi)^2 + o(t^2)$ as $t\rightarrow~0$.

The universal RMT laws [Eqs.~(\ref{nv}) and (\ref{LSD-PV})] apply to quantum systems with {\it completely chaotic} classical dynamics. Quantum systems whose classical geodesics is {\it completely integrable} belong to a different, Berry-Tabor universality class \cite{BT-1977}, partially shared by the Poisson point process. In particular, level spacings in a generic integrable quantum system exhibit statistics of waiting times between consecutive events in a Poisson process. This leads to the radically different fluctuation laws: the number variance $\Sigma^2_{\rm int}(L)=L$ is no longer logarithmic while the level spacing distribution $P_{\rm int}(s) = e^{-s}$ becomes exponential \cite{B-1987}, with no signatures of level repulsion whatsoever. Such a selectivity of short- and long-range spectral statistical measures has long been used to uncover underlying classical dynamics of quantum systems. (For a large class of quantum systems with mixed regular-chaotic classical dynamics, the reader is referred to Refs. \cite{T-1989, TU-1994, BKP-1998}.)

\subsection{Power spectrum: Definition and early results} \label{powerspectrum}
To obtain a more accurate characterization of the quantum chaos, it is advantageous to use spectral statistics which probe the correlations between {\it both} nearby and distant
eigenlevels. Such a statistical indicator -- {\it the power spectrum} -- has been suggested in Ref.~\cite{RGMRF-2002}.

\begin{definition}\label{def-01}
  Let $\{\varepsilon_1 \le \dots \le \varepsilon_N\}$ be a sequence of ordered unfolded eigenlevels, $N \in {\mathbb N}$, with the mean level spacing $\Delta$ and let
  $\langle \delta\varepsilon_\ell \delta\varepsilon_m \rangle$ be the covariance matrix of level displacements $\delta\varepsilon_\ell = \varepsilon_\ell - \langle \varepsilon_\ell\rangle$ from their mean $\langle \varepsilon_\ell\rangle$. A Fourier transform of the covariance matrix
\begin{eqnarray}\label{ps-def}
    S_N(\omega) = \frac{1}{N\Delta^2} \sum_{\ell=1}^N \sum_{m=1}^N \langle \delta\varepsilon_\ell \delta\varepsilon_m \rangle\, e^{i\omega (\ell-m)}, \quad \omega \in {\mathbb R}
\end{eqnarray}
is called the power spectrum of the sequence. Here, the angular brackets stand for an average over an ensemble of eigenlevel sequences.
\hfill $\blacksquare$
\end{definition}

Since the power spectrum is $2\pi$-periodic, real and even function in $\omega$,
\begin{eqnarray}\fl \qquad
S_N(\omega+2\pi) = S_N(\omega), \quad S_N^*(\omega) = S_N(\omega), \quad S_N(-\omega) = S_N(\omega),
\end{eqnarray}
it is sufficient to consider it in the interval $0 \le \omega \le \omega_{\rm Ny}$, where $\omega_{\rm Ny} = \pi$ is the Nyquist frequency. In the spirit of the discrete Fourier analysis, one may restrict dimensionless frequencies $\omega$ in Eq.~(\ref{ps-def}) to a finite set
\begin{eqnarray}\label{freq-k}
    \omega_k = \frac{2\pi k}{N}
\end{eqnarray}
with $k=\{1,2,\dots, N/2\}$, where $N$ is assumed to be an even integer. We shall see that resulting analytical expressions for $S_N(\omega_k)$ are slightly simpler than those for $S_N(\omega)$.
\noindent\par
\begin{remark}
We notice in passing that similar statistics has previously been used by Odlyzko \cite{Od-1987} who analyzed power spectrum of the {\it spacings} between zeros of the Riemann zeta function.
\hfill $\blacksquare$
\end{remark}

Considering Eq.~(\ref{ps-def}) through the prism of a semiclassical approach, one readily realizes that, at low frequencies $\omega \ll T_*/T_{\rm H}$, the power spectrum is largely affected by system-specific correlations between very distant eigenlevels (accounted for by short periodic orbits). For higher frequencies, $\omega \gtrsim T_*/T_{\rm H}$, the contribution of longer periodic orbits becomes increasingly important and the power spectrum enters the {\it universal regime}. Eventually, in the frequency domain $T_*/T_{\rm H} \ll \omega \le \omega_{\rm Ny}$, long periodic orbits win over and the power spectrum gets shaped by correlations between the nearby levels. Hence, tuning the frequency $\omega$ in $S_N(\omega)$ one may attend to spectral correlation between either adjacent or distant eigenlevels.

Numerical simulations \cite{RGMRF-2002} have revealed that the average power spectrum $S_N(\omega_k)$ discriminates sharply between quantum systems with chaotic and integrable classical dynamics. While this was not completely unexpected, another finding of Ref.~\cite{RGMRF-2002} came as quite a surprise: numerical data for $S_N(\omega_k)$, at not too high frequencies, could be fitted by simple power-law curves, $S_N(\omega_k) \sim 1/\omega_k$ and $S_N(\omega_k) \sim 1/\omega_k^2$, for quantum systems with chaotic and integrable classical dynamics, respectively. In quantum systems with mixed classical dynamics, numerical evidence was presented \cite{GRRFSVR-2005} for the power-law of the form $S_N(\omega_k) \sim 1/\omega_k^\alpha$ with the exponent $1 <\alpha < 2$ measuring a `degree of chaoticity'. The power spectrum of interface fluctuations in various growth models belonging to the $(1+1)$-dimensional Kardar-Parisi-Zhang universality class, studied in Ref.~\cite{KAT-2017} both numerically and experimentally, were found to follow the power law with $\alpha= 5/3$. The power spectrum was also measured in Sinai \cite{FKMMRR-2006} and perturbed rectangular \cite{BYBLDS-2016b} microwave billiards, microwave networks \cite{BYBLDS-2016,DYBBLS-2017} and three-dimensional microwave cavities \cite{LBYBS-2018}. For the power spectrum analysis of Fano-Feshbach resonances in an ultracold gas of Erbium atoms \cite{FMAFBMPK-2014}, the reader is referred to Ref.~\cite{PM-2015}.

For quantum chaotic systems, the universal $1/\omega_k$ law for the average power spectrum in the frequency domain $T_*/T_{\rm H}\lesssim \omega_k \ll 1$ can be read out from the existing RMT literature. Indeed, defining a set of discrete Fourier coefficients
\begin{eqnarray} \label{FK-def}
    a_k = \frac{1}{\sqrt{N}} \sum_{\ell=1}^N \delta \varepsilon_\ell \, e^{i \omega_k \ell}
\end{eqnarray}
of level displacements $\{\delta\varepsilon_\ell\}$, one observes the relation
\begin{eqnarray} \label{ak}
    S_N(\omega_k)={\rm var}[a_k].
\end{eqnarray}
Statistics of the Fourier coefficients $\{a_k\}$ were studied in some detail \cite{W-1987} within the Dyson's Brownian motion model \cite{D-1962}. In particular, it is known that, in the limit $k \ll N$, they are independent Gaussian distributed random variables with zero mean and the variance ${\rm var}[a_k] = N/(2\pi^2 \beta k)$. This immediately implies
\begin{eqnarray} \label{BM}
    S_N(\omega_k \ll 1) \approx  \frac{1}{\pi \beta \omega_k}
\end{eqnarray}
in concert with numerical findings. For larger $k$ (in particular, for $k \sim N$), fluctuation properties of the Fourier coefficients $\{a_k\}$ are unknown. In view of the relation Eq.~(\ref{ak}), a nonperturbative theory of the power spectrum to be developed in this paper sets up a well-defined framework for addressing statistical properties of discrete Fourier coefficients $\{a_k\}$ introduced in Ref.~\cite{W-1987}.

An attempt to determine $S_N(\omega_k)$ for higher frequencies up to $\omega_k = \omega_{\rm Ny}$ was undertaken in Ref.~\cite{FGMMRR-2004} whose authors claimed to express the large-$N$ power spectrum in the entire domain $T_*/T_{\rm H}\lesssim \omega_k \le \omega_{\rm Ny}$ in terms of the {\it spectral form-factor} \cite{M-2004}
\begin{eqnarray}\fl \label{FF-def}
    \quad
    K_N(\tau) = \frac{1}{N}
    \left( \Big<
        \sum_{\ell=1}^N \sum_{m=1}^N e^{2 i \pi \tau (\varepsilon_\ell - \varepsilon_m) }
    \Big> -  \Big<
        \sum_{\ell=1}^N e^{2 i \pi \tau \varepsilon_\ell }
    \Big>
    \Big<
        \sum_{m=1}^N e^{-2 i \pi \tau \varepsilon_m }
    \Big> \right)
\end{eqnarray}
of a quantum system, $\tau \ge 0$. Referring interested reader to Eqs.~(3), (8) and (10) of the original paper Ref.~\cite{FGMMRR-2004}, here we only quote a small-$\omega_k$ reduction of their result:
\begin{eqnarray} \label{FFA}
    \hat{S}_N(\omega_k \ll 1) \approx \frac{1}{\omega_k^2} K_N\left(
    \frac{\omega_k}{2\pi}
    \right).
\end{eqnarray}
(Here, the hat-symbol ($\;\hat{ }\;$) is used to indicate that the power spectrum $\hat{S}_N(\omega_k \ll 1)$ is the one furnished by the form-factor approximation.) A similar approach was also used in subsequent papers \cite{GG-2006,RMRFM-2008}.

Even though numerical simulations seemed to confirm a theoretical curve derived in Ref.~\cite{FGMMRR-2004}, we believe that the status of their heuristic approach needs to be clarified. This will be done in Section~\ref{uncor}.

\subsection{Spectra with uncorrelated spacings: Form-factor vs power spectrum} \label{uncor}

A simple mathematical model of eigenlevel sequences $\{\varepsilon_1,\dots,\varepsilon_N\}$ with identically distributed, {\it uncorrelated} spacings $\{s_1,\dots,s_N\}$, where $\ell$-th ordered eigenlevel equals
\begin{eqnarray} \label{ELSJ}
\varepsilon_\ell = \sum_{j=1}^\ell s_j,
\end{eqnarray}
provides an excellent playing ground to analyze validity of the form-factor approximation. Defined by the covariance matrix of spacings of the form ${\rm cov}(s_i,s_j) = \sigma^2 \delta_{ij}$, such that $\langle s_i \rangle =1$, it allows us to determine exactly {\it both} the power spectrum Eq.~(\ref{ps-def}) and the form-factor Eq.~(\ref{FF-def}).
\noindent\newline\newline
{\it Power spectrum.}---Indeed, realizing that the covariance matrix of ordered eigenlevels equals
\begin{eqnarray}
    \langle \delta\varepsilon_\ell \delta\varepsilon_m \rangle = \sigma^2 {\rm min}(\ell, m),
\end{eqnarray}
we derive an {\it exact} expression for the power spectrum ($N \in {\mathbb N}$)
\begin{eqnarray}\fl \qquad\quad \label{S_exp}
    S_N(\omega) =  \frac{2N+1}{4N} \frac{\sigma^2}{\sin^2(\omega/2)} \left(
        1 - \frac{1}{2N+1} \frac{\sin\left((N+1/2)\omega\right)}{\sin(\omega/2)}
    \right).
\end{eqnarray}
Equation~(\ref{S_exp}) stays valid in the entire region of frequencies $0 \le \omega \le \pi$. For a set of discrete frequencies $\omega_k = 2\pi k/N$, it reduces to
\begin{eqnarray}\label{smwk}
    S_N(\omega_k) = \frac{\sigma^2}{2\sin^2(\omega_k/2)}, \quad
    0 < \omega_k \le \pi.
\end{eqnarray}

\begin{remark}\label{rem-univer}
Notice that Eqs.~(\ref{S_exp}) and (\ref{smwk}) for the power spectrum of eigenlevel sequences with uncorrelated level spacings hold {\it universally}.
Indeed, both expressions appear to be independent of a particular choice of the level spacings distribution; the level spacing variance $\sigma^2$ is the only
model-specific parameter.
\hfill $\blacksquare$
\end{remark}
\begin{figure}
\includegraphics[width=\textwidth]{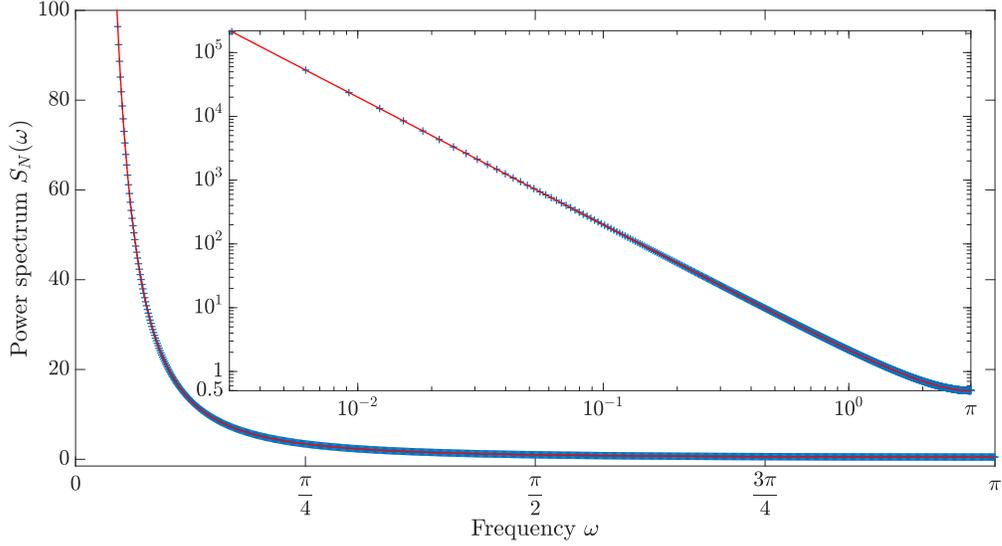}
\caption{Power spectrum $S_N(\omega)$ as a function of frequency $\omega$ for eigenlevel sequences with uncorrelated level spacings. Solid red line corresponds to the theoretical curve Eq.~(\ref{S_exp}) with $\sigma^2=1$. Blue crosses represent the average power spectrum simulated for $10$ million sequences of $N=2048$ random eigenlevels with uncorrelated, exponentially distributed spacings $s_i \sim {\rm Exp}(1)$. Inset: a log-log plot for the same graphs.}
\label{Figure_ps_exp}
\end{figure}

For illustration purposes, in Fig.~\ref{Figure_ps_exp}, we compare the theoretical power spectrum $S_N(\omega)$, Eq.~(\ref{S_exp}), with the average power spectrum {\it simulated} for an ensemble of sequences of random eigenlevels with uncorrelated, exponentially distributed level spacings $s_i \sim {\rm Exp}(1)$. Since the unit mean level spacing $\langle s_j \rangle =1$ is intrinsic to the model, the unfolding procedure is redundant. Perfect agreement between the theoretical and the simulated curves is clearly observed in the entire frequency domain $0<\omega\le \pi$.

For further reference, we need to identify three scaling limits of $S_N(\omega)$ that emerge as $N\rightarrow \infty$. In doing so, the power spectrum will be multiplied by $\omega^2$ to get rid of the singularity at $\omega=0$.

(i) The first -- infrared -- regime, refers to extremely small frequencies, $\omega \sim N^{-1}$. It is described by the double scaling limit
\begin{eqnarray} \label{SN-1st}
    {\mathcal S}^{\rm{(-1)}}(\Omega) = \lim_{N\rightarrow \infty} \omega^2 S_N(\omega)\Big|_{\omega=\Omega/N} = 2\sigma^2 \left(
    1- \frac{\sin \Omega}{\Omega}
\right),
\end{eqnarray}
where $\Omega={\mathcal O}(N^0)$. One observes:
\begin{eqnarray} \label{S1T}
    {\mathcal S}^{(-1)}(\Omega) = \left\{
                             \begin{array}{ll}
                               {\mathcal O}(\Omega^2), & \hbox{$\Omega\rightarrow 0$;} \\
                                2\sigma^2 + o(1), & \hbox{$\Omega\rightarrow \infty$.}
                             \end{array}
                           \right.
\end{eqnarray}

(ii) The second scaling regime describes the power spectrum for intermediately small frequencies $\omega \sim N^{-\alpha}$ with $0<\alpha<1$. In this case, a double
scaling limit becomes trivial:
\begin{eqnarray} \label{SN-2nd}
    {\mathcal S}^{\rm{(-\alpha)}}(\tilde\Omega) = \lim_{N\rightarrow \infty} \omega^2 S_N(\omega)\Big|_{\omega=\tilde{\Omega}/N^\alpha} = 2\sigma^2,
\end{eqnarray}
where $\tilde\Omega={\mathcal O}(N^0)$. In the forthcoming discussion of a spectral form-factor [Eq.~(\ref{K-interm})], such a scaling limit will appear with $\alpha=1/2$.

(iii) The third scaling regime describes the power spectrum for $\omega = {\mathcal O}(N^0)$ fixed as $N \rightarrow \infty$. In this case, we derive
\begin{eqnarray} \label{SN-3rd}
    {\mathcal S}^{\rm{(0)}}(\omega) = \lim_{N\rightarrow \infty} \omega^2 S_N(\omega) = \sigma^2\frac{\omega^2}{2\sin^2(\omega/2)},
\end{eqnarray}
where $\omega={\mathcal O}(N^0)$. One observes:
\begin{eqnarray} \label{S3T}
    {\mathcal S}^{(0)}(\omega) = \left\{
                             \begin{array}{ll}
                               2\sigma^2+{\mathcal O}(\omega^2), & \hbox{$\omega \rightarrow 0$;} \\
                                \sigma^2\pi^2/2, & \hbox{$\omega = \pi$.}
                             \end{array}
                           \right.
\end{eqnarray}

Equations~(\ref{S1T}), (\ref{SN-2nd}) and (\ref{S3T}) imply continuity of ${\mathcal S}(\omega)$ across the three scaling regimes. We shall return to the universal formulae Eqs.~(\ref{SN-1st}), (\ref{SN-2nd}) and (\ref{SN-3rd}) later on.
\newline\newline{\it Spectral form-factor.}---For eigenlevel sequences with identically distributed, uncorrelated level spacings, the form-factor $K_N(\tau)$ defined by Eq.~(\ref{FF-def}) can be calculated exactly, too. Defining the characteristic function of $i$-th level spacing,
\begin{eqnarray}
    \Psi_s(\tau) = \langle e^{2i\pi \tau \, s_i} \rangle = \int_{0}^{\infty} ds \, e^{2i\pi \tau s} f_{s_i}(s),
\end{eqnarray}
where $f_{s_i}(s)$ is the probability density of the $i$-th level spacing, we reduce Eq.~(\ref{FF-def}) to
\begin{eqnarray}\fl \label{K-tau-theor}
    K_N(\tau) = 1 + \frac{2}{N}{\rm Re\,} \left[
    \frac{\Psi_s(\tau)}{1-\Psi_s(\tau)} \left(
        N - \frac{1- \Psi_s^N(\tau)}{1- \Psi_s(\tau)}
    \right)
    \right] - \frac{1}{N} \left|
    \Psi_s(\tau) \frac{1- \Psi_s^N(\tau)}{1- \Psi_s(\tau)}
    \right|^2.\nonumber\\
    {}
\end{eqnarray}

In Fig.~\ref{Figure_KT_exp}, we compare the theoretical form-factor Eq.~(\ref{K-tau-theor}) with the average form-factor simulated for an ensemble of
sequences of random eigenlevels with uncorrelated, exponentially distributed level spacings as explained below Remark~\ref{rem-univer}. The simulation was based on Eqs.~(\ref{FF-def}) and (\ref{ELSJ}), and involved averaging \cite{P-1997} over ten million realizations. Referring the reader to a figure caption for detailed explanations, we plainly notice a perfect agreement between the simulations and the theoretical result Eq.~(\ref{K-tau-theor}).

\begin{figure}
\includegraphics[width=\textwidth]{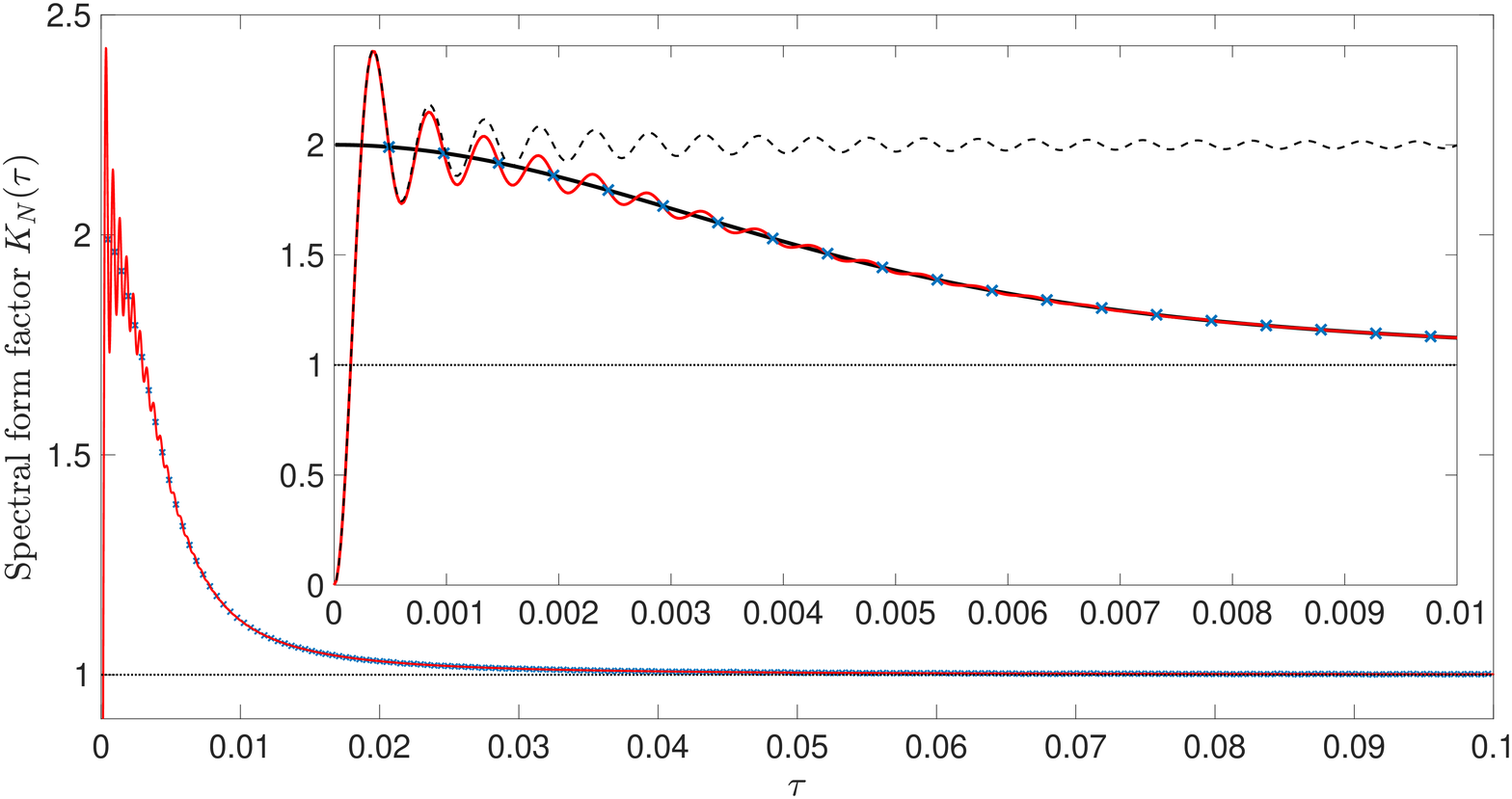}
\caption{Spectral form-factor $K_N(\tau)$ as a function of $\tau$ for a model and the data specified in the caption to Fig.~\ref{Figure_ps_exp}. Solid red line corresponds to the theoretical curve Eq.~(\ref{K-tau-theor}) with $\Psi_s(\tau) = (1-2 i \pi \tau)^{-1}$. Inset: a close-up view of the same graphs; additional black curves display limiting form-factor in various scaling regimes. Dashed line: regime ${\rm (I)}$, Eq.~(\ref{K-infrared}) with $\tau=T/N$. Solid line: regime ${\rm (II)}$, Eq.~(\ref{K-interm}) with $\tau={\mathcal T}/N^{1/2}$. Dotted line:
regime ${\rm (III)}$, Eq.~(\ref{K-tau-3}), see discussion there. Notice that the black dashed curve ${\rm [(I)]}$ starts to deviate from the red curve (after the fourth blue cross the deviation exceeds $10\%$; as $\tau$ grows further, the relative deviation approaches the factor $2$). For larger $\tau$, the black solid curve ${\rm [(II)]}$ becomes a better fit to the red curve. Finally, the red curve approaches the unity depicted by the black dotted line ${\rm [(III)]}$.}
\label{Figure_KT_exp}
\end{figure}

As $N\rightarrow \infty$, three different scaling regimes can be identified for the spectral form-factor. {\it Two} of them, arising in specific {\it double scaling} limits, appear to be {\it universal}.

(i) The first -- infrared -- regime, refers to extremely short times, $\tau \sim N^{-1}$. Assuming existence and convergence of the moment-expansion for the characteristic function $\Psi_s(\tau)$, we expand it up to the terms of order $N^{-2}$,
\begin{eqnarray}
    \Psi_s(\tau)\Big|_{\tau=T/N} = 1 + 2i\pi \frac{T}{N} - 2\pi^2 (\sigma^2+1) \frac{T^2}{N^2} + \mathcal{O}(N^{-3})
\end{eqnarray}
to derive the infrared double scaling limit for the form factor:
\begin{eqnarray} \label{K-infrared}
    K^{(-1)}(T)=\lim_{N\rightarrow \infty} K_N(\tau)\Big|_{\tau=T/N} = 2\sigma^2 \left(
        1 - \frac{\sin (2\pi T)}{2 \pi T}
    \right),
\end{eqnarray}
where $T = {\mathcal O}(N^0)$. Notice that this formula holds {\it universally} as $K^{(-1)}(T)$ does not
depend on a particular choice of the level spacings distribution; its variance $\sigma^2$ is the only model-specific parameter. One observes:
\begin{eqnarray} \label{K1T}
    K^{(-1)}(T) = \left\{
                             \begin{array}{ll}
                               {\mathcal O}(T^2), & \hbox{$T\rightarrow 0$;} \\
                                2\sigma^2 + o(1), & \hbox{$T\rightarrow \infty$.}
                             \end{array}
                           \right.
\end{eqnarray}

(ii) The second -- intermediate -- regime, refers to intermediately short times, $\tau \sim N^{-1/2}$. Expanding the
characteristic function $\Psi_s(\tau)$ up to the terms of order $N^{-1}$,
\begin{eqnarray}
    \Psi_s(\tau)\Big|_{\tau={\mathcal T}/N^{1/2}} = 1 + 2i\pi \frac{\mathcal T}{N^{1/2}} - 2\pi^2 (\sigma^2+1) \frac{{\mathcal T}^2}{N} + \mathcal{O}(N^{-3/2}),
\end{eqnarray}
we discover that, for intermediately short times, the double scaling limit of the form factor
reads
\begin{eqnarray} \label{K-interm}
    \fl \qquad\qquad
    K^{(-1/2)}({\mathcal T})=\lim_{N\rightarrow \infty} K_N(\tau)\Big|_{\tau={\mathcal T}/N^{1/2}} = \sigma^2 \left( 1
    + \frac{1 - e^{-4 \pi^2 \sigma^2 {\mathcal T}^2}}{4 \pi^2 \sigma^2 {\mathcal T}^2}\right),
\end{eqnarray}
where ${\mathcal T} = {\mathcal O}(N^0)$. Hence, in the intermediate double scaling limit, the form-factor exhibits the {\it universal} behavior too, as $K^{(-1/2)}(\mathcal{T})$ depends on a particular choice of the level spacings distribution only through its variance $\sigma^2$. One observes:
\begin{eqnarray} \label{K2T}
    K^{(-1/2)}(\mathcal{T}) = \left\{
                             \begin{array}{ll}
                               2 \sigma^2+{\mathcal O}({\mathcal T}^2), & \hbox{${\mathcal T}\rightarrow 0$;} \\
                                \sigma^2 + o(1), & \hbox{${\mathcal T}\rightarrow \infty$.}
                             \end{array}
                           \right.
\end{eqnarray}

(iii) The third scaling regime describes the form-factor for $\tau = {\mathcal O}(N^0)$ fixed as $N \rightarrow \infty$. Spotting that in this case the characteristic function $\Psi_s^N(\tau)$ vanishes exponentially fast, we derive
\begin{eqnarray} \label{K-tau-3}
    K^{(0)}(\tau) = \lim_{N\rightarrow \infty} K_N(\tau) =  1 + 2 {\rm Re\,} \left[
    \frac{\Psi_s(\tau)}{1-\Psi_s(\tau)}\right].
\end{eqnarray}
Notably, in the fixed-$\tau$ scaling limit, the form-factor is {\it no longer universal} as it depends explicitly on the particular distribution of level spacings \footnote{For the exponential distribution of level spacings the form-factor in the third scaling regime equals unity, $K^{(0)}(\tau) \equiv 1$.} through its characteristic function $\Psi_s(\tau)$. One observes:
\begin{eqnarray} \label{K3T}
    K^{(0)}(\tau) = \left\{
                             \begin{array}{ll}
                               \sigma^2+{\mathcal O}(\tau^2), & \hbox{$\tau \rightarrow 0$;} \\
                                1 + o(1), & \hbox{$\tau\rightarrow \infty$.}
                             \end{array}
                           \right.
\end{eqnarray}
\begin{figure}
\includegraphics[width=\textwidth]{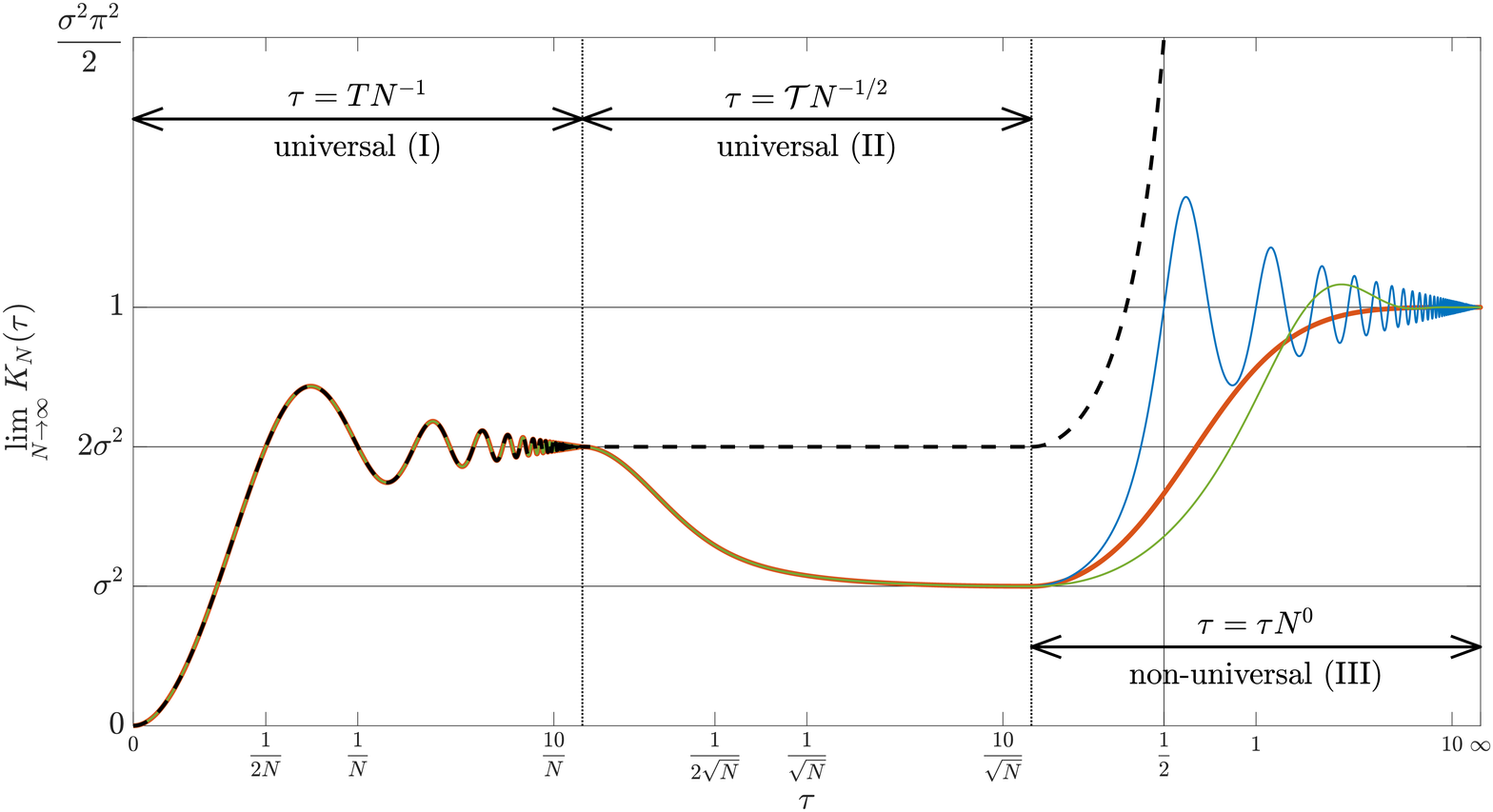}
\caption{Limiting curves ($N\rightarrow \infty$) for the form-factor across the three scaling regimes [${\rm (I)}$ -- Eq.~(\ref{K-infrared}), ${\rm (II)}$ -- Eq.~(\ref{K-interm}), and ${\rm (III)}$ -- Eq.~(\ref{K-tau-3})], glued together at vertical dotted lines. The functions $K^{(-1)}(T)$, $K^{(-1/2)}({\mathcal T})$ and
$K^{(0)}(\tau)$, describing the regimes ${\rm (I)}$, ${\rm (II)}$ and ${\rm (III)}$, correspondingly, are plotted vs variables $T=N\tau$, ${\mathcal T}=N^{1/2}\tau$ and $\tau$, each running over the entire real half-line compactified using the transformation $(0,\infty)=\tan((0,\pi/2))$. Solid red, green and blue curves correspond to the form-factor in the model of uncorrelated spacings drawn from the $\rm{Erlang}(3,3)$ (red), inverse Gaussian ${\rm IG}(1,3)$ (green) and uniform ${\rm U}(0,2)$ (blue) distributions, exhibiting identical mean and variance. The dashed black line -- to be discussed in the main text -- displays the limiting curve of the function $\lim_{N\rightarrow \infty}\omega^2 S_N(\omega)$ with $0\le \omega=2\pi\tau \le \pi$ (that is, $0\le \tau \le \rfrac{1}{2}$) for all three choices of the level spacing distribution. In the scaling regimes ${\rm (I)}$, ${\rm (II)}$ and ${\rm (III)}$, the curve is described by Eqs.~(\ref{SN-1st}), (\ref{SN-2nd}) with $\alpha=1/2$ and (\ref{SN-3rd}), respectively.
}
\label{Fig_cartoon}
\end{figure}

The three scaling regimes for the form-factor as $N \rightarrow \infty$ are illustrated in Fig.~\ref{Fig_cartoon}. The continuity of the entire curve is guaranteed by equality of limits $\lim_{T\rightarrow \infty} K^{(-1)}(T) = \lim_{{\mathcal T}\rightarrow 0} K^{(-1/2)}({\mathcal T})$ and $\lim_{{\mathcal T}\rightarrow \infty} K^{(-1/2)}({\mathcal T}) = \lim_{\tau\rightarrow 0} K^{(0)}(\tau)$, see Eqs.~(\ref{K1T}), (\ref{K2T}) and (\ref{K3T}). To highlight occurrence of both universal and non-universal $\tau$-domains in the form-factor, the latter is plotted for three different choices of level spacing distributions, $s_j \sim {\rm Erlang}(3,3)$, ${\rm IG}(1,3)$ and ${\rm U}(0,2)$, characterized by the same mean $\langle s_j \rangle =1$ and the variance $\sigma^2=1/3$:
\begin{eqnarray} \label{LSD-3}
    f_{s_j}(s) = \Theta(s) \times \left\{
                   \begin{array}{ll}
                     \displaystyle \frac{27}{2} s^2 \exp(-3s), & \hbox{${\rm Erlang}(3,3)$;} \\
                     \displaystyle \left(\frac{3}{2\pi s^3}\right)^{1/2} \exp \left( -\frac{3(s-1)^2}{2s}  \right), & \hbox{${\rm IG}(1,3)$;} \\
                     \displaystyle \frac{1}{2}\Theta(2-s), & \hbox{${\rm U}(0,2)$.}
                   \end{array}
                 \right.
\end{eqnarray}
The three curves coincide in the universal domains ${\rm (I)}$ and ${\rm (II)}$. On the contrary, in the third regime [${\rm (III)}$], the form-factor behavior is {\it non-universal} as the three curves evolve differently depending on a particular choice of the level spacing distribution. Yet, all three curves approach unity at infinity.
\noindent\newline\newline
{\it Implications for the power spectrum.}---We now turn to the discussion of a relation between the power spectrum Eq.~(\ref{ps-def}) and the form-factor Eq.~(\ref{FF-def}). To this end, we shall compare the limiting forms, as $N\rightarrow \infty$, of the form-factor, studied both analytically and numerically in the previous subsection, with the limiting behavior of the product $\omega^2 S_N(\omega)\mid_{\omega=2\pi\tau}$ as prompted by the form-factor approximation Eq.~(\ref{FFA}). The latter is plotted in Fig.~\ref{Fig_cartoon} by the black dashed line.

(i) For extremely low frequencies $\omega={\mathcal O}(N^{-1})$ (equivalently, short times $\tau={\mathcal O}(N^{-1})$) belonging to the first scaling regime [${\rm (I)}$], the two quantities are seen to {\it coincide}
\begin{eqnarray} \label{KS-1}
    {\rm (universal)}\;\;K^{\rm{(-1)}}(T) = {\rm (universal)}\;\;{\mathcal S}^{\rm{(-1)}}(\Omega)\mid_{\Omega=2\pi T},
\end{eqnarray}
see Eqs.~(\ref{SN-1st}) and (\ref{K-infrared}). The {\it universal} behavior of both spectral indicators in the domain ${\rm (I)}$ is illustrated in Fig.~\ref{Fig_cartoon} arranged for three different level spacing distributions specified by Eq.~(\ref{LSD-3}).

(ii) In the second scaling regime [${\rm (II)}$], characterized by intermediately low frequencies $\omega={\mathcal O}(N^{-1/2})$ (equivalently, $\tau={\mathcal O}(N^{-1/2})$), the limiting curve for the form-factor starts to deviate from the one for the product $\omega^2 S_N(\omega)\mid_{\omega=2\pi\tau}$, in concert with the analytical analysis,
\begin{eqnarray}
    {\rm (universal)}\;\;K^{\rm{(-1/2)}}(\mathcal{T}) \neq {\rm (universal)}\;\;{\mathcal S}^{\rm{(-1/2)}}(\tilde\Omega)\mid_{\tilde\Omega=2\pi {\mathcal T}}=2\sigma^2,
\end{eqnarray}
compare Eq.~(\ref{SN-2nd}) taken at $\alpha=1/2$ with Eq.~(\ref{K-interm}). While the product ${\mathcal S}^{\rm{(-1/2)}}(\tilde\Omega)$ is a constant throughout the entire domain ${\rm (II)}$, the form-factor is described by the universal function Eq.~(\ref{K-interm}) irrespective of a particular form of the level spacing distribution; the relative deviation between the two limiting curves reaches its maximum ($=2$) at the borderline between the regimes ${\rm (II)}$ and ${\rm (III)}$, in concert with the earlier conclusion of Ref.~\cite{ROK-2017}. How fast this factor of $2$ is approached depends only on the value of $\sigma^2$, as described by Eq.~(\ref{K-interm}). Hence, the relation Eq.~(\ref{FFA}) is clearly violated in the second scaling regime, apart from a single point at the border between the regimes ${\rm (I)}$ and ${\rm (II)}$ as stated below Eq.~(\ref{K3T}).

(iii) In the third scaling regime [${\rm (III)}$] emerging for $\omega = {\mathcal O}(N^0)$ (equivalently, $\tau = {\mathcal O}(N^0)$) the two limiting curves depart incurably from each other: while the product $\lim_{N\rightarrow \infty}\omega^2 S_N(\omega)$, shown by the dashed black line, follows the {\it universal} law Eq.~(\ref{SN-3rd}), the form-factor displays a {\it non-universal} behavior strongly depending on the particular form of level spacing distribution as highlighted by solid red, green and blue curves, see also Eq.~(\ref{K-tau-3}),
\begin{eqnarray}
    {\rm (nonuniversal)}\;\;K^{\rm{(0)}}(\tau) \neq {\rm (universal)}\;\; {\mathcal S}^{\rm{(0)}}(\omega)\mid_{\omega=2\pi \tau}.
\end{eqnarray}
Hence, the two spectral statistics -- the form-factor and the power spectrum -- {\it cannot} be reduced to each other for any finite frequency $0<\omega <\pi$ as $N \rightarrow \infty$.
\newline\newline\noindent
{\it Conclusion.}---Detailed analytical and numerical analysis, performed for eigenlevel sequences with uncorrelated, identically distributed level spacings, leads us to conclude that the spectral form-factor and the power spectrum are generically {\it two distinct statistical indicators}. This motivates us to revisit the problem of calculating the power spectrum for a variety of physically relevant eigenlevel sequences {\it beyond} the form-factor approximation. In the rest of the paper, this program, initiated in our previous publication \cite{ROK-2017}, will be pursued for (a) generic eigenlevel sequences possessing stationary level spacings and (b) eigenlevel sequences drawn from a variant of the circular unitary ensemble of random matrices. The latter case is of special interest as its $N\rightarrow \infty$ limit belongs to the spectral universality class shared by a large class of quantum systems with completely chaotic classical dynamics and broken time-reversal symmetry.

\section{Main results and discussion}

In this Section, we collect and discuss the major concepts and results of our work. Throughout the paper, we shall deal with eigenlevel sequences possessing {\it stationary level spacings} as defined below.

\begin{definition}\label{def-stationary} Consider an ordered sequence of (unfolded) eigenlevels $\{0 \le \varepsilon_1 \le \cdots \le \varepsilon_N\}$ with $N \in \mathbb{N}$. Let $\{s_1, \cdots, s_N \}$ be the sequence of spacings between consecutive eigenlevels such that $s_\ell = \varepsilon_\ell - \varepsilon_{\ell-1}$ with $\ell=1,\dots,N$ and $\varepsilon_0=0$. The sequence of level spacings is said to be {\it stationary} if (i) the average spacing
\begin{eqnarray}\label{skd}
\langle s_\ell \rangle = \Delta
\end{eqnarray}
is independent of $\ell=1,\dots,N$ and (ii) the covariance matrix of {\it spacings} is of the Toeplitz type:
\begin{eqnarray}\label{toep}
    {\rm cov}(s_\ell, s_m) =  I_{|\ell-m|} - \Delta^2
\end{eqnarray}
for all $\ell,m=1,\dots,N$. Here, $I_n$ is a function defined for non-negative integers $n$.
\hfill $\blacksquare$
\end{definition}
\noindent\par
\begin{remark}
While stationarity of level spacings is believed to emerge after unfolding procedure in the limit $N\rightarrow \infty$, see Ref.~\cite{BLS-2001}, it is not uncommon to observe stationarity even for {\it finite} eigenlevel sequences. Two paradigmatic examples of {\it finite}-$N$ eigenlevel sequences with stationary spacings include (i) a set of uncorrelated identically distributed eigenlevels \cite{RK-2019} mimicking quantum systems with integrable classical dynamics and (ii) eigenlevels drawn from the `tuned' circular ensembles of random matrices appearing in the random matrix theory approach to quantum systems with completely chaotic classical dynamics, see Section~\ref{PS-RMT-exact}.
\hfill $\blacksquare$
\end{remark}

\subsection{Main results}
{\bf First result.}---For {\it generic} eigenlevel sequences, the power spectrum Eq.~(\ref{ps-def}) is determined by {\it both} diagonal and off-diagonal elements of the covariance matrix $\langle \delta\varepsilon_\ell \delta\varepsilon_m \rangle$. In the important case of eigenlevel sequences {\it with stationary level spacings}, the power spectrum can solely be expressed in terms of {\it diagonal} elements $\langle \delta\varepsilon_\ell^2 \rangle$ of the covariance matrix. The Theorem \ref{PS-stationary-main} below, establishes an exact relation between the power spectrum (see Definition~\ref{def-01}) and a generating function of variances of {\it ordered} eigenvalues.

\begin{theorem}[First master formula]\label{PS-stationary-main}
  Let $N \in \mathbb{N}$ and $0\le \omega \le \pi$.  The power spectrum for an eigenlevel sequence $\{0\le \varepsilon_1 \le \cdots \le \varepsilon_N\}$ with stationary spacings equals
\begin{eqnarray}\label{smd-sum}
    S_N(\omega) = \frac{1}{N \Delta^2} {\rm Re} \left(  z \frac{\partial}{\partial z} - N - \frac{1-z^{-N}}{1-z}\right)
        \sum_{\ell=1}^N {\rm var}[\varepsilon_\ell]\, z^\ell,
\end{eqnarray}
where $z=e^{i \omega}$, $\Delta$ is the mean level spacing, and
\begin{eqnarray}
{\rm var}[\varepsilon_\ell] = \langle \delta\varepsilon_\ell^2 \rangle.
\end{eqnarray}
\end{theorem}

For the proof, the reader is referred to Section~\ref{Th23-proof}.
\noindent\newline\newline
{\bf Second result.}---Yet another useful representation -- the second master formula -- establishes an exact representation of the power spectrum in terms of a generating function
of probabilities $E_N(\ell;\epsilon)$ to observe exactly $\ell$ eigenlevels {\it below} the energy $\varepsilon$,
\begin{eqnarray}\fl \qquad \label{EML}
    E_N(\ell;\varepsilon) = \frac{N!}{\ell! (N-\ell)!}\left(\prod_{j=1}^\ell \int_{0}^{\varepsilon} d\epsilon_j\right) \left(\prod_{j=\ell+1}^N \int_{\varepsilon}^\infty d\epsilon_j\right)
    \, P_N(\epsilon_1,\dots,\epsilon_N).
\end{eqnarray}
Here, $P_N(\epsilon_1,\dots,\epsilon_N)$ is the joint probability density (JPDF) of $N$ unordered eigenlevels taken from a positive definite spectrum; it is assumed to be symmetric under a permutation of its arguments. Such an alternative albeit equivalent representation of the power spectrum will be central to the spectral analysis of quantum chaotic systems.

\begin{theorem}[Second master formula]\label{Th-2}
    Let $N \in \mathbb{N}$ and $0\le \omega \le \pi$, and let $\Phi_N(\varepsilon;\zeta)$ be the generating function
  \begin{eqnarray} \label{ps-gf}
    \Phi_N(\varepsilon;\zeta) = \sum_{\ell=0}^N (1-\zeta)^\ell E_N(\ell;\varepsilon)
  \end{eqnarray}
  of the probabilities defined in Eq.~(\ref{EML}). The power spectrum, Definition~\ref{def-01}, for an eigenlevel sequence with stationary spacings equals
\begin{eqnarray}\label{ps-2}\fl
    S_N(\omega) = \frac{2}{N \Delta^2} {\rm Re} \left(  z \frac{\partial}{\partial z} - N - \frac{1-z^{-N}}{1-z}\right)
        \frac{z}{1-z} \int_0^\infty d\epsilon \,\epsilon \left[
            \Phi_N(\epsilon;1-z) - z^N
        \right] - \tilde{S}_N(\omega),\nonumber\\
        {}
\end{eqnarray}
where $z=1-\zeta = e^{i\omega}$, $\Delta$ is the mean level spacing, and
\begin{eqnarray}\label{ps-tilde}
    \tilde{S}_N(\omega) = \frac{1}{N} {\rm Re} \left(  z \frac{\partial}{\partial z} - N - \frac{1-z^{-N}}{1-z}\right)
        \sum_{\ell=1}^N \ell^2 z^\ell \nonumber \\
         \qquad\qquad = \frac{1}{N} \left|
            \frac{1-(N+1)z^N + N z^{N+1}}{(1-z)^2}
        \right|^2.
\end{eqnarray}
\end{theorem}
For the proof, the reader is referred to Section~\ref{Th-2-proof}.

\begin{remark}
  Notably, representations Eqs.~(\ref{ps-gf}) and (\ref{ps-2}) suggest that the power spectrum is determined by {\it spectral correlation functions of all orders}. Contrary to the spacing distribution, which is essentially determined by the gap formation probability \cite{M-2004} $E_N(0;\varepsilon)$, the power spectrum depends on the {\it entire set} of probabilities $E_N(\ell;\varepsilon)$ with $\ell=0,1,\dots,N$.
\hfill $\blacksquare$
\end{remark}
\noindent\newline
{\bf Third result.}---To study the power spectrum in quantum systems with broken time-reversal symmetry and completely chaotic classical dynamics, let us consider the {\it tuned circular unitary ensemble} (${\rm TCUE}_N$). Obtained from the traditional circular unitary ensemble ${\rm CUE}_{N+1}$ \cite{M-2004} by conditioning its lowest eigen-angle to stay at zero, the ${\rm TCUE}_N$ is defined by the joint probability density of $N$ eigen-angles $\{\theta_1,\dots,\theta_N\}$ of the form
\begin{equation}\label{T-CUE} \fl \qquad
P_N(\theta_1,\dots,\theta_N) = \frac{1}{(N+1)!} \prod_{1 \le i < j \le N}^{}
    \left| e^{i\theta_i} - e^{i\theta_j} \right|^2
    \prod_{j=1}^{N} \left| 1 - e^{i\theta_j}\right|^2
\end{equation}
whose normalization is fixed by
\begin{eqnarray}\label{TCUE-norm}
\prod_{j=1}^{N}\int_0^{2\pi} \frac{d\theta_j}{2\pi}\, P_{N}(\theta_1,\dots,\theta_N)=1.
\end{eqnarray}
Such a seemingly minor tuning of ${\rm CUE}_{N+1}$ to ${\rm TCUE}_N$ induces stationarity of level spacings in ${\rm TCUE}_N$ for any $N \in {\mathbb N}$, see Corollary \ref{corr-theta-k} for the proof. We note in passing that traditional circular unitary ensemble lacks the stationarity property.

For the ${\rm TCUE}_N$, a general Definition~\ref{def-01} of the power spectrum can be adjusted as follows.

\begin{definition}
    Let $\{\theta_1 \le \cdots \le \theta_N\}$ be fluctuating ordered eigen-angles drawn from the ${\rm TCUE}_N$, $N \in {\mathbb N}$, with the mean level spacing $\Delta$ and let $\langle \delta\theta_\ell \delta\theta_m \rangle$ be the covariance matrix of eigen-angle displacements $\delta\theta_\ell = \theta_\ell - \langle \theta_\ell\rangle$ from their mean $\langle \theta_\ell\rangle$. A Fourier transform of the covariance matrix
\begin{eqnarray}\label{ps-def-tcue}
    S_N(\omega) = \frac{1}{N \Delta^2} \sum_{\ell=1}^N \sum_{m=1}^N \langle \delta\theta_\ell \delta\theta_m \rangle\, e^{i\omega (\ell-m)}, \quad \omega \in {\mathbb R}
\end{eqnarray}
is called the power spectrum of the ${\rm TCUE}_N$. Here, the angular brackets denote average with respect to the JPDF Eq.~(\ref{T-CUE}).
\hfill $\blacksquare$
\end{definition}

\begin{theorem}[Power spectrum in ${\rm TCUE}_N$]\label{Th-3}
    Let $\{\theta_1 \le \cdots \le \theta_N\}$ be fluctuating ordered eigen-angles drawn from the ${\rm TCUE}_N$.
    Then, for any $N \in {\mathbb N}$ and all $0 \le \omega \le \pi$, the power spectrum admits exact representation
\begin{eqnarray} \fl \label{ps-tcue-1}
    S_N(\omega) =  \frac{(N+1)^2}{\pi N} {\rm Re} \left(  z \frac{\partial}{\partial z} - N - \frac{1-z^{-N}}{1-z}\right)
        \frac{z}{1-z} \int_0^{2\pi} \frac{d\varphi}{2\pi} \,\varphi \, \Phi_N(\varphi;1-z) - \dbtilde{S}_N(\omega),\nonumber\\
{}
\end{eqnarray}
where
\begin{eqnarray} \label{ps-tcue-3}
    \dbtilde{S}_N(\omega) = \frac{1}{N} \left|
            \frac{1-(N+1)z^N + N z^{N+1}}{(1-z)^2}
        \right|^2 - \frac{(N+1)^2}{N} \frac{1}{|1-z|^2}
\end{eqnarray}
and
\begin{equation} \label{phin}
    \Phi_N(\varphi;\zeta) = \exp \left(
            -\int_{\cot(\varphi/2)}^{\infty} \frac{dt}{1+t^2} \left( \tilde{\sigma}_N(t;\zeta) + t \right)
        \right).
\end{equation}
Here, $z=1-\zeta=e^{i\omega}$ whilst the function $\tilde{\sigma}_N(t;\zeta)$ is a solution to the $\sigma$-Painlev\'e VI equation
\begin{eqnarray} \label{pvi} \fl
    \qquad \left( (1+t^2)\,\tilde{\sigma}_N^{\prime\prime} \right)^2 + 4 \tilde{\sigma}_N^\prime (\tilde{\sigma}_N - t \tilde{\sigma}_N^\prime)^2
    + 4 (\tilde{\sigma}_N^\prime+1)^2 \left(
        \tilde{\sigma}_N^\prime + (N+1)^2
    \right) = 0
\end{eqnarray}
satisfying the boundary condition
\begin{eqnarray} \label{pvi-bc}
    \tilde{\sigma}_N(t;\zeta) = -t + \frac{N(N+1)(N+2)}{3\pi t^2} \zeta + {\mathcal O}(t^{-4})
\end{eqnarray}
as $t\rightarrow \infty$.
\end{theorem}
For the proof of Theorem \ref{Th-3}, the reader is referred to Section~\ref{Th-3-proof}.

\begin{remark}
   Theorem \ref{Th-3} provides an exact RMT solution for the power spectrum in the ${\rm TCUE}_N$. Alternatively, but equivalently, the finite-$N$
    power spectrum can be expressed in terms of a Fredholm determinant (Section~\ref{Fredholm-sec}), Toeplitz determinant (Section~\ref{Toeplitz-sec}) and discrete Painlev\'e~V (${\rm dP_V}$) equations (Appendix~\ref{B-1}). While the Toeplitz representation is beneficial for performing a large-$N$ analysis of the power spectrum, the ${\rm dP_V}$ formulation is particularly useful for efficient numerical evaluation of the power spectrum for relatively large values of $N$.
\hfill $\blacksquare$
\end{remark}
\noindent\newline
{\bf Fourth (main) result.}---The most remarkable feature of the random matrix theory is its ability to predict universal statistical behavior of quantum systems. In this context, a large-$N$ limit of the power spectrum in the ${\rm TCUE}_N$ is expected to furnish a universal, parameter-free law, $S_\infty(\omega) = \lim_{N\rightarrow \infty} S_N(\omega)$, for the power spectrum. Its functional form is given in the Theorem~\ref{Th-4} below.

\begin{theorem}[Universal law]\label{Th-4}
For $0 < \omega < \pi$, the limit $S_\infty(\omega) = \lim_{N\rightarrow \infty} S_N(\omega)$ exists and equals
  \begin{eqnarray} \label{PS-exact} \fl
    S_\infty(\omega) = {\mathcal A}(\tilde{\omega}) \Bigg\{{\rm Im} \int_{0}^{\infty} \frac{d\lambda}{2\pi} \, \lambda^{1-2\tilde{\omega}^2} \, e^{i\tilde{\omega} \lambda}
    \nonumber\\
    \times
    \left[
        \exp \left(
                    - \int_{\lambda}^{\infty} \frac{dt}{t} \left( \sigma_1(t;\tilde\omega) - i \tilde{\omega} t + 2\tilde{\omega}^2\right)
            \right) -1
        \right] + {\mathcal B}(\tilde{\omega})\Bigg\}, \nonumber\\
    {}
\end{eqnarray}
where $\tilde{\omega} = \omega/2\pi$ is a rescaled frequency, and the functions ${\mathcal A}(\tilde{\omega})$ and ${\mathcal B}(\tilde{\omega})$ are defined as
\begin{eqnarray} \label{Aw-def}
    {\mathcal A}(\tilde{\omega}) = \frac{1}{2\pi} \frac{\prod_{j=1}^2 G(j+\tilde{\omega}) G(j-\tilde{\omega})}{\sin(\pi \tilde{\omega})}, \\
    \label{Bw-def}
    {\mathcal B}(\tilde{\omega}) = \frac{1}{2\pi} \sin(\pi \tilde{\omega}^2)\, \tilde{\omega}^{2\tilde{\omega}^2-2}\, \Gamma(2-2\tilde{\omega}^2).
\end{eqnarray}
Here, $G$ is the Barnes' $G$-function, $\Gamma$ is the Gamma function, whilst $\sigma_1(t;\tilde\omega)= \sigma_1(t)$ is the Painlev\'e V transcendent satisfying Eq.~(\ref{PV-family}) with $\nu=1$ and fulfilling the boundary conditions
\begin{eqnarray}\label{bc-s1-infty}
\sigma_1(t)=i \tilde{\omega} t-2\tilde{\omega}^2 - \frac{i t \gamma(t)}{1+\gamma(t)}+\mathcal{O}(t^{-1+2\tilde{\omega}}), && t\rightarrow\infty, \\
\label{bc-s1-zero}
\sigma_1(t) =\mathcal{O}(t \ln t), && t\rightarrow 0,
\end{eqnarray}
with $\gamma(t)$ being defined by Eq.~(\ref{eq:gamma}).
\end{theorem}

\begin{remark}
As a by-product of this Theorem, we have formulated a conjecture for a double integral identity involving a
fifth Painlev\'e transcendent. A mathematically-oriented reader is referred to Conjecture~\ref{conj}.
\hfill $\blacksquare$
\end{remark}

\begin{theorem}[Small-$\omega$ expansion]\label{Th-5}
  In the notation of Theorem \ref{Th-4}, the following expansion holds as $\omega\rightarrow 0$:
  \begin{eqnarray} \label{S-res-0}
    S_\infty (\omega) = \frac{1}{4\pi^2 \tilde\omega} + \frac{1}{2\pi^2} \tilde\omega \ln \tilde\omega +\frac{\tilde\omega}{12}
    + {\mathcal O}(\tilde\omega^2).
\end{eqnarray}
\end{theorem}

\begin{figure}
\includegraphics[width=\textwidth]{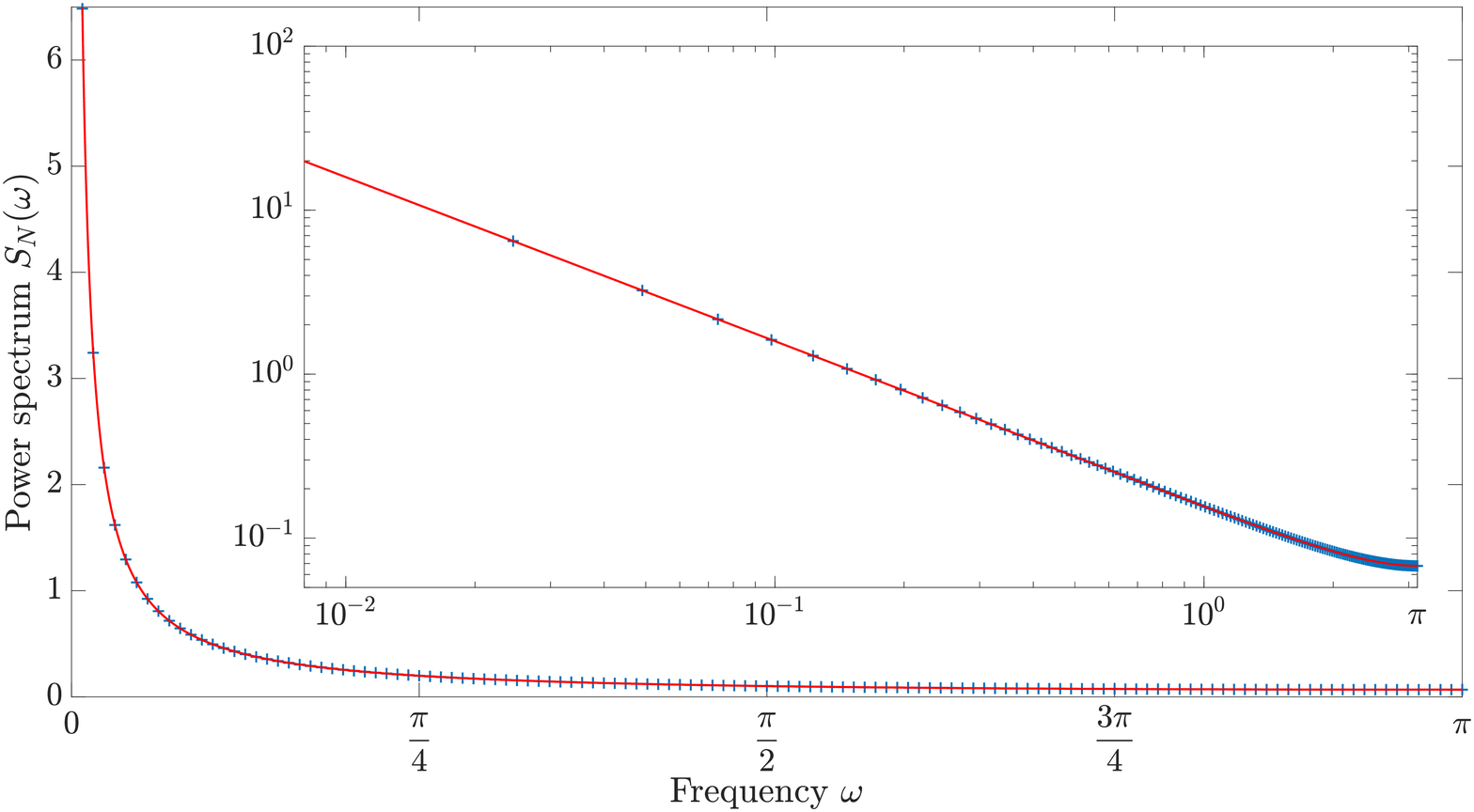}
\caption{A graph for the power spectrum as a function of frequency. Red line corresponds to the power spectrum calculated through the ${\rm dP_V}$ representation (Appendix~\ref{B-1}) of the exact Painlev\'e VI solution for $N=10^4$, see Theorem \ref{Th-3}. Blue crosses correspond to the power spectrum calculated for
sequences of $256$ unfolded ${\rm CUE}$ eigen-angles averaged over $10^7$ realizations. Inset: a log-log plot for the same graphs.
\label{Figure_ps_overall}}
\end{figure}

\begin{figure}
\includegraphics[width=\textwidth]{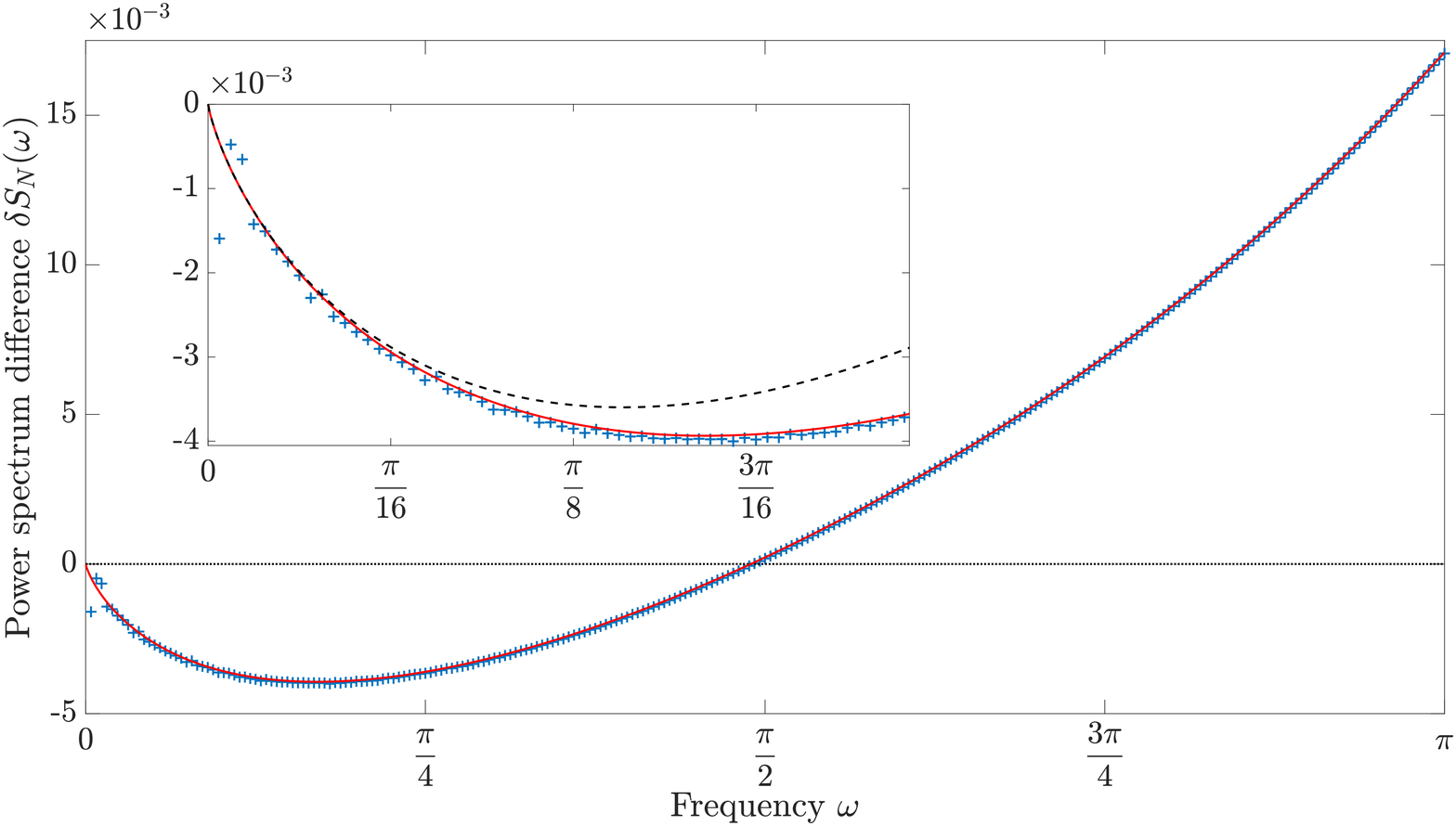}
\caption{Difference between the power spectrum and its singular part $1/2\pi\omega$ as described by Eq.~(\ref{BM}) at $\beta=2$ (see also the first term in Eq.~(\ref{S-res-0})). The singular part of the power spectrum corresponds to $\delta S_\infty(\omega)=0$ as represented by a gray dotted line. Red solid line: analytical prediction computed as explained in Fig.~\ref{Figure_ps_overall}. Blue crosses: simulation for $4 \times 10^8$ sequences of $512$ unfolded CUE eigenvalues. Inset: magnified portion of the same graph for $0\le \omega \le \pi/4$; additional black dashed line displays the difference $\delta S_\infty(\omega)$ calculated using the small-$\omega$ expansion Eq.~(\ref{S-res-0}).
\label{Figure_ps_diff}}
\end{figure}

For the proof of Theorems \ref{Th-4} and \ref{Th-5}, the reader is referred to Section~\ref{T-section}.

\subsection{Discussion}

In Figs.~\ref{Figure_ps_overall} and \ref{Figure_ps_diff}, the parameter-free prediction Eq.~(\ref{PS-exact}) for the power spectrum is confronted with the results of numerical simulations for the large-$N$ circular unitary ensemble ${\rm CUE}_N$. Two remarks are in order. (i) First, the limiting curve for $S_\infty(\omega)$ was approximated by the exact Painlev\'e VI solution computed for sufficiently large $N$ through its dPV representation worked out in detail in Appendix~\ref{B-1}. We have verified, by performing numerics for various values of $N$, that the convergence of ${\rm dP_V}$ representation of $S_N(\omega)$ to $S_\infty(\omega)$ is very fast, so that the $N=10^4$ curve provides an excellent approximation to the universal law for $S_\infty(\omega)$. A good match between the $N=10^4$ curve and the one plotted for a small-$\omega$ expansion Eq.~(\ref{S-res-0}) of $S_\infty(\omega)$ (see inset in Fig.~\ref{Figure_ps_diff}) lends an independent support to validity of our numerical procedure. (ii) Second, even though the theoretical results used for comparison refer to the ${\rm TCUE}_N$ -- rather than the ${\rm CUE}_N$ -- ensemble (which differ from each other by the weight function and the way the two are intrinsically unfolded \footnote{The spectra in ${\rm CUE}_N$ and ${\rm TCUE}_N$ ensembles are {\it intrinsically} unfolded for any $N\in{\mathbb N}$, albeit each in its own way. Indeed, in the ${\rm CUE}_N$ the {\it mean density} is a constant \cite{M-2004,PF-book}, while in the ${\rm TCUE}_N$ the {\it mean level spacing} is a constant, see Corollary \ref{corr-theta-k}. In the limit $N\rightarrow \infty$, the two types of unfolding are expected to become equivalent.}), the agreement between the ${\rm TCUE}_N$ theory and the ${\rm CUE}_N$ numerics is nearly perfect, which can naturally be attributed to the universality phenomenon emerging as $N\rightarrow \infty$.

The universal formula for $S_\infty(\omega)$, stated in Theorem \ref{Th-4}, is the {\it central result} of the paper. We expect it to hold {\it universally} for a wide class of random matrix models belonging to the $\beta=2$ Dyson's symmetry class, as the matrix dimension $N \rightarrow \infty$. Expressed in terms of a fifth Painlev\'e transcendent, the universal law Eq.~(\ref{PS-exact}) can be viewed as a power spectrum analog of the Gaudin-Mehta formula Eq.~(\ref{LSD-PV}) for the level spacing distribution.

Apart from establishing an explicit form of the universal random-matrix-theory law for $S_\infty(\omega)$, our theory reveals two important general aspects of the power spectrum which hold irrespective of a particular model of eigenlevel sequences: (i) similarly to the level spacing distribution, the power spectrum is determined by spectral correlations of {\it all orders}; (ii) in distinction to the level spacing distribution, which can solely be expressed in terms of the gap formation probability, the power spectrum is contributed by the {\it entire set} of probabilities that a spectral interval of a given length contains exactly $\ell$ eigenvalues with $\ell \ge 0$. As such, it provides a complementary statistical description of spectral fluctuations in stochastic spectra of various origins.

Considered through the prism of Bohigas-Giannoni-Schmit conjecture, the universal law Eq.~(\ref{PS-exact}) should hold for a variety of quantum systems with completely chaotic classical dynamics and broken time-reversal symmetry at not too low frequencies $T_*/T_H \lesssim \omega \le \pi$, when ergodicity and global symmetries -- rather than system specific features -- are responsible for shaping system's dynamics.

Potential applicability of our results to the non-trivial zeros of the Riemann zeta function deserves special mention. Indeed, according to the Montgomery-Odlyzko law (see, e.g., Ref.~\cite{KS-1999}), the zeros of the Riemann zeta function located high enough along the critical line are expected to follow statistical properties of the eigenvalues of large $U(N)$ matrices. This suggests that the universal law Eq.~(\ref{PS-exact}) could be detected "experimentally". Extensive, high-precision data accumulated by A.~M.~Odlyzko for billions of Riemann zeros \cite{O-2001} provide a unique opportunity for a meticulous test of the new universal law.

\section{Power spectrum for eigenlevel sequences with stationary spacings} \label{station}

In this Section, we provide proofs of two master formulae given by Theorem \ref{PS-stationary-main} and Theorem \ref{Th-2}.

\subsection{Stationary spectra} \label{stat-spec}

In view of Definition \ref{def-stationary}, we first establish a necessary and sufficient condition for eigenlevel sequences to possess stationarity of level spacings.

\begin{lemma} \label{Lemma-stat} For $N \in {\mathbb N}$, let $\{0\le \varepsilon_1 \le \cdots \le \varepsilon_N\}$ be an ordered sequence of unfolded eigenlevels such that $\langle
\varepsilon_1\rangle =\Delta$. An associated sequence of spacings between consecutive eigenlevels is stationary if and only if
\begin{eqnarray}\label{L-1}
    \langle (\varepsilon_\ell - \varepsilon_m)^q\rangle = \langle \varepsilon_{\ell-m}^q \rangle
\end{eqnarray}
for $\ell>m$ and both $q=1$ and $q=2$.
\end{lemma}
\noindent
\begin{proof}
The equivalence of Eq.~(\ref{skd}) to Eq.~(\ref{L-1}) at $q=1$ is self-evident. To prove the equivalence of Eq.~(\ref{toep}) to Eq.~(\ref{L-1}) at $q=2$, we proceed in two steps.

First, let the covariance matrix of level spacings be of the form Eq.~(\ref{toep}). Substituting Eq.~(\ref{ELSJ}) into the l.h.s.~of Eq.~(\ref{L-1}) taken at $q=2$, and making use of Eq.~(\ref{toep}) twice,
\begin{eqnarray}
    \langle (\varepsilon_\ell - \varepsilon_m)^2\rangle &=& \sum_{i=m+1}^\ell \sum_{j=m+1}^\ell  \langle  s_i s_j\rangle =
    \sum_{i=m+1}^\ell \sum_{j=m+1}^\ell  I_{|i-j|} \nonumber\\
        &=& \sum_{i^\prime=1}^{\ell-m} \sum_{j^\prime=1}^{\ell-m}  I_{|i^\prime-j^\prime|} =
         \sum_{i^\prime=1}^{\ell-m} \sum_{j^\prime=1}^{\ell-m}  \langle  s_{i^\prime} s_{j^\prime}\rangle =
         \langle \varepsilon_{\ell-m}^2\rangle, \nonumber
\end{eqnarray}
we derive the r.h.s.~of Eq.~(\ref{L-1}) with $q=2$.

Second, let the ordered eigenvalues satisfy Eq.~(\ref{L-1}) at $q=2$. Substituting $s_{\ell(m)} = \varepsilon_{\ell(m)} - \varepsilon_{\ell(m)-1}$ into the definition of covariance matrix ${\rm cov}(s_\ell,s_m)$ of level spacings and making use of Eq.~(\ref{L-1}), we observe that Eq.~(\ref{toep}) indeed holds with $I_{|\ell-m|}$ of the form
\begin{eqnarray}
    I_{|\ell-m|} = \frac{1}{2}\langle \varepsilon_{|\ell-m|+1}^2\rangle + \frac{1}{2}\langle \varepsilon_{|\ell-m|-1}^2\rangle - \langle \varepsilon_{|\ell-m|}^2\rangle.
\end{eqnarray}
\end{proof}

\subsection{Proof of Theorem \ref{PS-stationary-main}} \label{Th23-proof}

It follows from Eq.~(\ref{L-1}) of Lemma \ref{Lemma-stat} written in the form
\begin{eqnarray} \label{delta-2}
    \langle \delta\varepsilon_\ell  \delta\varepsilon_m\rangle = \frac{1}{2} \left(
        \langle \delta\varepsilon_\ell^2 \rangle + \langle \delta\varepsilon_m^2 \rangle - \langle \delta\varepsilon_{|\ell-m|}^2 \rangle
        \right),
\end{eqnarray}
where $\delta\varepsilon_\ell = \varepsilon_\ell -\ell \Delta$. Substituting Eq.~(\ref{delta-2}) into the definition Eq.~(\ref{ps-def}) and reducing the number of summations therein, we derive Eq.~(\ref{smd-sum}).
\hfill $\square$

\begin{remark}
  For discrete frequencies $\omega_k = 2\pi k/N$ the power spectrum representation Eq.~(\ref{smd-sum}) simplifies to
    \begin{eqnarray}\label{smd-sum-k}
    S_N(\omega_k) = \frac{1}{N \Delta^2} {\rm Re} \left(  z_k \frac{\partial}{\partial z_k} - N \right)
        \sum_{\ell=1}^N {\rm var}[\varepsilon_\ell]\, z_k^\ell.
    \end{eqnarray}
Here, $z_k = e^{i\omega_k}$ and the derivative with respect to $z_k$ should be taken as if $z_k$ were a continuous variable.
\hfill $\blacksquare$
\end{remark}

\subsection{Proof of Theorem \ref{Th-2}} \label{Th-2-proof}
To prove the Theorem \ref{Th-2}, we need the following Lemma:
\begin{lemma} \label{Lemma-probs}
 For $N \in {\mathbb N}$, let $\{ \varepsilon_1 \le \cdots \le \varepsilon_N\}$ be an ordered sequence of eigenlevels supported on the half axis $(0,\infty)$, and let $E_N(\ell;\varepsilon)$ be the probability to find exactly $\ell$ eigenvalues below the energy $\varepsilon$, given by Eq.~(\ref{EML}), with $\ell=0,1,\dots,N$. The following relation holds:
\begin{eqnarray}\label{Lemma-probs-1}
    \frac{d}{d\varepsilon} E_N(\ell;\varepsilon) = p_\ell(\varepsilon) - p_{\ell+1}(\varepsilon).
\end{eqnarray}
Here, $p_\ell(\varepsilon)$ is the probability density of $\ell$-th ordered eigenlevel where $p_0(\varepsilon) = p_{N+1}(\varepsilon) =0$ for $\varepsilon>0$. Equivalently,
\begin{eqnarray}\label{Lemma-probs-2}
    p_\ell (\varepsilon) = - \sum_{j=0}^{\ell-1} \frac{d}{d\varepsilon} E_N(j;\varepsilon), \quad \ell=1,\dots,N.
\end{eqnarray}
\end{lemma}
\noindent
\begin{proof}
Differentiating Eq.~(\ref{EML}) and having in mind that the probability density of $\ell$-th ordered eigenvalue equals
\begin{eqnarray}\fl \qquad \label{p-L}
    p_\ell(\varepsilon) = \frac{N!}{(\ell-1)! (N-\ell)!}\left(\prod_{j=1}^{\ell-1} \int_{0}^{\varepsilon} d\epsilon_j\right) \left(\prod_{j=\ell+1}^N \int_{\varepsilon}^\infty d\epsilon_j\right)\nonumber\\
        \qquad \qquad \qquad \qquad \times
    \, P_N(\epsilon_1,\dots,\epsilon_{\ell-1},\varepsilon, \epsilon_{\ell+1},\dots,\epsilon_N),
\end{eqnarray}
we derive Eqs.~(\ref{Lemma-probs-1}) and (\ref{Lemma-probs-2}).
\end{proof}
\noindent\newline
{\bf Proof of Theorem \ref{Th-2}.}---Equipped with Lemma \ref{Lemma-probs}, we are ready to prove Theorem \ref{Th-2}. First, we observe that Eqs.~(\ref{ps-gf}) and Eq.~(\ref{Lemma-probs-2}) induce the relation
  \begin{eqnarray}
    \sum_{\ell=1}^N  z^\ell p_\ell(\varepsilon) = -\frac{z}{1-z} \frac{d}{d\varepsilon} \left[
        \Phi_N(\varepsilon;1-z) - z^N
    \right].
\end{eqnarray}
Second, we split the variance Eq.~(\ref{smd-sum}) into ${\rm var}[\varepsilon_\ell]= \langle \varepsilon_\ell^2 \rangle - \ell^2 \Delta^2$. The later term
produces the contribution ${\tilde S}_N(\omega)$ in Eq.~(\ref{ps-2}) while the former brings
\begin{eqnarray} \label{aux-01}
    \sum_{\ell=1}^N \langle \varepsilon_\ell^2 \rangle z^\ell &=& - \frac{z}{1-z} \int_{0}^{\infty} d\epsilon\, \epsilon^2 \,
    \frac{d}{d\varepsilon} \left[
        \Phi_N(\varepsilon;1-z) - z^N
    \right] \nonumber\\
    &=& \frac{2z}{1-z} \int_{0}^{\infty} d\epsilon\, \epsilon \,
    \left[
        \Phi_N(\varepsilon;1-z) - z^N
    \right].
\end{eqnarray}
Integration by parts performed in the last line is justified provided an average number of eigenlevels ${\mathcal N}_N(\varepsilon)$ in the tail region $(\varepsilon,\infty)$ exhibits sufficiently fast decay ${\mathcal N}_N (\varepsilon) \sim \varepsilon^{-(2+\delta)}$ with $\delta>0$ as $\varepsilon \rightarrow \infty$ \footnote{Indeed, Eqs. (\ref{EML}) and (\ref{ps-gf}) imply an integral representation
$$
    \Phi_N(\varepsilon;1-z) = \prod_{\ell=1}^{N} \left(
        z \int_{0}^{\infty} d\epsilon_\ell + (1-z) \int_{\varepsilon}^{\infty} d\epsilon_\ell
    \right) \, P_N (\epsilon_1,\dots,\epsilon_N).
$$
Letting $\varepsilon\rightarrow\infty$, we generate a large-$\varepsilon$ expansion of the form
$$
    \Phi_N(\varepsilon;1-z) = z^N + \sum_{\ell=1}^N z^{N-\ell} (1-z)^\ell \left( \prod_{j=1}^{\ell} \int_{\varepsilon}^{\infty} d\epsilon_j\right) R_{\ell,N} (\epsilon_1,\dots,\epsilon_\ell),
$$
where
$$
R_{\ell,N} (\epsilon_1,\dots,\epsilon_\ell) = \frac{N!}{(N-\ell)!} \left( \prod_{j=\ell+1}^N \int_{0}^{\infty} d\epsilon_j \right)
P_N(\epsilon_1,\dots,\epsilon_N)
$$
is the $\ell$-point spectral correlation function. To the first order, the expansion brings
$\Phi_N(\varepsilon;1-z) = z^N + z^{N-1}(1-z) {\mathcal N}_N (\varepsilon) + \dots$, where ${\mathcal N}_N (\varepsilon)$ is the mean spectral density $R_{1,N}(\epsilon)$ integrated over the interval $(\varepsilon,\infty)$. Hence, the required decay of ${\mathcal N}_N (\varepsilon)$ at infinity readily follows.}. Substituting Eq.~(\ref{aux-01}) into Eq.~(\ref{smd-sum}), we derive the first term in Eq.~(\ref{ps-2}). This ends the proof of Theorem \ref{Th-2}.
\hfill $\square$

\begin{remark}
  For discrete frequencies $\omega_k = 2\pi k/N$ the power spectrum representation Eq.~(\ref{ps-2}) simplifies to
    \begin{eqnarray}\label{ps-gf-simple}\fl \quad
    S_N(\omega_k) = \frac{2}{N \Delta^2} {\rm Re} \left(  z_k \frac{\partial}{\partial z_k} - N\right)
        \frac{z_k}{1-z_k} \int_0^\infty d\epsilon \,\epsilon \left[
            \Phi_N(\epsilon;1-z_k) - 1
        \right] - \frac{N}{|1-z_k|^2}.\nonumber\\
        {}
  \end{eqnarray}
Here, $z_k = e^{i\omega_k}$ and the derivative with respect to $z_k$ should be taken as if $z_k$ were a continuous variable.
\hfill $\blacksquare$
\end{remark}

\section{Power spectrum in the tuned circular unitary ensemble} \label{PS-RMT-exact}

In this Section, a general framework developed in Section~\ref{station} and summed up in Theorem \ref{Th-2} will be utilized to determine the power spectrum in the tuned circular ensemble of random matrices, ${\rm TCUE}_N$, for any $N\in {\mathbb N}$. For the definition of ${\rm TCUE}_N$, the reader is referred to Eqs.~(\ref{T-CUE}) and (\ref{TCUE-norm}).

\subsection{Correlations between ordered eigen-angles in ${\rm TCUE}_N$}

The main objective of this subsection is to establish stationarity of spacings between ordered ${\rm TCUE}_N$ eigen-angles. To this end, we prove Lemma \ref{Lemma-circular} and Lemma \ref{Lemma-stat-TCUE}. The sought stationarity will then be established in Corollary \ref{corr-theta-k}.

\begin{lemma}[Circular symmetry] \label{Lemma-circular} For $q=0, 1, \dots$ and $\ell=1,2,\dots, N$ it holds that
\begin{eqnarray}\label{LL-1}
    \langle \theta_\ell^q\rangle = \langle (2\pi - \theta_{N-\ell +1})^q \rangle.
\end{eqnarray}
\end{lemma}
\begin{proof}
    The proof is based on the circular-symmetry identity
    \begin{eqnarray}
        p_\ell(\varphi) = p_{N-\ell+1}(2\pi - \varphi)
    \end{eqnarray}
    between the probability density functions of $\ell$-th and $(N-\ell+1)$-th ordered eigenangles in the ${\rm TCUE}_N$. This relation can formally be derived from the
    representation
    \begin{eqnarray} \fl \label{pk-tcue}
    p_\ell(\varphi) = \frac{1}{(N+1)!} \frac{N!}{(\ell-1)! (N-\ell)!} \left| 1-  e^{i\varphi}\right|^2 \nonumber\\
        \times
      \left(\prod_{j=1}^{\ell-1} \int_{0}^{\varphi} \frac{d\theta_j}{2\pi}\right) \left(\prod_{j=\ell}^{N-1} \int_{\varphi}^{2\pi} \frac{d\theta_j}{2\pi}\right)
      \nonumber \\
      \times
       \prod_{1 \le i < j \le N-1}^{}
    \left| e^{i\theta_i} - e^{i\theta_j} \right|^2
    \prod_{j=1}^{N-1} \left| e^{i\varphi} - e^{i\theta_j}\right|^2 \left| 1 - e^{i\theta_j}\right|^2.
    \end{eqnarray}
    Indeed, Eq.~(\ref{pk-tcue}) yields
    \begin{eqnarray}  \label{pk-tcue-mirror}
    p_{N-\ell+1}(2\pi - \varphi) = \frac{1}{(N+1)!} \frac{N!}{(\ell-1)! (N-\ell)!} \left| 1-  e^{i(2\pi -\varphi)}\right|^2 \nonumber\\
        \times
      \left(\prod_{j=1}^{N-\ell} \int_{0}^{2\pi - \varphi} \frac{d\theta_j}{2\pi}\right) \left(\prod_{j=N-\ell+1}^{N-1} \int_{2\pi-\varphi}^{2\pi} \frac{d\theta_j}{2\pi}\right)
      \nonumber \\
      \times
       \prod_{1 \le i < j \le N-1}^{}
    \left| e^{i\theta_i} - e^{i\theta_j} \right|^2
    \prod_{j=1}^{N-1} \left| e^{i(2\pi - \varphi)} - e^{i\theta_j}\right|^2 \left| 1 - e^{i\theta_j}\right|^2.
    \end{eqnarray}
    The change of variables $\theta_j = 2\pi - \theta_j^\prime$ reduces the r.h.s.~of Eq.~(\ref{pk-tcue-mirror}) to Eq.~(\ref{pk-tcue}). Consequently,
    \begin{eqnarray}
        \langle \theta_\ell^q\rangle &=& \int_0^{2\pi} \frac{d\varphi}{2\pi} \, \varphi^q p_\ell(\varphi) =
        \int_0^{2\pi} \frac{d\varphi}{2\pi} \, \varphi^q p_{N-\ell+1}(2\pi-\varphi) \nonumber\\
        &=& \int_0^{2\pi} \frac{d\varphi^\prime}{2\pi} \, (2\pi - \varphi^\prime)^q p_{N-\ell+1}(\varphi^\prime)
        = \langle (2\pi - \theta_{N-\ell+1})^q\rangle.
    \end{eqnarray}
\end{proof}

\begin{lemma}[Translational invariance in the index space] \label{Lemma-stat-TCUE}  For $q=0, 1, \dots$ and $1 \le m < \ell \le N$ it holds that
\begin{eqnarray}\label{LL-2}
    \langle (\theta_\ell - \theta_m) ^q\rangle = \langle \theta_{\ell-m}^q \rangle.
\end{eqnarray}
\end{lemma}

\begin{proof}
    It is advantageous to start with the JPDF of {\it ordered} eigenangles in the ${\rm TCUE}_N$,
    \begin{eqnarray}\fl \label{tcue-ord}
    \qquad
       P_N^{\rm{(ord)}}(\theta_1,\dots,\theta_N) =  N! \, P_N(\theta_1,\dots,\theta_N) \, \mathds{1}_{0 \le \theta_1 \le \cdots \le \theta_N \le 2\pi} \nonumber\\
       = \frac{1}{N+1} \prod_{1 \le i < j \le N}^{}
    \left| e^{i\theta_i} - e^{i\theta_j} \right|^2
    \prod_{j=1}^{N} \left| 1 - e^{i\theta_j}\right|^2  \, \mathds{1}_{0 \le \theta_1 \le \cdots \le \theta_N \le 2\pi},
    \end{eqnarray}
where we have used the notation
$$
\mathds{1}_{0\le\theta_1\le\dots\le\theta_N\le 2\pi}=\prod_{1\le i<j\le N}\Theta(\theta_j-\theta_i)
$$
with $\Theta$ being the Heaviside step function. Given Eq.~(\ref{tcue-ord}), the $q$-th moment of the difference $\theta_\ell-\theta_m$ equals
\begin{eqnarray}
    \langle (\theta_\ell -\theta_m)^q \rangle
     &=& \int_0^{2\pi} \frac{d\theta_1}{2\pi} \cdots \int_0^{2\pi} \frac{d\theta_m}{2\pi} \cdots
    \int_0^{2\pi} \frac{d\theta_\ell}{2\pi}
    \cdots \int_0^{2\pi} \frac{d\theta_N}{2\pi}  \nonumber\\
    &\times& \, (\theta_\ell - \theta_m)^q\,
    P_N^{\rm{(ord)}}(\theta_1,\dots,\theta_N).
\end{eqnarray}
Changing the integration variables $(\theta_1,\dots,\theta_N) \rightarrow (\theta_1^\prime,\dots,\theta_N^\prime)$ according to the map
\begin{eqnarray}\label{theta-map}
    \left\{
      \begin{array}{ll}
        \theta_{\ell-r}^\prime = \theta_\ell-\theta_r, & \hbox{$r=1,\dots,\ell-1$;} \\
        \theta_r^\prime = \theta_r, & \hbox{$r=\ell$;} \\
        \theta_{N+1+\ell-r}^\prime = 2\pi +\theta_\ell-\theta_r, & \hbox{$r=\ell+1,\dots,N$,}
      \end{array}
    \right.
\end{eqnarray}
and observing that both the probability density function $P_N^{\rm{(ord)}}$ and the integration domain stay invariant under the map Eq.~(\ref{theta-map}),
\begin{eqnarray}
    P_N^{\rm{(ord)}}(\theta_1^\prime,\dots,\theta_N^\prime) &=& P_N^{\rm{(ord)}}(\theta_1,\dots,\theta_N),\\
    \mathds{1}_{0 \le \theta_1 \le \cdots \le \theta_N \le 2\pi} &\rightarrow& \mathds{1}_{0 \le \theta_1^\prime \le \cdots \le \theta_N^\prime \le 2\pi},
\end{eqnarray}
we conclude that
\begin{eqnarray} \fl
    \langle (\theta_\ell -\theta_m)^q \rangle
     = \int_0^{2\pi} \frac{d\theta_1^\prime}{2\pi} \cdots  \int_0^{2\pi} \frac{d\theta_N^\prime}{2\pi} \, (\theta_{\ell-m}^\prime)^q\,
    P_N^{\rm{(ord)}}(\theta_1^\prime,\dots,\theta_N^\prime)
    = \langle \theta_{\ell-m}^q \rangle.
\end{eqnarray}
\end{proof}

\begin{corollary}\label{corr-theta-k}
    A sequence of spacings between consecutive eigenangles in ${\rm TCUE}_N$ is stationary such that the mean position of the $\ell$-th ordered
eigen-angle equals
    \begin{eqnarray}
        \langle \theta_\ell \rangle = \ell \Delta,
    \end{eqnarray}
where $\ell=1,\dots,N$ and
\begin{eqnarray}
    \Delta = \frac{2\pi}{N+1},
\end{eqnarray}
is the mean spacing.
\end{corollary}
\begin{proof}
  Indeed, combining Lemma \ref{Lemma-circular} taken at $q=1$ and Lemma \ref{Lemma-stat-TCUE} taken at $q=1$ and $m=\ell-1$, one concludes that the mean spacing
$$
    \Delta = \langle \theta_\ell - \theta_{\ell-1}\rangle = \frac{2\pi}{N+1}
$$
is constant everywhere in the eigenspectrum. Now we apply Lemma \ref{Lemma-stat} and Lemma \ref{Lemma-stat-TCUE} to complete the proof. \footnote{Notice that due to a formal convention $p_0(\varphi) =0$ stated below Eq.~(\ref{Lemma-probs-1}), one has to set $\langle \theta_0\rangle = 0$ if required.
}
\end{proof}

\subsection{Proof of Theorem \ref{Th-3}} \label{Th-3-proof}

Stationarity of level spacings in the ${\rm TCUE}_N$ established in Corollary \ref{corr-theta-k} allows us to use a `compactified' version of Theorem \ref{Th-2} in order to claim the representation stated by Eqs.~(\ref{ps-tcue-1}) and (\ref{ps-tcue-3}), where
\begin{eqnarray} \label{ps-tcue-2}
    \Phi_N(\varphi;\zeta) = \sum_{\ell=0}^N (1-\zeta)^\ell E_N(\ell;\varphi)
\end{eqnarray}
is the generating function of the probabilities
\begin{eqnarray}  \fl \label{EN-TCUE}
    E_N(\ell;\varphi) =
    \frac{N!}{\ell! (N-\ell)!}\left(\prod_{j=1}^\ell \int_{0}^{\varphi} \frac{d\theta_j}{2\pi}\right) \left(\prod_{j=\ell+1}^N \int_{\varphi}^{2\pi} \frac{d\theta_j}{2\pi}\right)
    \, P_N(\theta_1,\dots,\theta_N)
\end{eqnarray}
to find exactly $\ell$ eigen-angles in the interval $(0,\varphi)$ of the ${\rm TCUE}_N$ spectrum. The JPDF $P_N(\theta_1,\dots,\theta_N)$ is defined in Eq.~(\ref{T-CUE}).

Substituting Eqs.~(\ref{EN-TCUE}) and (\ref{T-CUE}) into Eq.~(\ref{ps-tcue-2}), one derives a multidimensional-integral representation of the generating function $\Phi_N(\varphi;\zeta)$ in the form
\begin{eqnarray} \fl
    \Phi_N(\varphi;\zeta) = \frac{1}{(N+1)!}
    \prod_{j=1}^N \left( \int_0^{2\pi} - \zeta \int_0^\varphi \right) \frac{d\theta_j}{2\pi}
    \prod_{1 \le i < j \le N}^{}
    \left| e^{i\theta_i} - e^{i\theta_j} \right|^2
    \prod_{j=1}^{N} \left| 1 - e^{i\theta_j}\right|^2, \nonumber\\
    \label{phintheta} {}
\end{eqnarray}
satisfying the {\it symmetry relation}
\begin{eqnarray}  \label{phin-sym}
   \Phi_N(2\pi-\varphi;\zeta) &=& (1-\zeta)^N \Phi_N\left(\varphi;\frac{\zeta}{\zeta-1}\right) \nonumber\\
   &=& (1-\zeta)^N \Phi_N(\varphi; \bar{\zeta})
        = (1-\zeta)^N \overline{\Phi_N(\varphi; \zeta)}.
\end{eqnarray}

Multidimensional integrals of the CUE-type akin to Eq.~(\ref{phintheta}) have been studied in much detail in Ref.~\cite{FW-2004} whose authors
employed the $\tau$-function theory \cite{O-1987} of Painlev\'e equations. To proceed with evaluation of the generating function of our interest, we introduce a new set of integration variables
\begin{eqnarray}\label{ch-var}
    e^{i\theta_j} = \frac{i\lambda_j-1}{i\lambda_j+1}
\end{eqnarray}
to write down the generating function Eq.~(\ref{phintheta}) in the form
\begin{eqnarray} \fl
    \Phi_N(\varphi;\zeta) = \frac{2^{N(N+1)}}{\pi^N (N+1)!}
    \prod_{j=1}^N \left( \int_{-\infty}^{+\infty} - \zeta \int_{\cot(\varphi/2)}^{+\infty} \right) \frac{d\lambda_j}{(1+\lambda_j^2)^{N+1}}  \prod_{1 \le i < j \le N}^{}
    \left| \lambda_i - \lambda_j \right|^2. \nonumber\\
    {}
\end{eqnarray}
Its Painlev\'e VI representation can be read off from Ref.~\cite{FW-2004} to establish Eqs.~(\ref{phin}), (\ref{pvi}) and also (\ref{pvi-bc}). For a detailed
derivation of the boundary condition Eq.~(\ref{pvi-bc}), the reader is referred to Appendix~\ref{A-1}.
\hfill $\square$

\begin{remark}
     For a set of discrete frequencies
$$
\omega_k^\prime = \frac{2\pi k}{N+1}
$$
the free term in Eq.~(\ref{ps-tcue-1}) nullifies, $\dbtilde{S}_N(\omega_k^\prime)=0$, bringing a somewhat tidier formula
\begin{eqnarray} \fl \label{ps-tcue-4}
  S_N(\omega_k^\prime) =  \frac{(N+1)^2}{\pi N} {\rm Re} \left(  z_k^\prime \frac{\partial}{\partial z_k^\prime} - N - 1 \right)
        \frac{z_k^\prime}{1-z_k^\prime} \int_0^{2\pi} \frac{d\varphi}{2\pi} \,\varphi \, \Phi_N(\varphi;1-z_k^\prime),
\end{eqnarray}
where $z_k^\prime = e^{i \omega_k^\prime}$. This is essentially Eq.~(17) previously announced in our paper Ref.~\cite{ROK-2017}.
\hfill $\blacksquare$
\end{remark}

\subsection{Power spectrum in ${\rm TCUE}_N$ as a Fredholm determinant} \label{Fredholm-sec}

To derive a Fredholm determinant representation of the ${\rm TCUE}_N$ power spectrum, a determinantal structure  \cite{M-2004,PF-book} of spectral correlation functions in the ${\rm TCUE}_N$ should be established. This is summarized in Lemma \ref{corr-f-tcue} below.

\begin{lemma} \label{corr-f-tcue} For $\ell =1,\dots,N$, the $\ell$-point correlation function \cite{M-2004,PF-book}
\begin{eqnarray} \label{RLN-TCUE}
    R_{\ell,N}(\theta_1,\dots,\theta_\ell) = \frac{N!}{(N-\ell)!} \left(\prod_{j=\ell+1}^{N} \int_{0}^{2\pi} \frac{d\theta_j}{2\pi}
    \right) \, P_N(\theta_1,\dots,\theta_N)
\end{eqnarray}
in the ${\rm TCUE}_N$ ensemble, defined by Eqs.~(\ref{T-CUE}) and (\ref{TCUE-norm}), admits the determinantal representation
\begin{eqnarray}\label{rk-tcue-kappa}
    R_{\ell,N} (\theta_1,\dots,\theta_\ell) = {\det}_{1\le i,j \le \ell} \left[
        {\kappa}_{N} (\theta_i, \theta_j)
    \right],
\end{eqnarray}
where the ${\rm TCUE}_N$ scalar kernel
\begin{eqnarray}\label{kn-cd-new}
    {\kappa}_{N}(\theta,\theta^\prime) = {\mathcal S}_{N+1}(\theta-\theta^\prime) - \frac{1}{N+1} {\mathcal S}_{N+1}(\theta) {\mathcal S}_{N+1}(\theta^\prime)
\end{eqnarray}
is expressed in terms of the sine-kernel
\begin{eqnarray}
    {\mathcal S}_{N+1}(\theta) = \frac{\sin[(N+1)\theta/2]}{\sin(\theta/2)}
\end{eqnarray}
of the ${\rm CUE}_{N+1}$ ensemble.
\end{lemma}
\begin{proof}
While the determinantal form [Eq.~(\ref{rk-tcue-kappa})] of spectral correlation functions is a universal manifestation of the $\beta=2$ symmetry of the circular ensemble \cite{M-2004,PF-book}, a precise form of the two-point scalar kernel $\kappa_N(\theta,\theta^\prime)$ depends on peculiarities of the ${\rm TCUE}_N$ probability measure encoded in the weight function ($z = e^{i\theta}$)
\begin{eqnarray} \label{wzm}
W(z)  = \frac{1}{2} |1-z|^2 = 1-\cos\theta
\end{eqnarray}
characterizing the ${\rm TCUE}_N$ measure in Eq.~(\ref{T-CUE}). For aesthetic reasons, it is convenient to compute a scalar kernel $\kappa_N(\theta,\theta^\prime)$
in terms of polynomials $\{\psi_j(z)\}$ orthonormal on the unit circle $|z|=1$
\begin{eqnarray}
    \frac{1}{2 i\pi} \oint_{|z|=1} \frac{dz}{z} \, W(z) \, \psi_\ell(z) \overline{\psi_m (z)} = \delta_{\ell m}
\end{eqnarray}
with respect to the weight function $W(z)$. In such a case, a scalar kernel is given by either of the two representations ($w=e^{i\theta^\prime}$):
\begin{eqnarray} \label{kernel-sum}
    \kappa_N(\theta,\theta^\prime) &=& \sqrt{W(z)W(w)}  \, \sum_{\ell=0}^{N-1} \psi_\ell(z) \, \overline{\psi_\ell(w)} \\
        \label{kernel-darboux}
        &=&  \sqrt{W(z)W(w)} \, \frac{\overline{\psi_N(w)}\, \psi_N(z) - \overline{\psi^*_N(w)} \psi^*_N(z)}{\bar{w} z -1}.
\end{eqnarray}
Equation~(\ref{kernel-darboux}), containing reciprocal polynomials
\begin{eqnarray} \label{rec-pol}
    \psi^*_\ell(z) = z^\ell \, \overline{\psi_\ell(1/\bar{z})},
\end{eqnarray}
follows from Eq.~(\ref{kernel-sum}) by virtue of the Christoffel-Darboux identity \cite{I-2005}.

Since for the ${\rm TCUE}_N$ weight function Eq.~(\ref{wzm}), the orthonormal polynomials are known as Szeg\"o-Askey polynomials (see \S18 in Ref.~\cite{NIST}),
\begin{eqnarray} \label{SAP-01}
    \psi_\ell(z) = \sqrt{\frac{2}{(\ell+1)(\ell+2)}} \;\, {}_2 F_1 \left( - \ell,2; -\ell; z \right),
\end{eqnarray}
the reciprocal Szeg\"o-Askey polynomials are readily available, too:
\begin{eqnarray} \label{SAP-02}
    \psi^*_\ell(z) =\sqrt{\frac{2 (\ell+1)}{\ell+2}}  \;\, {}_2 F_1 \left( - \ell,1; -\ell-1; z \right).
\end{eqnarray}
Hence, Eqs.~(\ref{kernel-darboux}), (\ref{SAP-01}) and (\ref{SAP-02}) furnish an explicit expression for the ${\rm TCUE}_N$ scalar kernel $\kappa_N(\theta,\theta^\prime)$.

This being said, we would like to represent the ${\rm TCUE}_N$  scalar kernel in a more suggestive form. To do so, we notice that Szeg\"o-Askey polynomials Eq.~(\ref{SAP-01}) admit yet another representation
\begin{eqnarray} \label{SAP-03}
    \psi_\ell(z) =
    \sqrt{\frac{2}{(\ell+1)(\ell+2)}}\; \sum_{j=1}^{\ell+1} j z^{j-1}.
\end{eqnarray}
Substituting it further into Eq.~(\ref{kernel-sum}), one obtains:
\begin{eqnarray} \fl \label{2pk-01}
    \kappa_N(\theta,\theta^\prime) &=& \frac{2 i}{N+1}  e^{-i (\theta-\theta^\prime)/2}
    \frac{\sin[\theta/2] \sin[\theta^\prime/2]}
    {\sin[(\theta-\theta^\prime)/2]} \sum_{j=0}^N \sum_{k=0}^N (N-j-k) z^j \bar{w}^k \\
    \fl
    &=& \frac{2 i}{N+1}  e^{-i (\theta-\theta^\prime)/2}
    \frac{\sin[\theta/2] \sin[\theta^\prime/2]}
    {\sin[(\theta-\theta^\prime)/2]} \sum_{j=0}^N \sum_{k=0}^N \left(N- z\frac{\partial}{\partial z}- {\bar w}\frac{\partial}{\partial \bar{w}}\right) z^j \bar{w}^k.
\end{eqnarray}
Owing to the representation of the ${\rm CUE}_N$ sine-kernel
\begin{eqnarray}
    {\mathcal S}_{N}(\theta) = \frac{\sin(N\theta/2)}{\sin(\theta/2)} = e^{-i (N-1) \theta/2} \sum_{j=0}^{N-1} z^j,
\end{eqnarray}
the above can further be reduced to
\begin{eqnarray}
    \kappa_N(\theta,\theta^\prime) &=& \frac{2}{N+1}  e^{i (N-1)(\theta-\theta^\prime)/2}
    \frac{\sin[\theta/2] \sin[\theta^\prime/2]}
    {\sin[(\theta-\theta^\prime)/2]} \nonumber\\
        &\times& \left(
            \frac{\partial}{\partial \theta^\prime} - \frac{\partial}{\partial \theta}
        \right)\, {\mathcal S}_{N+1}(\theta) {\mathcal S}_{N+1}(\theta^\prime).
\end{eqnarray}
Calculating derivatives therein, we derive
\begin{eqnarray}\label{kn-cd-new-100} \fl
    {\kappa}_{N}(\theta,\theta^\prime) = e^{i (N-1)(\theta-\theta^\prime)/2}
    \left({\mathcal S}_{N+1}(\theta-\theta^\prime) - \frac{1}{N+1} {\mathcal S}_{N+1}(\theta) {\mathcal S}_{N+1}(\theta^\prime)\right).
\end{eqnarray}
Spotting that the phase factor in Eq.~(\ref{kn-cd-new-100}) does not contribute to the determinant in Eq.~(\ref{rk-tcue-kappa}) completes the proof.
\end{proof}

\begin{remark}
An alternative determinantal representation of spectral correlation functions in the ${\rm TCUE}_N$ can be established if one
views the JPDF of the ${\rm TCUE}_N$ as the one of the traditional ${\rm CUE}_{N+1}$ ensemble, whose
lowest eigenangle is conditioned to stay at zero, as spelt out below.
\hfill $\blacksquare$
\end{remark}

\begin{lemma} \label{correlation-f} For $\ell =1,\dots,N$, the $\ell$-point correlation function, Eq.~(\ref{RLN-TCUE}), in the ${\rm TCUE}_N$ ensemble admits the determinantal representation
\begin{eqnarray}\label{rk-tcue}
    R_{\ell,N} (\theta_1,\dots,\theta_\ell) = \frac{1}{N+1} {\det}_{1\le i,j \le \ell+1} \left[
        {\mathcal S}_{N+1} (\theta_i -\theta_j)
    \right] \Big|_{\theta_{\ell+1}=0},
\end{eqnarray}
where ${\mathcal S}_{N+1}(\theta)$ is the ${\rm CUE}_{N+1}$ sine-kernel:
\begin{eqnarray}
    {\mathcal S}_{N+1}(\theta) = \frac{\sin[(N+1)\theta/2]}{\sin(\theta/2)}.
\end{eqnarray}
\end{lemma}
\begin{proof}
Equation (\ref{rk-tcue}) is self-evident as the determinant therein is the $(\ell+1)$-point correlation function in the ${\rm CUE}_{N+1}$ with one of the eigen-angles conditioned
to stay at zero whilst the denominator is the ${\rm CUE}_{N+1}$  mean density ${\mathcal S}_{N+1}(0)=N+1$.
\end{proof}

\begin{proposition} \label{prop-fred}
  The generating function $\Phi_N(\varphi;\zeta)$ in Eq.~(\ref{ps-tcue-1}) of Theorem \ref{Th-3} admits a Fredholm determinant representation
\begin{eqnarray}\label{Phi-FD}
    \Phi_N(\varphi;\zeta) = {\rm det} \big[
        \mathds{1} - \zeta  \mathbf{\hat{\kappa}}_N^{(0,\varphi)} \big],
\end{eqnarray}
where $\mathbf{\hat{\kappa}}_N^{(0,\varphi)}$ is an integral operator defined by
\begin{eqnarray} \label{Phi-FD-kappa}
    \big[\mathbf{\hat{\kappa}}_N^{(0,\varphi)} f\big] (\theta_1) = \int_{0}^{\varphi} \frac{d\theta_2}{2\pi} \kappa_N(\theta_1,\theta_2) \, f(\theta_2),
\end{eqnarray}
whilst $\kappa_N$ is the ${\rm TCUE}_N$ two-point scalar kernel specified in Lemma \ref{corr-f-tcue}.
\end{proposition}

\begin{proof}
To derive a Fredholm determinant representation of the power spectrum, we turn to Eq.~(\ref{ps-tcue-2}) rewriting it as a sum
\begin{eqnarray}
    \Phi_N(\varphi;\zeta) = \sum_{\ell=0}^{N} {{N}\choose{\ell}} \left( -\zeta \int_0^\varphi \right)^\ell  \left( \int_0^{2\pi} \right)^{N-\ell}
    \prod_{j=1}^N \frac{d\theta_j}{2\pi} \, P_N(\theta_1,\dots,\theta_N).\nonumber
\end{eqnarray}
Performing $(N-\ell)$ integrations, we obtain
\begin{eqnarray}
    \Phi_N(\varphi;\zeta) = \sum_{\ell=0}^{N} \frac{(-\zeta)^\ell}{\ell!} \left( \prod_{j=1}^\ell  \int_0^\varphi \frac{d\theta_j}{2\pi}\right) \, R_{\ell,N}(\theta_1,\dots,\theta_\ell), \nonumber
\end{eqnarray}
where $R_{\ell,N}(\theta_1,\dots,\theta_\ell)$ is the $\ell$-point correlation function in ${\rm TCUE}_N$ given by Eq.~(\ref{RLN-TCUE}). Its determinant representation Eq.~(\ref{rk-tcue-kappa})
yields the expansion
\begin{eqnarray}
    \Phi_N(\varphi;\zeta) = \sum_{\ell=0}^{N} \frac{(-\zeta)^\ell}{\ell!} \left(
    \prod_{j=1}^\ell \int_0^\varphi\frac{d\theta_j}{2\pi} \right)\,
    {\det}_{1\le i,j \le \ell} \left[
        {\kappa}_{N} (\theta_i, \theta_j)
    \right]. \nonumber
\end{eqnarray}
Here, $\kappa_N(\theta,\theta^\prime)$ is the two-point scalar kernel of the ${\rm TCUE}_N$ ensemble, see Lemma \ref{corr-f-tcue}
for its explicit form. Further, consulting, e.g., Appendix in Ref.~\cite{BK-2007}, one identifies
a sought Fredholm determinant representation given by Eqs.~(\ref{Phi-FD}) and (\ref{Phi-FD-kappa}).
\end{proof}

A Fredholm determinant representation of the power spectrum is particularly useful for asymptotic analysis of the power spectrum in the deep `infrared' limit $\omega \ll 1$
when $\zeta = 1-z \ll 1$.

\subsection{Power spectrum in ${\rm TCUE}_N$ as a Toeplitz determinant} \label{Toeplitz-sec}
To analyse the power spectrum in the limit $N \rightarrow \infty$ for $0<\omega<\pi$ being kept fixed, it is beneficial to represent the generating function $\Phi_N(\varphi;\zeta)$ [Eq.~(\ref{phintheta})] entering the exact solution Eq.~(\ref{ps-tcue-1}) with $\zeta = 1 - z$ in the form of a Toeplitz determinant with Fisher-Hartwig singularities.

\begin{proposition}\label{toep-prop}
The generating function $\Phi_N(\varphi;\zeta)$ in Eq.~(\ref{ps-tcue-1}) of Theorem \ref{Th-3} admits a Toeplitz determinant representation
  \begin{eqnarray} \label{GF-toeplitz-1}
    \Phi_N(\varphi; \zeta) = \frac{e^{i\varphi \tilde{\omega} N}}{N+1} \, D_N[f_{\tilde{\omega}}(z;\varphi)],
\end{eqnarray}
where $\tilde{\omega} = \omega/2\pi$, and
\begin{equation}\label{T-det-02}
    D_N[f_{\tilde{\omega}}(z;\varphi)] =  {\rm det}_{0 \le j,\ell \le N-1} \left( \frac{1}{2 i\pi} \oint_{|z|=1} \frac{dz}{z} \,z^{\ell-j}
    f_{\tilde{\omega}}(z;\varphi) \right)
\end{equation}
is the Toeplitz determinant whose Fisher-Hartwig symbol
\begin{eqnarray} \label{FHS}
    f_{\tilde{\omega}}(z;\varphi) = |z-z_1|^2 \left(\frac{z_2}{z_1}\right)^{\tilde{\omega}} g_{z_1,\tilde{\omega}} (z)\, g_{z_2,-\tilde{\omega}}(z)
\end{eqnarray}
possesses power-type singularity at $z=z_1 = e^{i\varphi/2}$ and jump discontinuities
\begin{eqnarray}
    g_{z_{j}, \pm{\tilde{\omega}}} (z) = \left\{
                                         \begin{array}{ll}
                                           e^{\pm i\pi \tilde{\omega}}, & \hbox{$0 \le {\rm arg\,} z < {\rm arg\,} z_{j}$} \\
                                           e^{\mp i\pi \tilde{\omega}}, & \hbox{${\rm arg\,} z_{j} \le {\rm arg\,} z < 2\pi$}
                                         \end{array}
                                       \right.
\end{eqnarray}
at $z = z_{1,2}$ with $z_2 = e^{i(2\pi -\varphi/2)}$.
\end{proposition}
\begin{proof}
  Start with the multiple integral representation Eq.~(\ref{phintheta}) and
  make use of
Andr\'eief's formula \cite{A-1883,dB-1955}
\begin{eqnarray}
    \left(\prod_{j=1}^{N} \int_{{\mathcal L}} \frac{d\theta_j}{2\pi}\right)\, w(\theta_j) \, {\rm det}_{1\le j,\ell \le N} [f_{j-1}(\theta_\ell)]
    \, {\rm det}_{1\le j,\ell \le N} [g_{j-1}(\theta_\ell)] \nonumber\\
    \qquad = N! \, {\rm det}_{1\le j,\ell \le N} \left(
    \int_{{\mathcal L}} \frac{d\theta}{2\pi}\, w(\theta) f_{j-1}(\theta) g_{\ell-1}(\theta)
    \right)
\end{eqnarray}
in which the weight function is set to $w(\theta)=(1-\zeta \Theta(\theta) \Theta(\varphi-\theta))|1-e^{i\theta}|^2$, integration domain is chosen to be ${\mathcal L} = (0, 2\pi)$, and $f_{j-1}(\theta)=\overline{g_{j-1}(\theta)} = e^{i(j-1)\theta}$, to derive
\begin{eqnarray} \label{T-det-01}
     \Phi_N(\varphi; \zeta) = \frac{1}{N+1} {\rm det}_{0 \le j,\ell \le N-1} \left[ M_{j-\ell}(\varphi;\zeta) \right],
\end{eqnarray}
where
\begin{eqnarray} \label{Mjk}
   M_{j-\ell}(\varphi;\zeta)  = \left(
            \int_{0}^{2\pi} - \zeta \int_{0}^{\varphi}
        \right)\frac{d\theta}{2\pi} |1- e^{i\theta}|^2 e^{-i(j-\ell)\theta}.
\end{eqnarray}
Introduce a new integration variable $z=e^{i\theta}$ in Eq.~(\ref{Mjk}), adopt the standard terminology and notation of Refs. \cite{DIK-2014,CK-2015} to figure out equivalence of Eqs.~(\ref{T-det-01}) and (\ref{Mjk}) to the statement of the proposition.
\end{proof}

\section{Power spectrum in quantum chaotic systems: Large-$N$ limit} \label{T-section}

In the limit $N \rightarrow \infty$, the exact solution for the ${\rm TCUE}_N$ power spectrum should converge to a universal law. To determine it, we shall perform an asymptotic analysis of the exact solution Eqs.~(\ref{ps-tcue-1}) and (\ref{ps-tcue-3}), stated in Theorem \ref{Th-3}, with the generating function $\Phi_N(\varphi; \zeta)$ being represented as a Toeplitz determinant specified in Proposition \ref{toep-prop}.

\subsection{Uniform asymptotics of the Toeplitz determinant}

To perform the integral in Eq.~(\ref{ps-tcue-1}) in the limit $N \rightarrow \infty$, {\it uniform} asymptotics of the Toeplitz determinant Eq.~(\ref{T-det-02}) are required in the subtle regime of two merging singularities. In our case, one singularity is of a root type while the other one is of both root and jump types. Relevant uniform asymptotics were recently studied in great detail by Claeys and Krasovsky \cite{CK-2015} who used the Riemann-Hilbert technique.

Two different, albeit partially overlapping, asymptotic regimes in $\varphi$ can be identified.
\newline\newline\noindent
{\it Asymptotics at the `left edge'.}---Defining the left edge as the domain $0\le \varphi <\varphi_0$, where $\varphi_0$ is sufficiently small \footnote{In fact, here $\varphi_0 =2\pi -\epsilon$ with $\epsilon >0$.}, the following asymptotic expansion holds {\it uniformly} as $N \rightarrow \infty$ (see Theorems 1.5 and 1.8 in Ref.~\cite{CK-2015})
\begin{eqnarray} \label{Edge-1}
    \ln D_N[f_{\tilde{\omega}}(z;\varphi)] &=& \ln N - i(N-1) \tilde{\omega} \varphi - 2\tilde{\omega}^2 \ln \left(
        \frac{\sin(\varphi/2)}{\varphi/2}\right) \nonumber\\
         &+& \int_{0}^{- i N \varphi} \frac{ds}{s}\, \sigma(s) + {\mathcal O}(N^{-1+ 2 \tilde{\omega}}),
\end{eqnarray}
so that
\begin{eqnarray} \label{Edge-2} \fl
    \Phi_N(\varphi;\zeta) = e^{i \tilde{\omega} \varphi} \left(
        \frac{\sin(\varphi/2)}{\varphi/2}
    \right)^{-2\tilde{\omega}^2} \exp\left(
        \int_{0}^{- i N \varphi} \frac{ds}{s}\, \sigma(s)
    \right) \left( 1 + {\mathcal O}(N^{-1+ 2 \tilde{\omega}}) \right).
\end{eqnarray}
Here $\tilde\omega = \omega/2\pi$ is a rescaled frequency so that $z=1-\zeta=e^{2 i\pi\tilde\omega}$. The function $\sigma(s)$ is the fifth Painlev\'e transcendent defined as the solution to the nonlinear equation
\begin{equation} \label{PV-eq}
 s^2 (\sigma^{\prime\prime})^2 = \left(\sigma - s \sigma^\prime + 2 (\sigma^\prime)^2 \right)^2 - 4 (\sigma^\prime)^2\left(
                (\sigma^\prime)^2 - 1
        \right)
\end{equation}
subject to the boundary conditions \cite{TC-private} \footnote{Notice that, in distinction to Ref.~\cite{CK-2015}, we kept two reminder terms in Eq.~(\ref{bc-inf}) -- oscillatory and non-oscillatory, even though the latter term is subleading. The reason for this is that the function $\sigma(s)$ will subsequently appear in the integral Eq.~(\ref{global-T}) which will make the non-oscillatory reminder term dominant.}
\begin{eqnarray}\label{bc-inf} \fl
    \sigma(s) = -{\tilde \omega} s - 2{\tilde \omega}^2 + \frac{s\gamma(s)}{1+\gamma(s)} + {\mathcal O}\left( e^{-i|s|}
        |s|^{-1+2{\tilde \omega}}
    \right)+ {\mathcal O}\left(
        |s|^{-1}
    \right) \quad {\rm as} \quad  s\rightarrow - i\infty
\end{eqnarray}
and
\begin{eqnarray}\label{bc-zero}
    \sigma(s)=  {\mathcal O}\left(
        |s|\,\ln |s|
    \right)\quad {\rm as} \quad s\rightarrow - i0_+.
\end{eqnarray}
The function $\gamma(s)$ in Eq.~(\ref{bc-inf}) equals
\begin{eqnarray}\label{eq:gamma}
    \gamma(s) = \frac{1}{4} \left|
        \frac{s}{2}
    \right|^{2(-1+2\tilde{\omega})} e^{-i |s|} e^{i\pi} \frac{\Gamma(2-\tilde{\omega}) \Gamma(1-\tilde{\omega})}{\Gamma(1+\tilde{\omega}) \Gamma(\tilde{\omega})}.
\end{eqnarray}
The above holds for $0 \le \tilde{\omega} < 1/2$.

\begin{remark}
Following Ref.~\cite{CK-2015}, we notice that in Eqs.~(\ref{Edge-1}) and (\ref{Edge-2}) the path of integration in the complex $s$-plane should be chosen to avoid a finite number of poles $\{ s_j \}$ of $\sigma(s)$ corresponding to zeros $\{ \varphi_j = i s_j/N \}$ in the asymptotics of the Toeplitz determinant $D_N[f_{\tilde{\omega}}(z;\varphi)]$. For the specific Fisher-Hartwig symbol Eq.~(\ref{FHS}) we expect $\{ s_j \}$ to be the empty set; numerical analysis of $D_N[f_{\tilde{\omega}}(z;\varphi)]$ suggests that its zeros stay away from the real line.
\hfill $\blacksquare$
\end{remark}
\noindent\newline\noindent
{\it Asymptotics in the `bulk'.}---Defining the `bulk' as the domain $\Omega(N)/N \le \varphi <\varphi_0$, where $\varphi_0$ is sufficiently small, and $\Omega(x)$ is a
positive smooth function such that $\Omega(N) \rightarrow \infty$ whilst $\Omega(N)/N \rightarrow 0$
as $N \rightarrow \infty$, the following asymptotic expansion holds {\it uniformly} (see Theorem 1.11 in Ref.~\cite{CK-2015}):
\begin{equation}\label{eq:DIK}\fl
    D_N[f_{\tilde{\omega}}(z;\varphi)] = N^{1-2\tilde{\omega}^2} G_{\tilde{\omega}}\, e^{i\tilde{\omega} \varphi} e^{-i\tilde{\omega}\pi} \left|2\sin \left( \frac{\varphi}{2} \right) \right|^{-2\tilde{\omega}^2}  \left( 1 + {\mathcal O}\left(\Omega(N)^{-1+2\tilde{\omega}}\right)  \right)
\end{equation}
so that
\begin{eqnarray}
\label{Th-111} \fl
\Phi_N(\varphi; \zeta)&=& N^{-2\tilde{\omega}^2} G_{\tilde{\omega}}\, e^{i \tilde{\omega} \varphi (N+1)} e^{-i\tilde{\omega}\pi}
        \left|
            2 \sin \left( \frac{\varphi}{2} \right)
        \right|^{-2\tilde{\omega}^2}\,\left( 1 + {\mathcal O}\left(\Omega(N)^{-1+2\tilde{\omega}}\right)  \right).
\end{eqnarray}
Here, $G_{\tilde{\omega}}$ is a known function of $\tilde{\omega}$
\begin{eqnarray}
    G_{\tilde{\omega}} = G(2+\tilde{\omega}) G(2-\tilde{\omega}) G(1+\tilde{\omega}) G(1-\tilde{\omega})
\end{eqnarray}
with $G(\cdots)$ being the Barnes' $G$-function. The above holds for $0 \le \tilde{\omega} < 1/2$. The leading term in Eqs.~(\ref{eq:DIK}) and (\ref{Th-111}) is due to Ehrhardt \cite{E-2001}.
\begin{remark}
  Since both asymptotic expansions [Eq.~(\ref{Edge-1}) and (\ref{eq:DIK})] hold uniformly in the domain
    $\Omega(N)/N \le \varphi <\varphi_0$, the following integral identity for $\sigma(s)$ should hold:
\begin{eqnarray} \label{global}
    \lim_{T\rightarrow +\infty} \left(
        \int_{0}^{-i T} \frac{ds}{s} \, \sigma(s) - i\tilde{\omega} T +2\tilde{\omega}^2 \ln T
    \right) = - i \pi \tilde{\omega} + \ln G_{\tilde{\omega}},
\end{eqnarray}
see Eq.~(1.26) in Ref.~\cite{CK-2015}. Had this global condition been derived independently, it would have provided an alternative route to producing the `bulk'
asymptotics out of those known in the edge region. Notice that as $T\rightarrow \infty$, the boundary condition Eq.~(\ref{bc-inf})
implies a stronger statement:
\begin{eqnarray} \label{global-T}
    \int_{0}^{-i T} \frac{ds}{s} \, \sigma(s) - i\tilde{\omega} T +2\tilde{\omega}^2 \ln T
     = - i \pi \tilde{\omega} + \ln G_{\tilde{\omega}} + {\mathcal O}(T^{-1}).
\end{eqnarray}
\hfill $\blacksquare$
\end{remark}

\subsection{Asymptotic analysis of the main integral}
In doing the large-$N$ asymptotic analysis of our exact solution for the power spectrum [Eqs. (\ref{ps-tcue-1}) and (\ref{phin})], we shall encounter
a set of integrals
\begin{eqnarray} \label{i-n-k}
    I_{N,k}(\zeta) = N \int_{0}^{2\pi} \frac{d\varphi}{2\pi} \, \varphi^k \Phi_N(\varphi;\zeta),
\end{eqnarray}
where $k$ is a non-negative integer and $\Phi_N(\varphi;\zeta)$ is given by Eq.~(\ref{phintheta}). We shall specifically be interested in $k=0$ and $1$.

\begin{lemma} \label{Lemma-IN0}
    In the notation of Eq.~(\ref{i-n-k}), we have:
    \begin{eqnarray} \label{In0-exact}
    I_{N,0} (\zeta) = \frac{N}{N+1} \frac{1-(1-\zeta)^{N+1}}{\zeta}.
\end{eqnarray}
Equation (\ref{In0-exact}) is exact for any $\zeta \in \mathbb{C}$.
\end{lemma}
\begin{proof}
To compute the integral Eq.~(\ref{i-n-k}) at $k=0$, we invoke the expansion Eq.~(\ref{ps-tcue-2}) of $\Phi_N(\varphi;\zeta)$ in terms
of probabilities $E_N(\ell;\varphi)$ of observing exactly $\ell$ eigenangles of ${\rm TCUE}_N$ in the interval $(0,\varphi)$,
\begin{eqnarray} \label{In0-expansion}
    I_{N,0}(\zeta) = N \sum_{\ell=0}^n (1-\zeta)^\ell \int_{0}^{2\pi} \frac{d\varphi}{2\pi}\,  E_N(\ell;\varphi).
\end{eqnarray}
The integral above can readily be calculated by performing integration by parts:
\begin{eqnarray} \fl
    \int_{0}^{2\pi} \frac{d\varphi}{2\pi}\,  E_N(\ell;\varphi) = \delta_{\ell,N} - \int_{0}^{2\pi}\frac{d\varphi}{2\pi}\, \varphi \frac{d}{d\varphi} E_N(\ell;\varphi)
    \nonumber\\
    = \delta_{\ell,N} + \frac{1}{2\pi}\int_{0}^{2\pi}\frac{d\varphi}{2\pi}\, \varphi \left( p_{\ell+1}(\varphi) - p_{\ell}(\varphi) \right).
\end{eqnarray}
In the second line, we have used the relation Eq.~(\ref{Lemma-probs-1}) which, in the context of ${\rm TCUE}_N$, acquires the multiplicative factor $1/2\pi$ in its r.h.s.; there,
$p_\ell(\varphi)$ is the probability density of the $\ell$-th ordered eigenangle. Further, identifying (see Corollary \ref{corr-theta-k})
\begin{eqnarray}
    \int_{0}^{2\pi}\frac{d\varphi}{2\pi}\, \varphi \, p_{\ell}(\varphi) = \langle \theta_\ell \rangle = \left\{
                                                                                                    \begin{array}{ll}
                                                                                                      \ell \Delta, & \hbox{$\ell=1,\dots,N$;} \\
                                                                                                      0, & \hbox{$\ell=0, N+1$.}
                                                                                                    \end{array}
                                                                                                  \right.
\end{eqnarray}
where $\Delta = 2\pi/(N+1)$ is the mean spacing, we conclude that
\begin{eqnarray} \label{IN0-res}
    \int_{0}^{2\pi} \frac{d\varphi}{2\pi}\,  E_N(\ell;\varphi) = \delta_{\ell,N} + \frac{\langle \theta_{\ell+1}\rangle - \langle \theta_{\ell}\rangle}{2\pi} = \frac{1}{N+1}
\end{eqnarray}
for all $\ell=0,\dots,N$. Substitution of Eq.~(\ref{IN0-res}) into Eq.~(\ref{In0-expansion}) ends the proof.
\end{proof}

\begin{remark}
    The fact that $I_{N,0} (\zeta)$ could be expressed in terms of elementary functions can be traced back to stationarity of level spacings in the ${\rm TCUE}_N$. For one, in the ${\rm CUE}_N$, an analogue of $I_{N,0} (\zeta)$ would have to be expressed in terms of the six Painlev\'e function.
\hfill $\blacksquare$
\end{remark}
\noindent\newline
{\it The integral $I_{N,k}$.}---Unfortunately, {\it exact} calculation of the same ilk is not readily available for $I_{N,k}$ with $k=1$. For this reason we would like to gain an insight from Eq.~(\ref{In0-exact}) as $N \rightarrow \infty$, which, eventually, is the limit we are mostly concerned with. To this end, we extract the leading order behavior of $I_{N,0}(\zeta)$ on the unit circle $|z|=|1-\zeta|=1$,
\begin{eqnarray} \label{i-n-0-asymp}
I_{N,0}(\zeta) = \frac{1}{\zeta} + (1-\zeta)^N \frac{1}{\bar{\zeta}} + {\mathcal O}(N^{-1}),
\end{eqnarray}
and observe that it contains terms of two types. (i) Those bearing a strongly oscillating prefactor $(1-\zeta)^N = z^N = e^{2 i \pi \tilde{\omega} N}$,
$$
(1-\zeta)^N \frac{1}{\bar{\zeta}}
$$
are contributed by a vicinity of $\varphi=2\pi$ in the integral Eq.~(\ref{i-n-k}) with $k=0$. (ii) On the contrary, such a prefactor is missing in the term coming from a vicinity of $\varphi=0$,
$$
\frac{1}{\zeta}.
$$
The contribution from the bulk of the integration domain appears to be negligible due to strong oscillations $e^{i \tilde{\omega}\varphi N}$ of the integrand therein, see Eq.~(\ref{Th-111}).

Equipped with these observations, we shall now proceed with an alternative, large-$N$, analysis of $I_{N,k}(\zeta)$ for $k=0$ and $k=1$, where terms of the same structure (with and without strongly oscillating prefactor) will appear. Aimed at the analysis of the power spectrum [Eq.~(\ref{ps-tcue-1})], whose representation contains a very particular $z$-operator, we shall only be interested in the leading order contributions to both terms. Notably, even though for $k=1$ a non-oscillating term is subleading as compared to an oscillating term, we shall argue that its contribution should still be kept.

To proceed with the large-$N$ analysis of $I_{N,k}$, we first rewrite the integral Eq.~(\ref{i-n-k}) as a sum of two
\begin{eqnarray} \label{I-n-k-1-2}
    I_{N,k}(\zeta) = I_{N,k}^{(1)}(\zeta) + I_{N,k}^{(2)}(\zeta)
\end{eqnarray}
such that
\begin{eqnarray} \label{i-n-k-1}
    I_{N,k}^{(1)}(\zeta) = N \int_{0}^{2\pi} \frac{d\varphi}{2\pi} \, \varphi^k \left( \Phi_N(\varphi;\zeta) - \Phi_N^{{\rm E}}(\varphi;\zeta) \right)
\end{eqnarray}
and
\begin{eqnarray} \label{i-n-k-2}
    I_{N,k}^{(2)}(\zeta) = N \int_{0}^{2\pi} \frac{d\varphi}{2\pi} \, \varphi^k \, \Phi_N^{{\rm E}}(\varphi;\zeta).
\end{eqnarray}
Here, $\Phi_N^{{\rm E}}(\varphi;\zeta)$ is an arbitrary integrable function; it will be specified later on.

Prompted by the `edge' and `bulk' asymptotic expansions of $\Phi_N(\varphi;\zeta)$ [Eqs.~(\ref{Edge-2}) and (\ref{Th-111})], we split the integral in Eq.~(\ref{i-n-k-1})
into three pieces
\begin{eqnarray} \label{I-1}
    I_{N,k}^{(1)}(\zeta) = L^{(1)}_{N,k}(\zeta) + C^{(1)}_{N,k}(\zeta)+ R^{(1)}_{N,k}(\zeta),
\end{eqnarray}
where
\begin{eqnarray}\label{L-i-1}
    L^{(1)}_{N,k}(\zeta) &=& N \int_{0}^{\Omega(N)/N} \frac{d\varphi}{2\pi} \, \varphi^k \left( \Phi_N(\varphi;\zeta) - \Phi_N^{{\rm E}}(\varphi;\zeta) \right),\\
    \label{C-i-1}
    C^{(1)}_{N,k}(\zeta) &=& N \int_{\Omega(N)/N}^{2\pi - \Omega(N)/N} \frac{d\varphi}{2\pi} \, \varphi^k \left( \Phi_N(\varphi;\zeta) - \Phi_N^{{\rm E}}(\varphi;\zeta) \right),\\
    \label{R-i-1}
    R^{(1)}_{N,k}(\zeta) &=& N \int_{2\pi - \Omega(N)/N}^{2\pi} \frac{d\varphi}{2\pi} \, \varphi^k \left( \Phi_N(\varphi;\zeta) - \Phi_N^{{\rm E}}(\varphi;\zeta) \right),
\end{eqnarray}
correspondingly.

To facilitate the asymptotic analysis, we would ideally like to choose $\Phi_N^{{\rm E}}(\varphi;\zeta)$ in such a way that the contribution of the `bulk' integral $C^{(1)}_{N,k}(\zeta)$ into $I_{N,k}^{(1)}(\zeta)$ becomes negligible. For the time being, let us {\it assume} that such a function is given by the leading term in Eq.~(\ref{Th-111}),
\begin{eqnarray} \label{ehrhardt}
    \Phi_N^{{\rm E}}(\varphi;\zeta) =  N^{-2\tilde{\omega}^2} G_{\tilde{\omega}}\, e^{i \tilde{\omega} \varphi (N+1)} e^{-i\tilde{\omega}\pi}
        \left|
            2 \sin \left( \frac{\varphi}{2} \right)
        \right|^{-2\tilde{\omega}^2}.
\end{eqnarray}
Then, $I_{N,k}^{(1)}(\zeta)$ will be dominated by the contributions coming from the `left-edge' [$L^{(1)}_{N,k}(\zeta)$] and the `right-edge' [$R^{(1)}_{N,k}(\zeta)$] parts of the integration domain. In fact, the contributions of the left and the right edges are related to each other; an exact relation between the two will be worked out and made explicit later on.
\noindent\newline\newline
{\it The integral $I_{N,k}^{(1)}(\zeta)$.}---Restricting ourselves to $k=0$ and $1$, we first consider the left-edge part $L_{N,k}^{(1)}(\zeta)$. Substituting Eqs.~(\ref{Edge-2}) and (\ref{ehrhardt}) into Eq.~(\ref{L-i-1}), we find, as $N\rightarrow \infty$:
\begin{eqnarray}\label{eq:L1-asym} \fl
    L^{(1)}_{N,k}(\zeta) =
                    N \int_0^{\Omega(N)/N} \frac{d\varphi}{2\pi} \varphi^k e^{i\tilde{\omega}\varphi}
                     \Bigg[ \left(\frac{\sin(\varphi/2)}{\varphi/2}\right)^{-2\tilde{\omega}^2}
            \exp\left(\int_0^{-i N \varphi} \frac{ds}{s} \sigma(s) \right) \nonumber\\
            \fl \qquad\qquad
            \times \left(1+\mathcal{O}(N^{-1+2\tilde{\omega}})\right)
            - N^{-2\tilde{\omega}^2} \left(2\sin(\varphi/2)\right)^{-2\tilde{\omega}^2} e^{i\tilde{\omega}\varphi N}
            e^{-i\tilde{\omega}\pi} G_{\tilde{\omega}}  \Bigg].
\end{eqnarray}
To get rid of $N$ in the integral over the Painlev\'e V transcendent, we make the substitution $\lambda=N \varphi$ to rewrite $ L^{(1)}_{N,k}(\zeta)$ in the
form
\begin{eqnarray} \fl
     L^{(1)}_{N,k}(\zeta) = \int_0^{\Omega(N)}\frac{d\lambda}{2\pi}\frac{\lambda^k}{N^k} e^{i\tilde{\omega}\lambda/N}\Bigg[ \left(\frac{\sin(\lambda/(2N))}{\lambda/(2N)}\right)^{-2\tilde{\omega}^2}
    \exp\left(\int_0^{-i \lambda} \frac{ds}{s} \sigma(s) \right) \nonumber\\
    \fl \qquad\qquad
        \times \left(1+\mathcal{O}(N^{-1+2\tilde{\omega}})\right)
     -N^{-2\tilde{\omega}^2} \left(2\sin(\lambda/(2N))\right)^{-2\tilde{\omega}^2} e^{i\tilde{\omega}\lambda}e^{-i\tilde{\omega}\pi}G_{\tilde{\omega}}  \Bigg].
\end{eqnarray}
Noting that $\lambda/N=\mathcal{O}(\Omega(N)/N)$ tends to zero as $N\rightarrow \infty$, we can further approximate $L^{(1)}_{N,k}(\zeta)$ as
\begin{eqnarray} \label{eq:L1Omega} \fl
    L^{(1)}_{N,k}(\zeta) =\frac{1}{N^k} \int_0^{\Omega(N)}\frac{d\lambda}{2\pi} \lambda^{k-2\tilde{\omega}^2} e^{i\tilde{\omega}\lambda}
    \Bigg[ \exp\left(\int_0^{-i \lambda} \frac{ds}{s} \sigma(s) -i\tilde{\omega}\lambda +2\tilde{\omega}^2 \ln\lambda \right)  \nonumber\\
    - e^{-i\tilde{\omega}\pi}G_{\tilde{\omega}}  \Bigg]
 +\mathcal{O}(\Omega(N)^{k+1} N^{-k-1+2\tilde{\omega}})+\mathcal{O}(\Omega(N)^{k+2} N^{-k-1}).
\end{eqnarray}
Next, one may use Eq.~(\ref{global-T}) to argue that replacing $\Omega(N)$ with infinity in Eq.~(\ref{eq:L1Omega}) produces an error term of the order
${\mathcal O}(\Omega(N)^{k-1-2\tilde{\omega}^2} N^{-k})$:
\begin{eqnarray} \label{eq:L1Omega-2} \fl
    L^{(1)}_{N,k}(\zeta) =\frac{1}{N^k} \int_0^{\infty}\frac{d\lambda}{2\pi} \lambda^{k-2\tilde{\omega}^2} e^{i\tilde{\omega}\lambda}
    \Bigg[ \exp\left(\int_0^{-i \lambda} \frac{ds}{s} \sigma(s) -i\tilde{\omega}\lambda +2\tilde{\omega}^2 \ln\lambda \right)  \nonumber\\
    - e^{-i\tilde{\omega}\pi}G_{\tilde{\omega}}  \Bigg]
 +\mathcal{O}(\Omega(N)^{k+1} N^{-k-1+2\tilde{\omega}}) \nonumber\\
     +\mathcal{O}(\Omega(N)^{k+2} N^{-k-1}) + {\mathcal O}(\Omega(N)^{k-1-2\tilde{\omega}^2} N^{-k}).
\end{eqnarray}
Further, choosing $\Omega(N)$ to be a slowly growing function, $\Omega(N) = \ln N$, one readily verifies that the third error term in Eq.~(\ref{eq:L1Omega-2})
is a dominant one out of the three as $0< \tilde{\omega} < 1/2$. Yet, it is smaller as compared to the
integral in Eq.~(\ref{eq:L1Omega-2}) by a factor $\Omega(N)^{k-1-2\tilde{\omega}^2}$ that tends to zero as $N\rightarrow \infty$. Thus, in the leading order,
we derive:
\begin{eqnarray} \label{eq:L1Omega-lead-a}
    L^{(1)}_{N,k}(\zeta) = \frac{1}{N^k} \mathfrak{L}_{k}^{(1)} (\zeta) + o(N^{-k}),
\end{eqnarray}
where
\begin{eqnarray} \fl \label{eq:L1Omega-lead-b}
    \mathfrak{L}_{k}^{(1)}(\zeta) = \int_0^{\infty}\frac{d\lambda}{2\pi} \lambda^{k-2\tilde{\omega}^2} e^{i\tilde{\omega}\lambda}
    \Bigg[ \exp\left(\int_0^{-i \lambda} \frac{ds}{s} \sigma(s) -i\tilde{\omega}\lambda +2\tilde{\omega}^2 \ln\lambda \right)
    - e^{-i\tilde{\omega}\pi}G_{\tilde{\omega}}  \Bigg], \nonumber\\
    {}
\end{eqnarray}
with $k=0$ and $1$.

Now, let us turn to the `right-edge' integral $R_{N,k}^{(1)}(\zeta)$. Due to the symmetry relation Eq.~(\ref{phin-sym}) shared by $\Phi_N^{{\rm E}}(\varphi;\zeta)$ too, we realize that the contributions of the left and the right edges are related to each other:
\begin{eqnarray} \fl \label{rhs}
   \overline{ R^{(1)}_{N,k}(\zeta)} = N (1-\bar{\zeta})^N \int_{0}^{\Omega(N)/N} \frac{d\varphi}{2\pi} \, (2\pi- \varphi)^k \left( \Phi_N(\varphi;\zeta) - \Phi_N^{{\rm E}}(\varphi;\zeta)\right).
\end{eqnarray}
Considering the integral in the r.h.s.~of Eq.~(\ref{rhs}) along the lines of the previous analysis, we conclude that the following formula holds as $N\rightarrow \infty$:
\begin{eqnarray} \label{eq:R1Omega-lead-a}
       R^{(1)}_{N,k}(\zeta) = (1-\zeta)^N \mathfrak{R}_k^{(1)}(\zeta) + o(1),
\end{eqnarray}
where
\begin{eqnarray} \fl \label{eq:R1Omega-lead-b}
    \overline{\mathfrak{R}_k^{(1)}(\zeta)} = (2\pi)^k \int_0^{\infty}\frac{d\lambda}{2\pi} \lambda^{-2\tilde{\omega}^2} e^{i\tilde{\omega}\lambda}
    \Bigg[ \exp\left(\int_0^{-i \lambda} \frac{ds}{s} \sigma(s) -i\tilde{\omega}\lambda +2\tilde{\omega}^2 \ln\lambda \right) \nonumber\\
        - e^{-i\tilde{\omega}\pi}G_{\tilde{\omega}}  \Bigg],
\end{eqnarray}
with $k=0$ and $1$.

Combining Eqs.~(\ref{eq:L1Omega-lead-a}),~(\ref{eq:L1Omega-lead-b}),~(\ref{eq:R1Omega-lead-a}) and (\ref{eq:R1Omega-lead-b}), we end up with the asymptotic result [Eq.~(\ref{I-1})]
\begin{eqnarray} \label{as-exp}
    I_{N,k}^{(1)}(\zeta) \mapsto \frac{1}{N^k} \mathfrak{L}_{k}^{(1)} (\zeta) + (1-\zeta)^N \mathfrak{R}_k^{(1)}(\zeta).
\end{eqnarray}
The notation $\mapsto$ was used here to stress that the r.h.s.~contains each leading order contribution of both terms, the oscillating and the non-oscillating, as discussed in the paragraph prior to Eq.~(\ref{I-n-k-1-2}).
\noindent\newline\newline
{\it The integral $I_{N,k}^{(2)}(\zeta)$.}---As soon as the function $\Phi_N^{\rm{E}}(\varphi;\zeta)$ contains a strongly oscillating factor $e^{i \tilde{\omega} \varphi N}$, the integral $I_{N,k}^{(2)}(\zeta)$ in Eq.~(\ref{i-n-k-2}) can be calculated by the stationary phase method \cite{T-2014}. Since there are no stationary points within the interval $(0,2\pi)$, the integral is dominated by contributions $L_{N,k}^{(2)}(\zeta)$ and $R_{N,k}^{(2)}(\zeta)$, coming from the vicinities of $\varphi=0$ and $\varphi=2\pi$, respectively.

\begin{lemma} \label{L-INK2}
    Let $I_{N,k}^{(2)}(\zeta)$ be defined by Eqs.~(\ref{i-n-k-2}) and (\ref{ehrhardt}), where $k$ is a fixed non-negative integer. Then, as $N\rightarrow \infty$, it can be represented in the following form:
    \begin{eqnarray}\label{i-n-k-2-sp}
    I_{N,k}^{(2)}(\zeta) = L_{N,k}^{(2)}(\zeta) + R_{N,k}^{(2)}(\zeta)
\end{eqnarray}
where
\begin{eqnarray}\label{eq:L2R2all}
    L^{(2)}_{N,k}(\zeta) &=& \frac{1}{N^k} \mathfrak{L}_k^{(2)} + o(N^{-k}), \\
    \label{eq:L2R2all-2}
    R^{(2)}_{N,k}(\zeta) &=& (1-\zeta)^N \mathfrak{R}_k^{(2)} + o(1),
\end{eqnarray}
and
\begin{eqnarray} \label{Lk2}
    \mathfrak{L}_k^{(2)}(\zeta) &=& \frac{G_{\tilde{\omega}}}{2\pi} \, e^{i\pi(k+1-2{\tilde{\omega}}-2{\tilde{\omega}}^2)/2} {\tilde{\omega}}^{-k-1+2{\tilde{\omega}}^2}
    \Gamma(k+1-2{\tilde{\omega}}^2),\\
    \label{Rk2}
    \mathfrak{R}_k^{(2)}(\zeta) &=&  \frac{G_{\tilde{\omega}}}{2\pi}\, e^{i\pi(-1+2{\tilde{\omega}}+2{\tilde{\omega}}^2)/2} {\tilde{\omega}}^{-1+2{\tilde{\omega}}^2} (2\pi)^{k} \Gamma(1-2{\tilde{\omega}}^2).
\end{eqnarray}
\end{lemma}
\begin{proof}
    Apply the stationary phase method \cite{T-2014} to calculate the integral Eq.~(\ref{i-n-k-2}).
\end{proof}

The Lemma \ref{L-INK2} brings the following asymptotic result
\begin{eqnarray}
    I_{N,k}^{(2)}(\zeta) \mapsto \frac{1}{N^k} \mathfrak{L}_{k}^{(2)} (\zeta) + (1-\zeta)^N \mathfrak{R}_k^{(2)}(\zeta),
\end{eqnarray}
compare with Eq.~(\ref{as-exp}).
\noindent\newline\newline
{\it The integral $I_{N,k}(\zeta)$.}---The calculation above implies that the main integral of our interest admits an asymptotic representation
\begin{eqnarray} \label{Ink-final}
        I_{N,k}(\zeta) \mapsto \frac{1}{N^k} \mathfrak{L}_k(\zeta) + (1-\zeta)^N \mathfrak{R}_k(\zeta)
\end{eqnarray}
with $k=0,1$ and
\begin{eqnarray} \label{Lkf}
        \mathfrak{L}_k(\zeta) &=& \mathfrak{L}_k^{(1)}(\zeta)+ \mathfrak{L}_k^{(2)}(\zeta), \\
        \label{Rkf}
        \mathfrak{R}_k(\zeta) &=& \mathfrak{R}_k^{(1)}(\zeta)+ \mathfrak{R}_k^{(2)}(\zeta).
\end{eqnarray}
We notice that both $\mathfrak{L}_k(\zeta)= \mathcal{O}(1)$ and $\mathfrak{R}_k(\zeta)= \mathcal{O}(1)$ and the factor $(1-\zeta)^N = z^N = e^{2 i \pi \tilde{\omega} N}$ in Eq.~(\ref{Ink-final}) is a strongly oscillating function of $\tilde{\omega}$ as $N \rightarrow \infty$, in concert with the discussion in the paragraph prior to Eq.~(\ref{I-n-k-1-2}).
\begin{remark}
  Our derivation of the main result of this Section, Eq.~(\ref{Ink-final}), was based on the assumption that the choice of $\Phi_N^{\rm E}(\varphi;\zeta)$ in the form Eq.~(\ref{ehrhardt}) makes the contribution of the `bulk' integral $C_{N,k}^{(1)}(\zeta)$ [Eq.~(\ref{C-i-1})] into $I_{N,k}^{(1)}(\zeta)$ negligible. If this is {\it not} the case, one should replace $\Phi_N^{\rm E}$ with some $\tilde{\Phi}_N^{\rm E}$ by adding to $\Phi_N^{\rm E}$ the higher-order corrections (up to ${\mathcal O(N^{-2})}$) that can be obtained from the full asymptotic expansion of $\Phi_N(\varphi;\zeta)$, see Remark 1.4 of Ref.~\cite{DIK-2014}. Inclusion of these higher-order corrections will reduce the contribution of $C_{N,k}^{(1)}(\zeta)$ to a negligible level as guaranteed by the rough upper-bound estimate
\begin{eqnarray} \label{upper-b}
    |C_{N,k}^{(1)}(\zeta)| &=& N \left|
           \int_{\Omega(N)/N}^{2\pi - \Omega(N)/N} \frac{d\varphi}{2\pi} \, \varphi^k \left( \Phi_N(\varphi;\zeta) - \Phi_N^{{\rm E}}(\varphi;\zeta) \right)
    \right| \nonumber\\
     &\le&
     N \int_{0}^{2\pi} \frac{d\varphi}{2\pi} \, \varphi^k \left| \Phi_N(\varphi;\zeta) - \tilde{\Phi}_N^{{\rm E}}(\varphi;\zeta) \right|
    = {\mathcal O}(N^{-1}).
\end{eqnarray}
On the other hand, the proposed modification of $\Phi_N(\varphi;\zeta)$ will produce corrections to the functions $L_{N,1}^{(1)}$, $R_{N,1}^{(1)}$, $L_{N,1}^{(2)}$ and $R_{N,1}^{(2)}$, which will clearly be subleading to those calculated in the leading order [see Eqs. (\ref{eq:L1Omega-lead-a}), (\ref{eq:R1Omega-lead-a}), (\ref{eq:L2R2all}),
(\ref{eq:L2R2all-2})]. For this reason, they will not affect the large-$N$ analysis of the power spectrum where only ${\mathcal O}(1)$ terms are kept.
\hfill $\blacksquare$
\end{remark}

\subsection{Proof of Theorem \ref{Th-4}}

Now we are in position to evaluate the power spectrum as $N\rightarrow \infty$. To proceed, we start with the exact, finite-$N$, representation
\begin{eqnarray} \fl \label{Sn-In1}
    S_N(\omega) = \frac{(N+1)^2}{\pi N^2} {\rm Re}\left\{ \left(  z \frac{\partial}{\partial z} - N - \frac{1-z^{-N}}{1-z}\right)
        \frac{z}{1-z} \, I_{N,1}(\zeta) \right\} - \dbtilde{S}_N(\omega)
\end{eqnarray}
following from Eqs.~(\ref{ps-tcue-1}) and (\ref{i-n-k}). Substituting $I_{N,1}$ given by Eq.~(\ref{Ink-final}) into Eq.~(\ref{Sn-In1}) and taking into account the
relation ${\mathfrak{R}}_{1}(\zeta)= 2\pi {\mathfrak{R}}_{0}(\zeta)$, following from Eqs.~(\ref{Rkf}), (\ref{Rk2}) and (\ref{eq:R1Omega-lead-b}),
we derive, as $N\rightarrow \infty$:
\begin{eqnarray} \fl \label{Snw}
    S_N(\omega) =  - \frac{1}{\pi}  \, {\rm Re}\, \Bigg\{ \frac{z}{1-z}\, {\mathfrak{L}}_{1}(\zeta)\Bigg\} +  2  {\rm Re}\, \Bigg\{
      \frac{z}{1-z} \left(  z^{N+1} \frac{d}{dz} {\mathfrak{R}}_{0}(\zeta)+\frac{{\mathfrak{R}}_{0}(\zeta)}{1-z}  \right) \Bigg\} \nonumber\\
        - 2 {\rm Re} \Bigg\{ \frac{(z-1)(1+z^N)}{|1-z|^4}
    \Bigg\} + o(1).
\end{eqnarray}
Here, the third term originates from the large-$N$ expansion of $\dbtilde{S}_N(\omega)$ [Eq.~(\ref{ps-tcue-3})]. Surprisingly, the last two terms in Eq.~(\ref{Snw}) cancel each other. This follows from the identity
\begin{eqnarray}\label{R0}
    \mathfrak{L}_0(\zeta) = \overline{\mathfrak{R}_0(\zeta)} = \frac{1}{\zeta}
\end{eqnarray}
that can be identified by comparing Eq.~(\ref{i-n-0-asymp}) with Eq.~(\ref{Ink-final}) taken at $k=0$. The cancellation implies the $N\rightarrow \infty$
result
\begin{eqnarray} \label{Snw-infty}
    S_\infty(\omega) =  - \frac{1}{\pi}  \, {\rm Re}\, \Bigg\{ \frac{z}{1-z}\, {\mathfrak{L}}_{1}(\zeta)\Bigg\}.
\end{eqnarray}
Substituting Eqs.~(\ref{Lkf}), (\ref{eq:L1Omega-lead-b}) and (\ref{Lk2}) into Eq.~(\ref{Snw-infty}),
we derive
\begin{eqnarray}\label{eq:SnResult} \fl
    S_\infty(\omega) = \frac{1}{\pi}  \, {\rm Re}\, \Bigg\{ \frac{e^{2 i \pi\tilde{\omega}}}{e^{2 i \pi\tilde{\omega}}-1}
    \Bigg( \int_0^{\infty}\frac{d\lambda}{2\pi} \lambda^{1-2\tilde{\omega}^2} e^{i\tilde{\omega}\lambda} \Bigg[ \exp\Big(\int_0^{-i \lambda} \frac{ds}{s} \sigma(s)
    -i\tilde{\omega}\lambda +2\tilde{\omega}^2 \ln\lambda \Big) \nonumber\\
    -e^{-i\tilde{\omega}\pi}G_{\tilde{\omega}}  \Bigg]
     - \frac{G_{\tilde{\omega}}}{2\pi} e^{-i\pi(\tilde{\omega}+\tilde{\omega}^2)} \tilde{\omega}^{-2+2\tilde{\omega}^2} \Gamma(2-2\tilde{\omega}^2)   \Bigg)\Bigg\},
\end{eqnarray}
where $\tilde{\omega} = \omega/2\pi$ is a rescaled frequency, and the function $\sigma(s)$ is the fifth Painlev\'e transcendent defined by Eqs.~(\ref{PV-eq}), (\ref{bc-inf}) and (\ref{bc-zero}). Equation~(\ref{eq:SnResult}) can be simplified down to
\begin{eqnarray}\label{eq:SnResult-2} \fl
    \qquad S_\infty(\omega)=  {\mathcal A}(\tilde{\omega}) \, \Bigg\{ {\rm Im}\,  \Bigg( \int_0^{\infty}\frac{d\lambda}{2\pi} \lambda^{1-2\tilde{\omega}^2}
    e^{i\tilde{\omega}\lambda} \nonumber\\
    \times \left[ \exp\left(\int_{-i \infty}^{-i\lambda} \frac{ds}{s}\left( \sigma(s) + s \tilde{\omega} +2\tilde{\omega}^2 \right)
    \right) -1  \right] \Bigg)
    +{\mathcal B}(\tilde{\omega}) \Bigg\},
\end{eqnarray}
where the functions ${\mathcal A}(\tilde\omega)$ and ${\mathcal B}(\tilde\omega)$ are defined as in Eqs.~(\ref{Aw-def}) and (\ref{Bw-def}). To derive Eq.~(\ref{eq:SnResult-2}) we have used the integral identity Eq.~(\ref{global}) to transform the exponent
\begin{eqnarray}\fl
    \exp \left(\int_0^{-i\lambda} \frac{\sigma(s)}{s} ds - i \tilde{\omega}\lambda+2\tilde{\omega}^2\ln \lambda \right)
    = G_{\tilde{\omega}} e^{-i\pi\tilde{\omega}} \nonumber\\
    \times \lim_{T\rightarrow\infty} \exp \left[ \int_{-i T}^{-i \lambda} \frac{\sigma(s)}{s} ds + i\tilde{\omega}(T-\lambda)
    + 2\tilde{\omega}^2 \ln(\lambda/T) \right]\nonumber\\
   = G_{\tilde{\omega}} e^{-i\pi\tilde{\omega}}  \exp \left[ \int_{-i \infty}^{-i \lambda} \frac{ds}{s} \left(\sigma(s)+\tilde{\omega} s +2\tilde{\omega}^2 \right) \right].
\end{eqnarray}
Finally, we notice that $\sigma(s=-i t)=\sigma_1(t)$ satisfies Eq.~(\ref{PV-family})
with $\nu=1$ supplemented by the boundary conditions Eqs.~(\ref{bc-s1-infty}) and (\ref{bc-s1-zero}). With help of this, we recover the statement of Theorem \ref{Th-4}
from Eq.~(\ref{eq:SnResult-2}).
\hfill $\square$
\begin{remark}\label{convergence}
  Note that the global condition Eq.~(\ref{global-T}) ensures that the expression in the square brackets in Eq.~(\ref{eq:SnResult}) exhibits ${\mathcal O}(\lambda^{-1})$ behavior as $\lambda \rightarrow \infty$. This guarantees that the external $\lambda$-integral in Eq.~(\ref{eq:SnResult}) converges for any $\tilde\omega \in (0, 1/2)$.
\hfill $\blacksquare$
\end{remark}

\begin{remark}
Notice that Eq.~(\ref{R0}) combined with Eqs.~(\ref{Rkf}), (\ref{Rk2}) and (\ref{eq:R1Omega-lead-b}) taken at $k=0$, motivates the following conjecture.
\hfill $\blacksquare$
\end{remark}

\begin{conjecture} \label{conj}
Let $0<\tilde{\omega}<1/2$ and let $\sigma(s)$ be the solution of the fifth Painlev\'e transcendent satisfying Eq.~(\ref{PV-eq}) and the boundary conditions
Eq.~(\ref{bc-inf}) -- (\ref{eq:gamma}). Then the following double integral relation holds
\begin{eqnarray}\label{eq:int_conjecture} \fl
    \int_0^{\infty}\frac{d\lambda}{2\pi} \lambda^{-2\tilde{\omega}^2} e^{i\tilde{\omega}\lambda} \left[ \exp\left(\int_0^{-i \lambda} \frac{ds}{s} \sigma(s)
     -i\tilde{\omega}\lambda +2\tilde{\omega}^2 \ln\lambda \right) -e^{-i\tilde{\omega}\pi}G_{\tilde{\omega}}  \right]
        \nonumber \\
        = \frac{1}{1-e^{2\pi i\tilde{\omega}}}
    -i \frac{G_{\tilde{\omega}}}{2\pi} e^{-i\pi(\tilde{\omega}+\tilde{\omega}^2)} \tilde{\omega}^{-1+2\tilde{\omega}^2} \Gamma(1-2\tilde{\omega}^2).
\end{eqnarray}
\hfill $\blacksquare$
\end{conjecture}

\begin{remark}
    To extend the proof of Theorem \ref{Th-4} for $\omega=\pi$, one would have to use the Theorem~1.12 of Ref.~\cite{CK-2015} instead of Theorems 1.5, 1.8 and 1.11 of the same paper. Since numerical calculations indicate that the power spectrum is continuous at $\omega=\pi$, we did not study this case analytically.
\hfill $\blacksquare$
\end{remark}

\subsection{Proof of Theorem \ref{Th-5}}

Below, the universal law $S_\infty(\omega)$ for the power spectrum will be studied in the vicinity of $\omega=0$. In the language of $S_N(\omega)$ this corresponds to performing a small-$\omega$ expansion after taking the limit $N\rightarrow \infty$. Equation (\ref{PS-exact}) will be the starting point of our analysis.
\newline\newline\noindent
{\it Preliminaries.}---Being interested in the small-$\tilde\omega$ behavior of the power spectrum Eq.~(\ref{PS-exact}), we observe that the functions ${\mathcal A}(\tilde{\omega})$ and ${\mathcal B}(\tilde{\omega})$, defined by Eqs.~(\ref{Aw-def}) and (\ref{Bw-def}), admit the expansions
\begin{eqnarray}
    {\mathcal A}(\tilde{\omega}) =\frac{1}{2 \pi^2 \tilde{\omega}} + \left(
        \frac{1}{6} - \frac{1+\gamma}{\pi^2} \right) \tilde{\omega}
     + {\mathcal O}(\tilde\omega^3),
\end{eqnarray}
\begin{eqnarray}
    {\mathcal B}(\tilde{\omega}) =\frac{1}{2} + \tilde{\omega}^2 \ln \tilde{\omega} + (\gamma-1)\, \tilde{\omega}^2  +  {\mathcal O}(\tilde\omega^4 \ln^2\tilde{\omega}),
\end{eqnarray}
so that the power spectrum, as $\tilde{\omega} \rightarrow 0$, can be written as
\begin{eqnarray} \label{S_inf_exp} \fl
     S_\infty(\omega) = \frac{1}{4 \pi^2 \tilde{\omega}} + \left(
        \frac{1}{12} - \frac{1}{\pi^2} \right) \tilde{\omega}
        + \frac{1}{2 \pi^2} \tilde\omega \ln \tilde\omega \nonumber\\
        + \left\{
            \frac{1}{2 \pi^2 \tilde\omega} + \left(
                \frac{1}{6} - \frac{1+\gamma}{\pi^2}
            \right) \tilde\omega + {\mathcal O} (\tilde\omega^3)
        \right\} \hat{\Lambda}(\tilde\omega) + {\mathcal O}(\tilde\omega^3 \ln^2\tilde\omega).
\end{eqnarray}
Here, $\hat{\Lambda}(\tilde\omega)$ denotes a small-$\omega$ expansion of the function
\begin{eqnarray} \label{eq:Lambda}\fl
    \Lambda(\tilde\omega) = {\rm Im} \int_{0}^{\infty} \frac{d\lambda}{2\pi} \, \lambda^{1-2\tilde{\omega}^2} \, e^{i\tilde{\omega} \lambda}
    \left[
        \exp \left(
                    - \int_{\lambda}^{\infty} \frac{dt}{t} \left( \sigma_1(t) - i \tilde{\omega} t + 2\tilde{\omega}^2\right)
            \right) -1
        \right],
\end{eqnarray}
such that
\begin{eqnarray} \label{approx-def}
    \Lambda(\tilde\omega) = \hat{\Lambda}(\tilde\omega) + {\mathcal O}(\tilde\omega^3),
\end{eqnarray}
see Eqs.~(\ref{PS-exact}) and (\ref{S_inf_exp}). Notice that convergence of the external $\lambda$-integral at infinity is ensured by the oscillating exponent $e^{i\tilde\omega \lambda}$.
\newline\newline
\noindent
{\it Small-$\tilde\omega$ ansatz for the fifth Painlev\'e transcendent.}---To proceed, we postulate the following ansatz for a small-$\omega$ expansion of the fifth Painlev\'e function
$\sigma_1(t)$:
\begin{eqnarray} \label{sigma-expan}
    \sigma_1(t) = \tilde{\omega} f_1(t) + \tilde{\omega}^2 f_2(t) + \tilde{\omega}^3 f_3(t)+ \cdots.
\end{eqnarray}
Here, the functions $f_k(t)$ with $k=1,2,\dots$ satisfy the equations
\begin{eqnarray} \label{fk-eqns}
    t^2 f_k^{\prime\prime\prime} + t f_k^{\prime\prime} + (t^2-4) f_k^{\prime} - t f_k(t) = F_k(t),
\end{eqnarray}
where
\begin{eqnarray}\label{FL-1}
    F_1(t) &=& 0, \\
    \label{FL-2}
    F_2(t) &=& 4 f_1(t) f_1^\prime
    - 6 t (f_1^\prime)^2, \\
    \label{FL-3}
    F_3(t) &=& 4 f_1(t) f_2^\prime + 4 f_1^\prime f_2(t) - 12 t f_1^\prime f_2^\prime,
\end{eqnarray}
etc. The above can easily be checked by substituting Eq.~(\ref{sigma-expan}) into Chazy form \cite{C-1911,C-2000}
\begin{eqnarray} \label{PV-chazy}
    t^2 \sigma_\nu^{\prime\prime\prime} + t \sigma_\nu^{\prime\prime} + 6 t (\sigma_\nu^{\prime})^2 - 4
    \sigma_\nu \sigma_\nu^{\prime} + (t^2 - 4\nu^2) \sigma_\nu^{\prime} - t \sigma_\nu = 0
\end{eqnarray}
of the Painlev\'e V equation Eq.~(\ref{PV-family}) taken at $\nu=1$. The boundary conditions are generated by Eqs.~(\ref{bc-s1-infty}) and (\ref{bc-s1-zero}):
\begin{eqnarray} \label{f1-bc}
    f_1(t)\rightarrow 0 \; {\rm as}\; t\rightarrow 0, \qquad f_1(t) = it + o(t) \;{\rm as\;} t\rightarrow +\infty,
\end{eqnarray}
\begin{eqnarray} \label{f2-bc}
    f_2(t)\rightarrow 0 \; {\rm as}\; t\rightarrow 0, \qquad f_2(t)  \rightarrow -2 \;{\rm as\;} t\rightarrow +\infty,
\end{eqnarray}
\begin{eqnarray} \label{f3-bc}
    f_3(t)\rightarrow 0 \; {\rm as}\; t\rightarrow 0, \qquad f_3(t) \rightarrow 0 \;{\rm as\;} t\rightarrow +\infty.
\end{eqnarray}
The third order differential equation (\ref{fk-eqns}) can be solved to bring
\begin{eqnarray}
    f_k(t) &=& \left( t -\frac{2}{t}\right) \left( c_{1,k} + \int_{0}^{t} \frac{dx}{x^3} F_k(x) \right)\nonumber\\
         &+& \frac{e^{it}}{t}
    \left(
      c_{2,k} + \int_{0}^{t} \frac{dx}{2 x^3} e^{-i x} (-x^2+ 2i x+2) \, F_k(x)
    \right) \nonumber\\
    &+& \frac{e^{-it}}{t}
    \left(
      c_{3,k} + \int_{0}^{t} \frac{dx}{2 x^3} e^{i x} (-x^2 - 2i x+2) \, F_k(x)
    \right).
\end{eqnarray}
This representation assumes that the integrals are convergent; integration constants have to be fixed by the boundary conditions Eqs.~(\ref{f1-bc}), (\ref{f2-bc}), (\ref{f3-bc}), etc. In particular, we derive:
\begin{eqnarray} \label{eq:f1}
    f_1(t) = i \frac{t^2 +2\cos t -2}{t},
\end{eqnarray}
\begin{eqnarray} \label{eq:f2}
    f_2(t) &=& -2 - \frac{6}{t^2} + \frac{2\pi}{t} - \pi t + 2 \cos t + 8 \frac{\cos t}{t^2}
    -\frac{2\pi}{t} \cos t \nonumber\\
     &-& 2 \frac{\cos(2 t)}{t^2} + 8 \gamma \frac{\sin t}{t} - 8 {\rm Ci}(t) \frac{\sin t}{t}
    + 8 \ln t \frac{\sin t}{t} \nonumber\\
     &-& 4 \frac{{\rm Si}(t)}{t} + 2 t {\rm Si}(t) + 4 \cos t \frac{{\rm Si}(t)}{t},
\end{eqnarray}
where
\begin{eqnarray} \label{euler}
    \gamma = \lim_{n\rightarrow \infty} \left(
        - \ln n + \sum_{k=1}^n \frac{1}{k}
    \right) \simeq 0.577216
\end{eqnarray}
is the Euler's constant, and
\begin{eqnarray} \label{eq:f3}
    f_3(t) &=& \left( \frac{2}{t}-t \right) \int_{t}^{\infty} \frac{dx}{x^3} F_3(x)
    - 2 \frac{\cos t}{t}  \int_{0}^{\infty} \frac{dx}{x^3} F_3(x)
        \nonumber\\
         &+& i {\rm Im} \left\{ \frac{e^{it}}{t}
        \int_{0}^{t} \frac{dx}{x^3} e^{-i x} (-x^2+ 2i x+2) \, F_3(x)
        \right\}.
\end{eqnarray}
Here, the function $F_3(t)$ is known explicitly from Eqs.~(\ref{FL-3}), (\ref{eq:f1}) and (\ref{eq:f2}). We notice that
$$
f_1(t) \in i\mathbb{R}, \quad f_2(t) \in \mathbb{R}, \quad f_3(t) \in i\mathbb{R},
$$
and
$$
F_2(t) \in \mathbb{R}, \quad  F_3(t) \in i\mathbb{R}.
$$
\newline
\noindent
{\it Representation of $\hat{\Lambda}(\tilde\omega)$ as a partial sum}.---Having determined the functions $f_1(t)$, $f_2(t)$ and $f_3(t)$, we now turn to the small-$\omega$ analysis of $\Lambda(\tilde{\omega})$ [Eq.~(\ref{eq:Lambda})]. Expanding the expression in square brackets in small $\omega$, we obtain:
\begin{eqnarray} \fl \label{bra-1}\
    \exp \left(
                    - \int_{\lambda}^{\infty} \frac{dt}{t} \left( \sigma_1(t) - i \tilde{\omega} t + 2\tilde{\omega}^2\right)
            \right) -1 &=&  -\tilde\omega {\mathcal G}_1(\lambda) \nonumber\\
                &-& \tilde\omega^2 {\mathcal G}_2(\lambda) - \tilde\omega^3 {\mathcal G}_3(\lambda) - \cdots.
\end{eqnarray}
The functions ${\mathcal G}_k(\lambda)$ can be evaluated explicitly in terms of integrals containing $f_k(\lambda)$ defined in Eq.~(\ref{sigma-expan}). For example,
\begin{eqnarray} \label{G1}
    {\mathcal G}_1 (\lambda) &=& {\mathcal F}_1 (\lambda),\\
    \label{G2}
    {\mathcal G}_2 (\lambda) &=&  - \frac{1}{2} {\mathcal F}_1^2(\lambda) + {\mathcal F}_2(\lambda),\\
    \label{G3}
    {\mathcal G}_3 (\lambda) &=&  {\mathcal F}_3 (\lambda) - {\mathcal F}_1(\lambda){\mathcal F}_2(\lambda) + \frac{1}{6} {\mathcal F}_1^3(\lambda).
\end{eqnarray}
Here,
\begin{eqnarray} \label{F1-def}
    {\mathcal F}_1(\lambda) = \int_\lambda^\infty \frac{dt}{t} (f_1(t) - i t) = -2 i \left(
        \frac{1-\cos\lambda}{\lambda} + \frac{\pi}{2} - {\rm Si}(\lambda)
    \right),
\end{eqnarray}
\begin{eqnarray} \label{F2-def}
    {\mathcal F}_2(\lambda) = \int_\lambda^\infty \frac{dt}{t} (f_2(t) +2),
\end{eqnarray}
and
\begin{eqnarray} \label{F3-def}
    {\mathcal F}_3(\lambda) = \int_\lambda^\infty \frac{dt}{t} \, f_3(t).
\end{eqnarray}
Notice that
$$
  {\mathcal F}_1(\lambda) \in i \mathbb{R},\quad {\mathcal F}_2(\lambda) \in \mathbb{R},\quad {\mathcal F}_3(\lambda) \in i\mathbb{R}
$$
and, hence,
$$
  {\mathcal G}_1(\lambda) \in i \mathbb{R},\quad {\mathcal G}_2(\lambda) \in \mathbb{R},\quad {\mathcal G}_3(\lambda) \in i\mathbb{R}.
$$
Substituting Eq.~(\ref{bra-1}) into Eq.~(\ref{eq:Lambda}), we split $\Lambda (\tilde\omega)$ into a partial sum
\begin{eqnarray} \label{L-partial}
    \Lambda (\tilde\omega) = \Lambda_1 (\tilde\omega) + \Lambda_2 (\tilde\omega) + \Lambda_3 (\tilde\omega)+ \cdots,
\end{eqnarray}
where
\begin{eqnarray}\label{Lam-K}
    \Lambda_k (\tilde\omega) = -\tilde{\omega}^k {\rm Im} \int_{0}^{\infty} \frac{d\lambda}{2\pi} \, \lambda^{1-2\tilde{\omega}^2} \, e^{i\tilde{\omega} \lambda} \mathcal{G}_k(\lambda).
\end{eqnarray}
A small-$\tilde\omega$ expansion of $\Lambda_k (\tilde\omega)$ is of our immediate interest.
\noindent\newline\newline
{\it Calculation of $\hat{\Lambda}_1(\tilde\omega)$.}---Equations (\ref{Lam-K}), (\ref{G1}) and (\ref{F1-def}) yield
\begin{eqnarray}
    \Lambda_1 (\tilde\omega) = 2 \tilde{\omega}   \int_{0}^{\infty} \frac{d\lambda}{2\pi} \, \lambda^{1-2\tilde{\omega}^2} \, \cos(\tilde{\omega} \lambda)
    \left(
    \frac{1-\cos\lambda}{\lambda} + \frac{\pi}{2} - {\rm Si}(\lambda)
    \right).
\end{eqnarray}
Performing the integral, we obtain:
\begin{eqnarray} \fl
    \Lambda_1 (\tilde\omega) = \frac{1}{\pi}\tilde{\omega}^3 \Gamma(-2 \tilde{\omega}^2) \sin(\pi \tilde\omega^2) \Bigg\{
        (1-\tilde\omega)^{2\tilde\omega^2-1} + (1+\tilde\omega)^{2\tilde\omega^2-1} - 2 \tilde\omega^{2\tilde\omega^2-1}
    \nonumber\\
        -  \frac{1-2\tilde\omega^2}{1-\tilde\omega^2} \, _3F_2\left(1-\tilde\omega^2,1-\tilde\omega^2,\frac{3}{2}-\tilde\omega^2;\frac{1}{2},2-\tilde\omega^2;\tilde\omega^2\right)
        \Bigg\}.
\end{eqnarray}
Its small-$\tilde\omega$ expansion $\Lambda_1 (\tilde\omega) = \tilde\omega^2 + {\mathcal O}(\tilde\omega^3)$ brings
\begin{eqnarray} \label{L1-hat}
    \hat{\Lambda}_1 (\tilde\omega) = \tilde\omega^2,
\end{eqnarray}
see Eq.~(\ref{approx-def}) for the definition of $\hat{\Lambda}(\omega)$.
\newline\newline\noindent
{\it Estimate of ${\Lambda}_k(\tilde\omega)$.}---To treat $\Lambda_k(\tilde\omega)$ for $k \ge 2$, we split it into two parts
\begin{eqnarray} \label{A-B}
    \Lambda_k(\tilde\omega) =  A_k(\tilde\omega, T) +  B_k(\tilde\omega, T),
\end{eqnarray}
where
\begin{eqnarray} \label{a-k}
    A_k(\tilde\omega, T) &=& - \tilde\omega^k {\rm Im} \int_{0}^{T} \frac{d\lambda}{2\pi} \, \lambda^{1-2\tilde{\omega}^2} \, e^{i\tilde{\omega} \lambda} \mathcal{G}_k(\lambda),\\
    \label{b-k}
    B_k(\tilde\omega, T) &=& - \tilde\omega^k {\rm Im}\int_{T}^{\infty} \frac{d\lambda}{2\pi} \, \lambda^{1-2\tilde{\omega}^2} \, e^{i\tilde{\omega} \lambda} \mathcal{G}_k(\lambda).
\end{eqnarray}
Here, $T$ is an arbitrary positive number to be taken to infinity in the end.

Since a small-$\tilde\omega$ expansion of $A_k(\tilde\omega, T)$ is well justified for any finite $T$, see e.g. Eq.~(\ref{A2-exp}) below, we conclude that
\begin{eqnarray} \label{A3-w3}
     A_k(\tilde\omega, T) = {\mathcal O}(\tilde\omega^k).
\end{eqnarray}
To estimate $B_k(\tilde\omega, T)$, we refer to Remark \ref{convergence} which implies that ${\mathcal G}_k(\lambda) = {\mathcal O}(\lambda^{-1})$ as $\lambda \rightarrow \infty$. Replacing ${\mathcal G}_k(\lambda)$ with $1/\lambda$ in Eq.~(\ref{b-k}), we perform the integration by parts twice in the resulting integral
\begin{eqnarray} \fl \label{one-more-f}
    \int_{T}^{\infty} \frac{d\lambda}{2\pi} \, \frac{e^{i\tilde{\omega} \lambda}}{\lambda^{2\tilde{\omega}^2}} =
    -\frac{e^{i\tilde\omega T}}{2 i\pi\tilde\omega} T^{-2\tilde\omega^2} + \frac{1}{\pi} e^{i\tilde\omega T} T^{-1-2\tilde\omega^2}
       - 2 (1+2\tilde\omega^2)  \int_{T}^{\infty} \frac{d\lambda}{2\pi} \frac{e^{i\tilde\omega \lambda}}{\lambda^{2+2\tilde\omega^2}}
\end{eqnarray}
to conclude that it is of order ${\mathcal O}(\tilde\omega^{-1})$. This entails
\begin{eqnarray} \label{Bk-wk}
     B_k(\tilde\omega, T) = {\mathcal O}(\tilde\omega^{k-1}).
\end{eqnarray}
Since we are interested in calculating $\Lambda(\tilde\omega)$ up to the terms ${\mathcal O}(\tilde\omega^3)$, see Eq.~(\ref{approx-def}), we need to consider
$A_k(\tilde\omega, T)$ and $B_{k+1}(\tilde\omega, T)$ for $k \le 2$ only.
\noindent\newline\newline
{\it Calculation of $\hat{\Lambda}_2(\tilde\omega)$.}---
A small-$\tilde\omega$ expansion of $A_2(\tilde\omega, T)$ brings
\begin{eqnarray} \fl\label{A2-exp}
    A_2(\tilde\omega, T) = - \tilde\omega^2 {\rm Im} \int_{0}^{T} \frac{d\lambda}{2\pi} \, \lambda
        \left(
            1 + i \tilde\omega \lambda - 2 \tilde\omega^2 \ln \lambda  -\frac{1}{2} \tilde \omega^2 \lambda^2
            + {\mathcal O}(\tilde\omega^3)
        \right)
     \mathcal{G}_2(\lambda).
\end{eqnarray}
Since $\mathcal{G}_2(\lambda) \in \mathbb{R}$, we even conclude that
\begin{eqnarray} \label{A2-w3}
         A_2(\tilde\omega, T) &=& {\mathcal O}(\tilde\omega^3).
\end{eqnarray}
For this reason, $A_2(\tilde\omega, T)$ does not contribute to $\hat{\Lambda}_2(\tilde\omega)$.

Evaluation of $B_2(\tilde\omega, T)$, given by
\begin{eqnarray} \label{B2-def}
    B_2(\tilde\omega, T) &=& - \tilde\omega^2 \int_{T}^{\infty} \frac{d\lambda}{2\pi} \, \lambda^{1-2\tilde{\omega}^2} \, \sin(\tilde{\omega} \lambda) \mathcal{G}_2(\lambda),
\end{eqnarray}
is more involved. A simplification comes from the fact that, at some point, we shall let $T$ tend to infinity. For this reason, it suffices to consider a large-$\lambda$ expansion of $\mathcal{G}_2(\lambda)$ in the integrand. Straightforward calculations bring
\begin{eqnarray}\label{F12-asym}
    \mathcal{F}_1(\lambda) = -\frac{2i}{\lambda} - 2i \frac{\sin\lambda}{\lambda^2}  + \mathcal{O}\left(\frac{\cos\lambda}{\lambda^3}\right), \\
    \label{F2-asym}
    \mathcal{F}_2 (\lambda) = -\frac{6}{\lambda^2} + 8 \frac{\cos\lambda \ln \lambda}{\lambda^2} + 2 (4\gamma-1) \frac{\cos\lambda}{\lambda^2}  +
    \mathcal{O}\left( \frac{\ln\lambda}{\lambda^3}\right).
\end{eqnarray}
Equation (\ref{F12-asym}) is furnished by the large-$\lambda$ expansion of Eq.~(\ref{F1-def}). To derive Eq.~(\ref{F2-asym}), we first calculated the integral Eq.~(\ref{F2-def}) replacing an integrand therein with its large-$t$ asymptotics, and then expanded the resulting expression in parameter $\lambda \rightarrow \infty$. By virtue of Eq.~(\ref{G2}), this yields
\begin{eqnarray} \label{G2-exp}
    \mathcal{G}_2(\lambda) = -\frac{4}{\lambda^2} + 8 \frac{\cos\lambda \ln \lambda}{\lambda^2} + 2 (4\gamma-1) \frac{\cos\lambda}{\lambda^2}  +
    \mathcal{O}\left( \frac{\ln\lambda}{\lambda^3}\right).
\end{eqnarray}
The expansion Eq.~(\ref{G2-exp}), being substituted into Eq.~(\ref{B2-def}), generates two families of integrals:
\begin{eqnarray} \label{I-j-def}
    {\mathcal I}_j(\tilde\omega, T) &=& \int_{T}^{\infty} \frac{d\lambda}{2\pi} \, \frac{\sin[(\tilde{\omega} +j) \lambda]}{\lambda^{1+2\tilde{\omega}^2}}
\end{eqnarray}
with $j=0, \pm 1$ and
\begin{eqnarray} \label{K-j-def}
    {\mathcal K}_j(\tilde\omega, T) &=& \int_{T}^{\infty} \frac{d\lambda}{2\pi} \,\ln \lambda \, \frac{\sin[(\tilde{\omega} +j) \lambda]}{\lambda^{1+2\tilde{\omega}^2}}
\end{eqnarray}
with $j=\pm 1$, such that
\begin{eqnarray} \label{B2-sum}
    B_2(\tilde\omega, T) &=& \tilde\omega^2 \Big\{ 4 {\mathcal I}_0(\tilde\omega, T) - (4 \gamma-1)[{\mathcal I}_{-1}(\tilde\omega, T) + {\mathcal I}_1(\tilde\omega, T)] \nonumber\\
        &-& 4
    [{\mathcal K}_{-1}(\tilde\omega, T) + {\mathcal K}_1(\tilde\omega, T)]\Big\}.
\end{eqnarray}
To determine a small-$\tilde\omega$ expansion of $B_2(\tilde\omega, T)$, we shall further concentrate on small-$\tilde\omega$ expansions of its constituents, ${\mathcal I}_{0}(\tilde\omega, T)$, ${\mathcal I}_{\pm 1}(\tilde\omega, T)$ and ${\mathcal K}_{\pm1}(\tilde\omega, T)$.
\noindent\newline\newline
{\it (a)}.---The function ${\mathcal I}_{0}(\tilde\omega, T)$ can be evaluated exactly,
\begin{eqnarray} \fl
     \qquad \qquad {\mathcal I}_0(\tilde\omega, T) = \frac{1}{4\pi} \sin(\pi \tilde\omega^2)\,  \tilde\omega^{2\tilde\omega^2-2}\Gamma(1-2 \tilde\omega^2) \nonumber\\
     - \frac{1}{2\pi} T^{1-2\tilde\omega^2} \frac{\tilde\omega}{1- 2 \tilde\omega^2} \,
     {}_1F_2\left(\frac{1}{2}-\tilde\omega^2;\frac{3}{2}, \frac{3}{2}-\tilde\omega^2; - \frac{T^2}{4}\tilde\omega^2\right).
\end{eqnarray}
Expanding this result around $\tilde\omega=0$ we derive
\begin{eqnarray} \label{I-0}
     {\mathcal I}_0(\tilde\omega, T) =  \frac{1}{4} + {\mathcal O}(\tilde\omega).
\end{eqnarray}
\noindent\newline
{\it (b)}.---To analyze a small-$\tilde\omega$ expansion
\begin{eqnarray}
  {\mathcal I}_{j\neq 0}(\tilde\omega, T)=\alpha_0(j,T) + \tilde\omega\, \alpha_1(j,T) + {\mathcal O}(\tilde\omega^2),
\end{eqnarray}
we proceed in two steps. First, we determine the coefficient $\alpha_0(j,T)$ directly from Eq.~(\ref{I-j-def})
\begin{eqnarray} \label{I-j-def-0}
    \alpha_0(j,T) = {\mathcal I}_{j\neq 0}(0, T) &=&  \int_{T}^{\infty} \frac{d\lambda}{2\pi} \, \frac{\sin(j \lambda)}{\lambda}
    < \infty, \quad \forall\, T>0,
\end{eqnarray}
to deduce the relation $(j \neq 0)$
\begin{eqnarray}
    \alpha_0(-j,T) = -\alpha_0(j,T).
\end{eqnarray}
Second, to determine a linear term of a small-$\tilde\omega$ expansion of ${\mathcal I}_{j\neq 0}(\tilde\omega, T)$, we perform integration by parts in Eq.~(\ref{I-j-def}) to derive the representation
\begin{eqnarray} \fl \label{I-j-def-alt}
    {\mathcal I}_{j\neq 0}(\tilde\omega, T) = T^{-1-2\tilde\omega^2} \frac{\cos[(\tilde\omega +j)T]}{2 \pi (\tilde\omega +j)} -
     \frac{1+2 \tilde\omega^2}{\tilde\omega +j} \int_{T}^{\infty} \frac{d\lambda}{2 \pi} \frac{\cos[(\tilde\omega +j)\lambda]}{\lambda^{2 + 2 \tilde\omega^2}}
\end{eqnarray}
whose integral term possesses a better convergence when $\tilde\omega$ approaches zero, as compared to the one given by Eq.~(\ref{I-j-def}). Differentiating Eq.~(\ref{I-j-def-alt}) with respect to $\tilde\omega$ and setting $\tilde\omega=0$ we derive:
\begin{eqnarray} \fl \label{I-j-der}
    \alpha_1(j,T) =  \frac{d \, {\mathcal I}_{j \neq 0}(\tilde\omega, T)}{d\tilde{\omega}}\Bigg|_{\tilde\omega =0} &=& -\frac{1}{2\pi j} \left(
        \sin(j T) + \frac{\cos(jT)}{j T}
    \right)
        + \frac{1}{j} \int_{T}^{\infty} \frac{d\lambda}{2\pi}\frac{\sin(j\lambda)}{\lambda} \nonumber\\
        &+&
            \frac{1}{j^2} \int_{T}^{\infty} \frac{d\lambda}{2\pi}\frac{\cos(j\lambda)}{\lambda^2} < \infty, \quad \forall\, T>0.
\end{eqnarray}
This implies the relation $(j \neq 0)$
\begin{eqnarray}
    \alpha_1(-j,T) = \alpha_1(j,T).
\end{eqnarray}
As a consequence, we conclude that
\begin{eqnarray} \label{I-sum}
    {\mathcal I}_{-1}(\tilde\omega, T) + {\mathcal I}_{1}(\tilde\omega, T) =  {\mathcal O}(\tilde\omega).
\end{eqnarray}
(It is this particular combination that appears in Eq.~(\ref{B2-sum}).)
\noindent\newline\newline
{\it (c)}.---To examine a small-$\tilde \omega$ expansion
\begin{eqnarray}
{\mathcal K}_{j\neq 0}(\tilde\omega, T)= \kappa_0(j,T)+ \tilde\omega \,\kappa_1(j,T) +{\mathcal O}(\tilde\omega^2),
\end{eqnarray}
we follow the same strategy. First, we determine the coefficient $\kappa_0(j, T)$ directly from Eq.~(\ref{K-j-def})
\begin{eqnarray} \label{K-j-def-0} \fl
    \kappa_0(j, T)= {\mathcal K}_{j\neq 0}(0, T) =
    \int_{T}^{\infty} \frac{d\lambda}{2\pi} \,\ln\lambda \frac{\sin(j \lambda)}{\lambda}
    < \infty, \quad \forall\, T>0,
\end{eqnarray}
to observe the relation $(j \neq 0)$
\begin{eqnarray}
    \kappa_0(-j,T) = -\kappa_0(j,T).
\end{eqnarray}
Second, to examine a linear term of a small-$\tilde\omega$ expansion of ${\mathcal K}_{j\neq 0}(\tilde\omega, T)$, we
perform integration by parts in Eq.~(\ref{K-j-def}) in order to improve integral's convergence:
\begin{eqnarray} \fl \label{K-j-parts}
    {\mathcal K}_{j \neq 0}(\tilde\omega,T) = \frac{\ln T}{2\pi (\tilde\omega+j)} \frac{\cos[(\tilde\omega +j)T]}{T^{1+2\tilde\omega^2}}
    + \frac{1}{\tilde\omega+j}
    \int_{T}^{\infty} \frac{d\lambda}{2\pi}  \frac{\cos[(\tilde\omega+j)\lambda]}{\lambda^{2+2\tilde\omega^2}} \nonumber\\
        - \frac{1+2\tilde\omega^2}{\tilde\omega+j} \int_{T}^{\infty} \frac{d\lambda}{2\pi}  \frac{\cos[(\tilde\omega+j)\lambda]}{\lambda^{2+2\tilde\omega^2}} \ln\lambda.
\end{eqnarray}
Differentiating Eq.~(\ref{K-j-parts}) with respect to $\tilde\omega$ and setting $\tilde\omega=0$, we obtain:
\begin{eqnarray}  \label{K-j-der} \fl
        \kappa_1(j, T)= \frac{d \, {\mathcal K}_{j \neq 0}(\tilde\omega, T)}{d\tilde{\omega}}\Bigg|_{\tilde\omega =0} =
        - \frac{\sin(jT)}{2\pi j} \ln T - \frac{\cos(jT)}{2\pi j^2 T} \ln T \nonumber\\
        - \frac{1}{j} \int_{T}^{\infty} \frac{d\lambda}{2\pi} \frac{\sin(j\lambda)}{\lambda} \left( 1 - \ln\lambda\right)
        \nonumber\\
        - \frac{1}{j^2} \int_{T}^{\infty} \frac{d\lambda}{2\pi} \frac{\cos(j\lambda)}{\lambda^2} \left( 1 -  \ln\lambda\right)
        < \infty, \quad \forall\, T>0.
\end{eqnarray}
This implies the relation $(j \neq 0)$
\begin{eqnarray}
    \kappa_1(-j,T) = \kappa_1(j,T).
\end{eqnarray}
As a consequence, we conclude that
\begin{eqnarray} \label{K-sum}
    {\mathcal K}_{-1}(\tilde\omega, T) + {\mathcal K}_{1}(\tilde\omega, T) = {\mathcal O}(\tilde\omega).
\end{eqnarray}
(Again, it is this particular combination that appears in Eq.~(\ref{B2-sum}).)

Collecting the results Eqs.~(\ref{A2-w3}), (\ref{B2-sum}), (\ref{I-0}), (\ref{I-sum}), and (\ref{K-sum}), we observe that $\Lambda_2(\tilde \omega) = \tilde\omega^2 + {\mathcal O}(\tilde\omega^3)$; hence
\begin{eqnarray}
    \hat{\Lambda}_2(\tilde \omega) = \hat{\Lambda}_1(\tilde \omega) =\tilde\omega^2,
\end{eqnarray}
see Eqs.~(\ref{approx-def}) and (\ref{L1-hat}).
\noindent\newline\newline
{\it Calculation of $\hat\Lambda_3(\tilde\omega)$}.---Since $A_3(\tilde\omega,T)= {\mathcal O}(\tilde\omega^3)$, see Eq.~(\ref{A3-w3}), we need to deal with $B_3(\tilde\omega,T)$ only:
\begin{eqnarray} \label{b-3}
    B_3(\tilde\omega, T) = - \tilde\omega^3 {\rm Im}\int_{T}^{\infty} \frac{d\lambda}{2\pi} \, \lambda^{1-2\tilde{\omega}^2} \, e^{i\tilde{\omega} \lambda} \mathcal{G}_3(\lambda),
\end{eqnarray}
large-$\lambda$ asymptotics of $\mathcal{G}_3(\lambda)$ defined by Eq.~(\ref{G3}) are required.

To proceed, we need to complement the expansions Eqs.~(\ref{F12-asym}) and (\ref{F2-asym}) with the one for $\mathcal{F}_3$ defined by Eq.~(\ref{F3-def}). To this end we, first, employ Eqs.~(\ref{FL-3}), (\ref{eq:f1}), (\ref{eq:f2}) to determine a large-$t$ behavior of $F_3(t)$,
\begin{eqnarray} \label{F3-as}\fl
    \qquad \frac{F_3(t)}{t^3} = i \left\{
        \frac{a_1}{t^3} + a_2 \frac{\cos t}{t^3} + a_3 \frac{\cos t \ln t}{t^3} + {\mathcal O}\left( \frac{\sin(\star\, t)\ln t}{t^4}\right)
    \right\},
\end{eqnarray}
where $a_1, a_2$ and $a_3$ are real coefficients whose explicit values are not required for our analysis; $\sin(\star\, t)$ stands to denote $\sin t$ and $\sin(2t)$, both of which are present in the remainder term. This expansion combined with Eq.~(\ref{eq:f3}) implies the following large-$t$ behavior of $f_3(t)$:
\begin{eqnarray} \label{f3-small-as} \fl
    \qquad f_3(t) = i \left\{
        a_1^\prime \frac{1}{t} + a_2^\prime \frac{\sin t}{t} + a_3^\prime \frac{\cos t}{t} + a_4^\prime \frac{\cos t}{t} \ln t
        +  a_5^\prime \frac{\cos t}{t} \ln^2 t
    \right\} \nonumber\\ + \mathcal{O} \left( \frac{\sin t \ln t}{t^2} \right).
\end{eqnarray}
Here, the coefficients $a_j^\prime \in \mathbb{R}$ are real.

Now, a large-$\lambda$ behavior of $\mathcal{F}_3(\lambda)$ can be read off from Eq.~(\ref{F3-def}):
\begin{eqnarray} \label{F3-asym} \fl
\qquad \qquad
    \mathcal{F}_3(\lambda) = i \left\{
        a_1^{\prime\prime} \frac{1}{\lambda} + a_2^{\prime\prime} \frac{\sin \lambda}{\lambda^2} + a_3^{\prime\prime} \frac{\cos \lambda}{\lambda^2}
        + a_4^{\prime\prime} \frac{\sin \lambda}{\lambda^2} \ln \lambda + a_5^{\prime\prime} \frac{\sin \lambda}{\lambda^2} \ln^2 \lambda
    \right\} \nonumber\\
    \qquad \qquad + \mathcal{O} \left(
        \frac{\cos \lambda}{\lambda^3} \ln^3 \lambda
    \right),
\end{eqnarray}
where the coefficients $a_j^{\prime\prime} \in \mathbb{R}$ are real, again. Inspection of Eqs.~(\ref{G3}), (\ref{F12-asym}), (\ref{F2-asym}) and (\ref{F3-asym}) shows that a large-$\lambda$ behavior of ${\mathcal G}_3(\lambda)$ coincides with that of
$\mathcal{F}_3(\lambda)$.

Having determined a large-$\lambda$ asymptotics of ${\mathcal G}_3(\lambda)$, we turn to the analysis of the function $B_3(\tilde\omega,T)$ as $\tilde\omega \rightarrow 0$. Since ${\mathcal G}_3(\lambda) \in i \mathbb{R}$, a substitution of Eq.~(\ref{F3-asym}) into Eq.~(\ref{b-3}) generates several integrals (see below), whose small-$\tilde\omega$ behavior should be studied in order to figure out if $B_3(\tilde\omega,T)$ contributes to $\hat\Lambda_3(\tilde\omega)$ as defined by Eqs.~(\ref{approx-def}), (\ref{L-partial}) and (\ref{A-B}). This knowledge is required to complete calculation of the small-$\omega$ expansion of the power spectrum $S_\infty(\omega)$, see Eq.~(\ref{S_inf_exp}).
\noindent\newline\newline
{\it (a)}.---The first integral, originating from the $a_1^{\prime\prime}$ term in Eq.~(\ref{F3-asym}), admits a small-$\tilde\omega$ expansion
\begin{eqnarray}
    B_{3,1}(\tilde\omega,T) = \int_{T}^{\infty} \frac{d\lambda}{2\pi} \,  \,\frac{\cos(\tilde{\omega} \lambda)}{\lambda^{2\tilde{\omega}^2}}
    ={\mathcal O}(\tilde\omega^0).
\end{eqnarray}
This result is obtained from the real part of the r.h.s.~of Eq.~(\ref{one-more-f}) evaluated at $\tilde\omega=0$.
Hence, due to Eq.~(\ref{b-3}), the contribution of $B_{3,1}(\tilde\omega,T)$ to $B_3(\tilde\omega,T)$ is of order ${\mathcal O}(\tilde\omega^3)$.
\noindent\newline\newline
{\it (b)}.---The second integral, originating from the $a_2^{\prime\prime}$ term in Eq.~(\ref{F3-asym}), reads
\begin{eqnarray}
    B_{3,2}(\tilde\omega,T) =\int_{T}^{\infty} \frac{d\lambda}{2\pi} \,\sin \lambda\, \frac{\cos(\tilde{\omega} \lambda)}{\lambda^{1+2\tilde\omega^2}}
    ={\mathcal O}(\tilde\omega^0)
\end{eqnarray}
as can be seen by setting $\tilde\omega=0$ directly in the integrand.
\noindent\newline\newline
{\it (c)}.---All other integrals generated by the remaining terms in Eq.~(\ref{F3-asym}) can be treated analogously.

As a consequence, we conclude that $B_3(\tilde\omega,T)$ is of order ${\mathcal O}(\tilde\omega^3)$. Taken together with Eqs.~(\ref{A-B}) and (\ref{A3-w3}), this implies that $\Lambda_3(\tilde\omega) = {\mathcal O}(\tilde\omega^3)$ so that
\begin{eqnarray} \label{lambda-finite}
    \hat\Lambda(\tilde\omega) = 2 \tilde\omega^2 + {\mathcal O}(\tilde\omega^3).
\end{eqnarray}
Substituting Eq.~(\ref{lambda-finite}) into Eq.~(\ref{S_inf_exp}), we derive the sought small-$\tilde\omega$ expansion of the power spectrum $S_\infty(\omega)$ as
stated in Theorem \ref{Th-5}.
\hfill $\square$

\section*{Acknowledgments}
Roman Riser thanks Tom Claeys for insightful discussions and kind hospitality at the Universit\'e catholique de Louvain.
This work was supported by the Israel Science Foundation through the Grants No.~648/18 (E.K. and R.R.) and No.~2040/17 (R.R.). Support from the
Simons Center for Geometry and Physics, Stony Brook University, where a part of this work was completed, is gratefully acknowledged (E.K).

\newpage
\renewcommand{\appendixpagename}{\normalsize{Appendices}}
\addappheadtotoc
\appendixpage
\renewcommand{\thesection}{\Alph{section}}
\renewcommand{\theequation}{\thesection.\arabic{equation}}
\setcounter{section}{0}

\section{Boundary conditions for Painlev\'e VI function $\tilde{\sigma}_N (t; \zeta)$ as $t \rightarrow \infty$} \label{A-1}
To derive the $t \rightarrow \infty$ boundary condition for $\tilde{\sigma}_N (t; \zeta)$ satisfying Eq.~(\ref{pvi}) of the Theorem \ref{Th-3}, we make use of Eqs.~(\ref{phintheta}) and (\ref{phin}) to observe
the relation
\begin{eqnarray}\label{sn-bc}
    \tilde{\sigma}_N(t;\zeta) = -t - 2 \frac{d}{d\varphi} \ln \Phi_N(\varphi;\zeta)\Big|_{\varphi= 2\arctan(1/t)}
\end{eqnarray}
which holds true for $t >0$ and $0\le \varphi < \pi/2$. Since $\varphi \rightarrow 0$ as $t \rightarrow \infty$, we shall consider a small-$\varphi$
expansion of the generating function $\Phi_N(\varphi;\zeta)$
\begin{eqnarray} \label{ps-tcue-2A}
    \Phi_N(\varphi;\zeta) &=&
    \prod_{j=1}^N  \left( \int_0^{2\pi} - \zeta \int_0^\varphi \right) \frac{d\theta_j}{2\pi}
    P_N(\theta_1,\dots,\theta_N) \nonumber\\
    &=& 1+ \sum_{\ell=1}^N \frac{(-\zeta)^\ell}{\ell!} \left( \prod_{j=1}^\ell \int_{0}^{\varphi} \frac{d\theta_j}{2\pi} \right)
    R_{\ell,N} (\theta_1,\dots,\theta_\ell),
\end{eqnarray}
where the JPDF $P_N(\theta_1,\dots,\theta_N)$ is that of ${\rm TCUE}_N$ [Eq.~(\ref{T-CUE})], and $R_{\ell,N} (\theta_1,\dots,\theta_\ell)$ stands for
the $\ell$-th correlation function in ${\rm TCUE}_N$. Due to the Lemma \ref{correlation-f}, these admit a determinantal respresentation
\begin{eqnarray}
    R_{\ell,N} (\theta_1,\dots,\theta_\ell) = \frac{1}{N+1} {\det}_{1\le i,j \le \ell+1} \left[
        {\mathcal S}_{N+1} (\theta_i -\theta_j)
    \right] \Big|_{\theta_{\ell+1}=0},
\end{eqnarray}
where ${\mathcal S}_{N+1}(\theta)$ is the ${\rm CUE}_{N+1}$ sine-kernel:
$$
    {\mathcal S}_{N+1}(\theta) = \frac{\sin[(N+1)\theta/2]}{\sin(\theta/2)}.
$$
For one,
\begin{eqnarray}
     R_{1,N} (\theta) = \frac{1}{N+1} \det\left(
                                              \begin{array}{ll}
                                                {\mathcal S}_{N+1}(0) & {\mathcal S}_{N+1}(\theta) \\
                                                {\mathcal S}_{N+1}(\theta) & {\mathcal S}_{N+1}(0) \\
                                              \end{array}
                                            \right),
\end{eqnarray}
\begin{eqnarray} \fl
    R_{2,N} (\theta_1,\theta_2) = \frac{1}{N+1} \det
        \left(
          \begin{array}{lll}
            {\mathcal S}_{N+1}(0) & {\mathcal S}_{N+1}(\theta_1-\theta_2) & {\mathcal S}_{N+1}(\theta_1) \\
            {\mathcal S}_{N+1}(\theta_1-\theta_2) & {\mathcal S}_{N+1}(0) & {\mathcal S}_{N+1}(\theta_2) \\
            {\mathcal S}_{N+1}(\theta_1) & {\mathcal S}_{N+1}(\theta_2) & {\mathcal S}_{N+1}(0) \\
          \end{array}
        \right),
\end{eqnarray}
etc.

A straightforward calculation produces a small-$\varphi$ expansion of $\Phi_N(\varphi;\zeta)$ whose several initial terms read:
\begin{eqnarray} \fl \label{ps-tcue-2AA}
    \Phi_N(\varphi;\zeta) = 1 - \frac{\zeta}{2\pi} \left(
        R_{1,N} (0) \, \varphi + \frac{1}{2!} R_{1,N}^\prime (0) \, \varphi^2 + \frac{1}{3!} R_{1,N}^{\prime\prime} (0) \, \varphi^3
    \right) \nonumber\\
     +  \frac{1}{2!} \left(\frac{\zeta}{2\pi}\right)^2 \left( R_{2,N} (0,0) \, \varphi^2 +
        \frac{1}{2} \left[
            R_{2,N}^{[0,1]}(0,0) +  R_{2,N}^{[1,0]}(0,0)
        \right] \varphi^3
    \right) \nonumber\\
    - \frac{1}{3!} \left(\frac{\zeta}{2\pi}\right)^3  R_{3,N} (0,0,0) \, \varphi^3 +
    o(\varphi^3).
    \nonumber\\
    {}
\end{eqnarray}
Only one, out of six, coefficients in the expansion is nontrivial,
$$
    R_{1,N}^{\prime\prime} (0) = \frac{N(N+1)(N+2)}{6}
$$
yielding
\begin{eqnarray} \label{ps-tcue-2AAA}
    \Phi_N(\varphi;\zeta) = 1 - \frac{N(N+1)(N+2)}{72\pi} \, \zeta \varphi^3 +
    o(\varphi^3).
\end{eqnarray}
By virtue of Eq.~(\ref{sn-bc}), the boundary condition for $\tilde{\sigma}_N (t; \zeta)$ as $t \rightarrow \infty$ reads
\begin{eqnarray}\label{1texp}
    \tilde{\sigma}_N (t; \zeta) = -t + \zeta\,\frac{N(N+1)(N+2)}{3\pi t^2} + {\mathcal O}(t^{-3}).
\end{eqnarray}
Further terms in the $1/t$-expansion Eq.~(\ref{1texp}) can be restored with the help of the Painlev\'e VI equation itself [Eq.~(\ref{pvi})]. Substituting the large-$t$
ansatz
\begin{eqnarray} \label{ansatz}
\tilde{\sigma}_N (t; \zeta) = -t + \sum_{j=2}^\infty \frac{\sigma_j(N,\zeta)}{t^j}
\end{eqnarray}
therein, we deduce:
\begin{eqnarray}
    \tilde{\sigma}_N (t; \zeta) = -t + \frac{\sigma_2(N,\zeta)}{t^2} +  \frac{\sigma_4(N,\zeta)}{t^4}  + \frac{\sigma_5(N,\zeta)}{t^5} + \mathcal{O}(t^{-6}),
\end{eqnarray}
where
\begin{eqnarray}
    \sigma_2(N,\zeta) &=& \zeta\,\frac{N(N+1)(N+2)}{3\pi}, \nonumber\\
    \sigma_4(N,\zeta) &=& - \frac{2N^2+4N+9}{15} \sigma_2(N,\zeta), \nonumber\\
    \sigma_5(N,\zeta) &=& \frac{\sigma_2^2(N,\zeta)}{3}.
\end{eqnarray}
\begin{remark} \label{Rem-A1}
Since the above procedure is capable of producing the expansion coefficients $\sigma_{j}(N,\zeta)$ of any finite order, it can also be utilized -- by virtue of Eq.~(\ref{phin}) -- to generate a small-$\varphi$ expansion of $\Phi_N(\varphi;\zeta)$ up to required accuracy.
\hfill $\blacksquare$
\end{remark}

\section{Generating function $\Phi_N(\varphi;\zeta)$ and discrete Painlev\'e V equations ($\rm dP_V$)} \label{B-1}
To avoid intricacies \cite{B-2010} of a numerical evaluation of the six Painlev\'e function $\tilde{\sigma}_N(t;\zeta)$ appearing in the generating function Eq.~(\ref{phin}), we opt for an alternative representation of $\Phi_N(\varphi;\zeta)$ in terms of discrete Painlev\'e V equations.

To proceed, we follow Ref.~\cite{FW-2003-arXiv} (see also Refs.~(\cite{FW-2003,FW-2005,FW-2006})), to observe that a sequence of $U(N)$ integrals
\begin{eqnarray}
    I_N(\varphi;\zeta) = \frac{1}{N!} \prod_{j=1}^N  \left(  \int_{-\pi}^{\pi} - \zeta \int_{\pi-\varphi}^{\pi} \right) \frac{d\theta_j}{2\pi}
      \, \prod_{1 \le i < j \le N} |e^{i\theta_i} - e^{i\theta_j}|^2  \nonumber\\
        \qquad \times \prod_{j=1}^N e^{\omega_2 \theta_j} \left| 1+ e^{i\theta_j} \right|^{2\omega_1} t^{-\mu} e^{-i\mu \theta_j}
    \left(1 + t e^{i\theta_j} \right)^{2\mu},
\end{eqnarray}
where $t=e^{i\varphi}$, satisfies a recurrence relation
\begin{eqnarray} \label{IN-rec}
    \frac{I_{N+1} I_{N-1}}{I_{N}^2} = 1 - r_N \bar{r}_N.
\end{eqnarray}
Here, $r_N$ and $\bar{r}_N$ are so-called reflection coefficients appearing in the Szeg\"o theory \cite{S-2003} of orthogonal polynomials
on the unit circle. Remarkably, there exists the $N$-recurrence for reflection coefficients $\{r_N,\bar{r}_N\}$ as specified in Proposition 4.4 in Ref.~\cite{FW-2003-arXiv}; a variation of their Proposition is given below.
\begin{proposition}
The $N$-recurrence for the reflection coefficients of polynomials orthogonal on the unit circle $|z|=1$ with respect to the weight
\begin{eqnarray} \fl
    w(z) = t^{-\mu} z^{-\mu-\omega_1 -i \omega_2} (1+z)^{2\omega_1} (1+t z)^{2\mu} \left\{
                                                                       \begin{array}{ll}
                                                                         1, & \hbox{$\theta \notin (\pi -\phi,\pi)$} \\
                                                                         1-\zeta, & \hbox{$\theta \in (\pi -\phi,\pi)$.}
                                                                       \end{array}
                                                                     \right.
\end{eqnarray}
is governed by two systems of coupled first order discrete Painlev\'e equations (${\rm dP_V}$). The first is
\begin{equation}
\left\{
\begin{array}{l}
  g_{N+1}g_N = t \displaystyle\frac{(f_N + N)(f_N + N + \mu)}{f_N(f_N - 2\omega_1)}, \\
    \\
      f_N + f_{N+1} =  2\omega_1 + \displaystyle\frac{N-1+\mu + \omega}{g_N-1} + \displaystyle\frac{t(N+\mu + \bar{\omega})}{g_N-t},
\end{array}
\right.
\end{equation}
subject to the initial conditions
\begin{eqnarray}\label{IC-1}
    g_1 = t \frac{\mu+\omega + (1+\mu+\bar{\omega})r_1}{\mu+\omega + (1+\mu+\bar{\omega}) t r_1}, \quad f_0=0, \quad r_1 = - \frac{w_{-1}}{w_0}.
\end{eqnarray}
The second system is
\begin{eqnarray}
\left\{
\begin{array}{l}
  \bar{g}_{N+1} \bar{g}_N = t^{-1} \displaystyle\frac{(\bar{f}_N + N)(\bar{f}_N + N + 2\omega_1)}{\bar{f}_N(\bar{f}_N - 2\mu)}, \\
    \\
      \bar{f}_N + \bar{f}_{N+1} = 2\mu + \displaystyle\frac{N+\mu + \omega}{\bar{g}_N-1} + \displaystyle\frac{(N-1 +\mu + \bar{\omega}) t^{-1}}{\bar{g}_N-t^{-1}},
\end{array}
\right.
\end{eqnarray}
subject to the initial conditions
\begin{eqnarray}\label{IC-2}
    \bar{g}_1 =  \frac{\mu+\bar{\omega} + (1+\mu+\omega) t^{-1} \bar{r}_1}{\mu+\bar{\omega} + (1+\mu+\omega) \bar{r}_1}, \quad \bar{f}_0=0,
    \quad \bar{r}_1 = - \frac{w_{1}}{w_0}.
\end{eqnarray}
Here, $\omega = \omega_1 + i\omega_2$ and $\bar{\omega}= \omega_1 - i\omega_2$. The coefficients $w_0, w_{\mp 1}$ in Eqs.~(\ref{IC-1}) and (\ref{IC-2}) are
\begin{eqnarray} \label{WL}
    w_\ell = \frac{1}{2 i \pi}\oint \frac{dz}{z^{\ell+1}}\, w(z).
\end{eqnarray}
The transformations relating the variables $\{g_N, \bar{g}_N\}$ to the reflection coefficients $\{r_N,\bar{r}_N\}$ read:
\begin{eqnarray}
    \frac{r_N}{r_{N-1}} = \frac{1-t^{-1}g_N}{g_N-1} \frac{N-1-\mu +\omega}{N+\mu+\bar{\omega}}
\end{eqnarray}
and
\begin{eqnarray}
    \frac{\bar{r}_N}{\bar{r}_{N-1}} = \frac{1-\bar{g}_N}{\bar{g}_N-t^{-1}} \frac{N-1-\mu +\bar{\omega}}{N+\mu+\omega},
\end{eqnarray}
respectively.
\newline
\end{proposition}
The Proposition yields a sought ${\rm dP_V}$ representation of the generating function $\Phi_N(\varphi;\zeta)$, see Eq.~(\ref{phintheta}). Indeed, setting $\omega=\omega_1 + i\omega_2 =1$ and $\mu=0$, one observes the relation
$$
    \Phi_N(\varphi;\zeta) = \frac{I_N(\varphi;\zeta)}{N+1}
$$
so that Eq.~(\ref{IN-rec}) translates to
\begin{eqnarray} \label{FN-rec}
    \frac{\Phi_{N+1}\Phi_{N-1}}{\Phi_{N}^2} = \frac{(N+1)^2}{N(N+2)} \left(
        1 - r_N \bar{r}_N
    \right),
\end{eqnarray}
where the reflection coefficients $\{r_N, \bar{r}_N\}$ are determined by equations
\begin{eqnarray}
    \frac{r_N}{r_{N-1}} = \frac{1-t^{-1}g_N}{g_N-1} \frac{N}{N+1}
\end{eqnarray}
and
\begin{eqnarray}
    \frac{\bar{r}_N}{\bar{r}_{N-1}} = \frac{1-\bar{g}_N}{\bar{g}_N-t^{-1}} \frac{N}{N+1},
\end{eqnarray}
considered in conjunction with two systems of coupled first order discrete Painlev\'e equations (${\rm dP_V}$):
\begin{equation}
\left\{
\begin{array}{l}
  g_{N+1}g_N = t \displaystyle\frac{(f_N + N)^2}{f_N(f_N - 2)}, \\
    \\
      f_N + f_{N+1} =  2 + \displaystyle\frac{N}{g_N-1} + \displaystyle\frac{t(N+1)}{g_N-t}
\end{array}
\right.
\end{equation}
and
\begin{eqnarray}
\left\{
\begin{array}{l}
  \bar{g}_{N+1} \bar{g}_N = t^{-1} \displaystyle\frac{(\bar{f}_N + N)(\bar{f}_N + N + 2)}{\bar{f}_N^2}, \\
    \\
      \bar{f}_N + \bar{f}_{N+1} = \displaystyle\frac{N+1}{\bar{g}_N-1} + \displaystyle\frac{N}{t \bar{g}_N-1}.
\end{array}
\right.
\end{eqnarray}
The initial conditions read
\begin{eqnarray}
    \Phi_0 = 1, \quad \Phi_1 = 1 - \frac{\zeta}{2\pi} (\varphi - \sin\varphi),
\end{eqnarray}
\begin{eqnarray}\label{IC-1-ex}
    g_1 = t \frac{w_0 - 2 w_{-1}}{w_0 - 2 t w_{-1}}, \quad f_0=0
\end{eqnarray}
and
\begin{eqnarray}\label{IC-2-ex}
    \bar{g}_1 =  \frac{w_0 - 2 t^{-1} w_1}{w_0 - 2 w_1}, \quad \bar{f}_0=0,
\end{eqnarray}
respectively. By virtue of Eq.~(\ref{WL}), a set of parameters $\{w_0, w_{\pm 1}\}$ can be calculated explicitly:
\begin{eqnarray} \label{WIJ}
    w_0 &=& 2 - \frac{\zeta}{i \pi} \left( \frac{1-t^2}{2t} + \ln t  \right), \\
    \label{WIJ-2}
    w_{\pm 1} &=& 1 \mp \frac{\zeta}{i \pi} \left(
        \frac{1}{4} (t^{\pm 1}-1) (t^{\pm 1}-3) + \frac{1}{2} \ln (t^{\pm 1})
    \right).
\end{eqnarray}
Equations (\ref{FN-rec}) -- (\ref{WIJ-2}) provide the ${\rm dP_V}$ representation of the generating function $\Phi_N(\varphi;\zeta)$.
\begin{remark}
To avoid numerical $z$--differentiation of $\Phi_N(\varphi;1-z)$ appearing in the formula Eq.~(\ref{ps-tcue-1}), it is beneficial to produce a similar system of coupled recurrence equations for $(\partial/\partial z) \Phi_N(\varphi;1-z)$. Since the resulting recurrences are too cumbersome to state them here, we leave their (straightforward) derivation to the inquisitive reader.

\hfill $\blacksquare$
\end{remark}
\begin{remark}
Away from the endpoints $\varphi=0$ and $\varphi=2\pi$, the ${\rm dP_V}$ representation opens a way for effective numerical evaluation of both
$\Phi_N(\varphi;\zeta)$ and  $(\partial/\partial z) \Phi_N(\varphi;1-z)$ for finite $N$. Since the recurrence procedure tends to accumulate numerical errors, we have used quadruple precision numbers to achieve sufficient precision for very large $N$ (e.g., for $N = 10^4$, see Figs.~\ref{Figure_ps_overall} and \ref{Figure_ps_diff}).
\hfill $\blacksquare$
\end{remark}
\begin{remark}
In the vicinity of the endpoints $\varphi=0$ and $\varphi=2\pi$, numerical precision of the above recurrence procedure worsens drastically since the recurrence equations
start to exhibit a singular behavior. To circumvent this drawback at $\varphi=0$, we have used a small-$\varphi$ expansion of $\Phi_N(\varphi;\zeta)$ as described in the Remark
\ref{Rem-A1}. In the vicinity of $\varphi=2\pi$, the symmetry relation Eq.~(\ref{phin-sym}) combined with a small-$\varphi$ expansion makes the job.
\hfill $\blacksquare$
\end{remark}

\smallskip\smallskip\smallskip

\newpage

\section*{References}
\fancyhead{} \fancyhead[RE,LO]{References}
\fancyhead[LE,RO]{\thepage}

\end{document}